\documentclass[11pt]{article}

\usepackage{epsfig,epsf,fancybox}
\usepackage{amsmath}
\usepackage{mathrsfs}
\usepackage{amssymb}
\usepackage{graphicx}
\usepackage{color}
\usepackage{multirow}
\usepackage{paralist}
\usepackage{verbatim}
\usepackage{galois}
\usepackage{algorithm}
\usepackage{algorithmic}
\usepackage{boxedminipage}
\usepackage{booktabs}
\usepackage{accents}
\usepackage{stmaryrd}

\usepackage{subfig}

\usepackage{natbib}

\usepackage{url}% for url's in bib
\usepackage[colorlinks,linkcolor=magenta,citecolor=blue, pagebackref=true,backref=true]{hyperref}
\renewcommand*{\backrefalt}[4]{%
    \ifcase #1 \footnotesize{(Not cited.)}%
    \or        \footnotesize{(Cited on page~#2.)}%
    \else      \footnotesize{(Cited on pages~#2.)}%
    \fi}

\newtheorem{theorem}{Theorem}[section]
\newtheorem{corollary}[theorem]{Corollary}
\newtheorem{lemma}[theorem]{Lemma}
\newtheorem{proposition}[theorem]{Proposition}

\newtheorem{definition}[theorem]{Definition}
\newtheorem{example}{Example}[section]

\newtheorem{remark}[theorem]{Remark}

\numberwithin{equation}{section}

\newcommand{\sa}{\mathbf a}

\newcommand{\argmin}{\mathop{\rm argmin}}
\newcommand{\argmax}{\mathop{\rm argmax}}

\newcommand{\GCal}{\mathcal{G}}

\newcommand{\XCal}{\mathcal{X}}

\newcommand{\br}{\mathbb{R}}

\newcommand{\ba}{\begin{array}}
\newcommand{\ea}{\end{array}}

\newcommand{\PCal}{\mathcal{P}}

\newcommand{\EE}{{\mathbb{E}}}
\newcommand{\PP}{\mathbb{P}}

\newcommand{\NCal}{\mathcal{N}}

\newcommand{\one}{\textbf{1}}
\newcommand{\zero}{\textbf{0}}

\begin{document}

%%%%%%% TITLE PAGE %%%%%%%%%%%%%%%%%%%%%%%%%%%%%%%%%%%%%%%%%%%%%%%%%%%

\begin{center}

{\bf{\LARGE{Adaptive, Doubly Optimal No-Regret Learning in \\[.1cm] Strongly Monotone and Exp-Concave Games with \\[.3cm] Gradient Feedback}}}

\vspace*{.2in}
{\large{ \begin{tabular}{c}
Michael I. Jordan$^{\diamond, \dagger}$ \and Tianyi Lin$^\ddagger$ \and Zhengyuan Zhou$^\square$ \\
\end{tabular}
}}

\vspace*{.2in}

\begin{tabular}{c}
Department of Electrical Engineering and Computer Sciences$^\diamond$ \\
Department of Statistics$^\dagger$ \\
University of California, Berkeley \\
Department of Industrial Engineering and Operations Research (IEOR), Columbia University$^\ddagger$ \\
Stern School of Business, New York University$^\square$
\end{tabular}

\vspace*{.2in}

\today

\vspace*{.2in}

\begin{abstract}
Online gradient descent (OGD) is well known to be doubly optimal under strong convexity or monotonicity assumptions: (1) in the single-agent setting, it achieves an optimal regret of $\Theta(\log T)$ for strongly convex cost functions; and (2) in the multi-agent setting of strongly monotone games, with each agent employing OGD, we obtain last-iterate convergence of the joint action to a unique Nash equilibrium at an optimal rate of $\Theta(\frac{1}{T})$. While these finite-time guarantees highlight its merits, OGD has the drawback that it requires knowing the strong convexity/monotonicity parameters. In this paper, we design a fully adaptive OGD algorithm, \textsf{AdaOGD}, that does not require a priori knowledge of these parameters. In the single-agent setting, our algorithm achieves $O(\log^2(T))$ regret under strong convexity, which is optimal up to a log factor. Further, if each agent employs \textsf{AdaOGD} in strongly monotone games, the joint action converges in a last-iterate sense to a unique Nash equilibrium at a rate of $O(\frac{\log^3 T}{T})$, again optimal up to log factors.  We illustrate our algorithms in a learning version of the classical newsvendor problem, where due to lost sales, only (noisy) gradient feedback can be observed. Our results immediately yield the first feasible and near-optimal algorithm for both the single-retailer and multi-retailer settings. We also extend our results to the more general setting of exp-concave cost functions and games, using the online Newton step (ONS) algorithm. 
\end{abstract}

\end{center}

%!TEX root = ../paper.tex
\section{Introduction}\label{sec:introduction}
The problem of online learning with gradient feedback~\citep{Blum-1998-Online, Shalev-2012-Online, Hazan-2016-Introduction} can be described in its essential form by the following adaptive decision-making process:
\begin{enumerate}
\item An agent interfaces with the environment by choosing an \emph{action} $x^t \in \XCal$ at period $t$ where $\XCal \subseteq \br^d$ is a convex and compact set. For example, the action is a route in a traffic network or an output quantity in an oligopoly. The agent chooses $x^t$ through an online learning algorithm, which makes its choice adaptively based on observable historical information.
\item The environment then returns a cost function $f_t(\cdot)$ so that the agent incurs cost $f_t(x^t)$ and receives $\nabla f_t(x^t)$ as feedback. The process then moves to the next period $t+1$ and repeats.
\end{enumerate}
One appealing feature of the online learning framework is that one need not impose any statistical regularity assumption: $f_1(\cdot), \dots, f_T(\cdot)$ can be an arbitrary fixed sequence of cost functions, hence accommodating a non-stationary or even adversarial environment. Further, the cost function $f_t(\cdot)$ does need not to be known by the agent (and indeed in many applications it is not known); only the gradient feedback is needed. In this general framework, the standard metric for judging the performance of an online learning algorithm is \emph{regret}~\citep{Blum-2007-External}---the difference between the total cost incurred by the algorithm up to $T$ and the total cost incurred by the best fixed action in hindsight:
\begin{equation}\label{def:regret}
\text{Regret}(T) = \sum_{t=1}^T f_t(x^t) - \min_{x \in \XCal} \left\{\sum_{t=1}^T f_t(x)\right\}.
\end{equation}
If the average regret (obtained by dividing by $T$) goes to zero, then the algorithm is referred to as a ``no-regret learning algorithm.''

A canonical example of a no-regret learning algorithm is online gradient descent (OGD), where the agent takes a gradient step (given the current action $x_t$) and performs a projection step onto $\XCal$ to obtain the next action $x_{t+1}$.  Analyzing OGD in the standard setting where the cost function $f_t$ is convex,~\citet{Zinkevich-2003-Online} proved that the algorithm with the stepsize rule $\eta_t = \frac{1}{\sqrt{T}}$ achieves a regret bound of $\Theta(\sqrt{T})$, which is known to be minimax optimal~\citep{Hazan-2016-Introduction}. If $f_t$ is further assumed to be $\mu$-strongly convex,~\citet{Hazan-2007-Logarithmic} proved that the algorithm with the stepsize rule $\eta_t = \frac{1}{\mu t}$ achieves a regret bound of $\Theta(\log T)$; again this rate is minimax optimal. In fact, the $\Theta(\log T)$ regret bound is achievable even for a class of cost functions that are more general than strongly convex cost functions: if $f_t$ is exp-concave---a class of functions properly subsuming strongly convex functions that has found widespread applications (see Section~\ref{sec:SA} for a detailed discussion)---then the online Newton step (ONS)~\citep{Hazan-2007-Logarithmic} achieves the minimax optimal regret bound of $\Theta(d\log T)$. In summary, OGD and ONS provide two of the most well-known optimal no-regret learning algorithms in the online learning/online convex optimization literature, with their algorithmic simplicity and theoretical elegance being matched by their broad applicability in practice.   

Given their appealing theoretical and practical properties, the aforementioned no-regret learning algorithms such as OGD and ONS have also served as natural candidates for game-theoretic learning.  In this setting, each agent makes online decisions in an environment consisting of other agents, each of whom are making adaptive decisions. Even if the game is fixed, the fact that all agents are adjusting their strategies simultaneously means that the stream of costs seen by any single agent is non-stationary and complex. It might be hoped that no-regret learning---given its robustness to assumptions---can cope with the complexity of the multi-agent setting.  An extensive literature has shown that this hope is borne out---under no-regret dynamics, the time average of the joint actions converges to equilibria in various classes of games~\citep{Cesa-2006-Prediction, Shoham-2008-Multiagent, Viossat-2013-Noregret, Bloembergen-2015-Evolutionary, Monnot-2017-Limits}.

Further progress in the online learning and optimization literature has yielded more refined statements regarding convergence; in particular, last-iterate rate guarantees have been established for many no-regret learning algorithms.  This leads to an analogous question for the game-theoretic setting: \emph{If each agent employs a no-regret learning algorithm to minimize its own regret, can the joint action converge to a Nash equilibrium at an optimal last-iterate rate?}

An affirmative answer to this question\footnote{For instance, if each vehicle in a traffic network employs an optimal no-regret learning algorithm (such as OGD) to choose their route adaptively over a certain horizon, would the system converge to a stable traffic distribution or devolve to perpetual congestion as users ping-pong between different routes? If it does converge to a stable distribution, is it Nash? Because if not, each agent is being irrational---by \textit{not} following the no-regret learning algorithm, agents can do individually do better.} would establish a remarkable ``double optimality" for an online algorithm: while the algorithm itself is only designed for maximizing the (transient) performance for a finite time horizon $T$, the resulting long-run performance would also be optimal for all agents, in the sense that any agent, by unilaterally deviating from the action suggested by the algorithm, could only incur higher cost (by the definition of a Nash equilibrium). Without this double optimality, an agent could incur ``regret" in the long term, since it may do better by not following such an algorithm.

To address this question, it is necessary to study the last-iterate convergence of algorithms (i.e., the convergence of the \textit{actual} joint action), a problem that has been recognized to be considerably more difficult
than the characterization of convergence of time averages~\citep{Krichene-2015-Convergence, Balandat-2016-Minimizing, Zhou-2017-Mirror, Zhou-2018-Learning, Mertikopoulos-2018-Optimistic, Mertikopoulos-2019-Learning}. For instance, as pointed out by~\citet{Mertikopoulos-2018-Cycles}, there are situations where the time average of the iterate converges to a Nash equilibrium, but the last iterate cycles around the equilibrium point. Progress has been made on this problem during the past five years, but much of it only provides qualitative or asymptotic results, with only a few quantitative (finite-time, last-iterate convergence guarantees) results obtained, for games having special structures or using metrics other than the distance to Nash equilibria (see Section~\ref{subsec:related}). In particular, \citet{Zhou-2021-Robust} established the last-iterate convergence of multi-agent OGD to the unique Nash equilibrium in strongly monotone games\footnote{When the strongly monotone games have Lipschitz gradients---a condition that does not hold in many games of interest---a classic result from the variational inequality literature implies that multi-agent OGD converges to the unique Nash equilibrium at a geometric rate due to a contraction.} at an optimal rate of $\Theta({\frac{1}{T}})$. This optimal convergence rate continues to hold even when the gradient feedback is corrupted by certain forms of noise, in which case we have $E[\|x_t - x^\star\|_2^2] = \Theta(\frac{1}{T})$, where $x_t$ is the (random) joint action of all agents and $x^\star$ is the unique Nash equilibrium.\footnote{This result is further generalized in~\citet{Loizou-2021-Stochastic}, who obtain the same last-iterate convergence rate for strongly variationally stable games, under weaker noise assumptions.} Thus, OGD is doubly optimal when a strong convexity structure is available (i.e., the game is strongly monotone or the cost functions are strongly convex from a single-agent perspective), giving a compelling argument for its adoption in single-agent and multi-agent settings. 

However, this argument suffers from a key, subtle weakness: the theoretical guarantees for OGD require choices of step size, in both single-agent~\citep{Hazan-2007-Logarithmic} and multi-agent~\citep{Zhou-2021-Robust} settings, and these choices require prior knowledge of problem parameters. In particular, it is generally assumed that the strong convexity parameter of the cost functions (single-agent) or the strong monotonicity parameter of the game (multi-agent) are known. Thus, OGD's appealing guarantees are not feasible in practice if the choice of step sizes cannot be made \emph{fully adaptive} to problem parameters. Further, the feasibility issue is more acute in the multi-agent setting: in addition to requiring prior knowledge of problem parameters for step-size designs, recent work on adaptive OGD has assumed that each agent determines their step size using global information from all agents~\citep{Lin-2020-Finite, Antonakopoulos-2021-Adaptive, Hsieh-2021-Adaptive}. This is a practical and theoretical conundrum---if the agents can achieve this level of coordination, learning would be unnecessary in the first place. The same issue occurs for ONS in the single-agent setting, where the exp-concave parameter is needed as an input to the algorithm.\footnote{The multi-agent ONS has not yet been explored, and even the time-average convergence of ONS still remains to be established.} Consequently, the feasibility considerations lead us to consider the following question: \textit{Can we design a feasible and doubly optimal variant of OGD under strong convexity and strong monotonicity? What about ONS?}

Our answer is a ``yes'' in a strong sense. We present a single feasible OGD algorithm---and hence a single parameter-adaptive scheme---that simultaneously (up to log factors) achieves optimal regret in the single-agent setting and optimal last-iterate convergence rate to the unique Nash equilibrium in the multi-agent setting. This analysis is different from and more challenging than that involved in the design of feasible OGD algorithms separately for single-agent and multi-agent settings. In particular, it could be that an effective adaptive scheme for the strong convexity parameter in single-agent setting is different from that for the strong monotonicity parameter in multi-agent setting, in which case one has \textit{at best} either a feasible algorithm with optimal regret or a feasible algorithm that has an optimal convergence-to-Nash guarantee, but not both. Such results would still be of considerable value but our results in this paper show that such intermediate results can be bypassed; indeed, the best of both worlds can be achieved. We also develop a single feasible variant of ONS that (up to log factors) achieves optimal regret in the single-agent setting with exp-concave loss functions and optimal time-average convergence rate to the unique Nash equilibrium in the multi-agent setting.  For the latter result we introduce and analyze a new class of exp-concave games.   

Our results can also be cast in the framework of variational inequalities (VIs).  Indeed, they can be viewed as contributing to the VI literature by presenting a decentralized, feasible optimization algorithm for finding a solution of a strongly monotone VI.  We prefer to emphasize, however, the online learning perspective, and the design of no-regret algorithms, given the direct connection of those algorithms to game-theoretic settings.  In multi-agent games, it is natural to focus on decentralized algorithms and on algorithms that make minimal assumptions about their environment, allowing that environment to consist of other agents that may be responding in complex ways to an agent's actions. Our double optimality contributions are best understood as a further weakening of these assumptions, allowing interacting agents to choose actions effectively in an unknown, possibly adversarial, environment.

\subsection{Related Work}\label{subsec:related}
In both single-agent online learning and offline optimization,  considerable attention has been paid to the development of adaptive gradient-based schemes. In particular,~\citet{Duchi-2011-Adaptive} presented an adaptive gradient algorithm (known as \textsf{AdaGrad}) for online learning with convex cost functions that updates the step sizes without needing to know the problem parameters. This algorithm is guaranteed to achieve a minimax-optimal regret of $O(\sqrt{T})$. Subsequently, \textsf{Adam} was proposed in the offline optimization setting to further exploit geometric aspects of iterate trajectories, exhibiting better empirical convergence performance~\citep{Kingma-2015-Adam}. Theoretical guarantees have been obtained for \textsf{Adam} and other related adaptive algorithms in both offline optimization and online learning settings~\citep{Reddi-2018-Adam, Zou-2019-Sufficient}. In parallel, the norm version of \textsf{AdaGrad} was developed and theoretical guarantees were established for related convex and/or nonconvex optimization problems~\citep{Levy-2017-Online, Levy-2018-Online, Ward-2019-Adagrad, Li-2019-Convergence}. 

This adaptive family of algorithms has also been studied in the online learning literature under an assumption of strong convexity~\citep{Mukkamala-2017-Variants,Wang-2020-SAdam}. The algorithms are guaranteed to achieve a minimax-optimal regret of $O(\log(T))$.  However, unlike in the convex setting, these adaptive algorithms are not feasible in practice since they often require  knowledge of the strong convexity parameter. In the offline optimization setting, the gradient-based methods can be made adaptive to the strong convexity parameter by exploiting the Polyak stepsize~\citep{Polyak-1987-Introduction, Hazan-2019-Revisiting}. However, these algorithms do not extend readily to online learning since the sub-optimality gap is not well-defined. A recent line of research has shown that adaptive algorithms can be  designed in the finite-sum setting, where they exhibit adaptivity to the strong convexity parameter~\citep{Roux-2012-Stochastic, Defazio-2014-Saga, Xu-2017-Adaptive, Lei-2017-Less, Lei-2020-Adaptivity, Vaswani-2019-Painless, Nguyen-2022-Finite}. In particular,~\citet{Lei-2017-Less, Lei-2020-Adaptivity} showed that the use of random, geometrically-distributed epoch length yields full adaptivity in variance-controlled stochastic optimization under an assumption of strong convexity. However, these algorithm are not no-regret and thus their strategies do not extend readily to the online setting. 

In terms of the last-iterate convergence to Nash equilibria, due to the difficulties mentioned earlier, much of the existing literature provides only qualitative convergence guarantees for non-adaptive (and hence infeasible) no-regret learning algorithms for various games~\citep{Krichene-2015-Convergence, Balandat-2016-Minimizing, Zhou-2017-Mirror, Zhou-2018-Learning, Mertikopoulos-2018-Optimistic, Mertikopoulos-2019-Learning}. More recently, finite-time last-iterate convergence rates have been obtained for specially structured games, such as strongly monotone games~\citep{Zhou-2021-Robust,Loizou-2021-Stochastic},  unconstrained cocoercive games~\citep{Lin-2020-Finite}, unconstrained smooth games~\citep{Golowich-2020-Tight} and constrained smooth games~\citep{Cai-2022-Finite}. Except for a class of strongly monotone games, the last-iterate convergence rate is measured in metrics other than $\|x^t - x^\star\|_2^2$. Further, among these results, only~\citet{Lin-2020-Finite} provides an adaptive online learning algorithm that does not require knowing the cocoercivity parameter. However, their algorithm falls short in two respects: (i) it may not be no-regret; (ii) each agent needs to know all other agents' gradients, thus again rendering it infeasible in practice. Recently,~\citet{Antonakopoulos-2021-Adaptive} has developed an adaptive extragradient algorithm for strictly monotone games that converges asymptotically in a last-iterate sense to the unique Nash equilibrium. However, their algorithm also requires each agent to know all others' gradients and the no-regret property cannot be guaranteed; indeed, the original extragradient algorithm was shown to not be no-regret~\citep{Golowich-2020-Tight}. \citet{Hsieh-2021-Adaptive} has proposed a no-regret adaptive online learning algorithm based on optimistic mirror descent and established a regret of $O(\sqrt{T})$ for convex cost functions. They also proved asymptotic last-iterate convergence to the unique Nash equilibrium for strictly variationally stable games (a superset of strictly monotone games).

Another line of relevant literature focuses on stochastic approximation methods for solving strongly monotone VIs. An early proposal using such an approach was presented by~\citet{Jiang-2008-Stochastic},  who proposed a stochastic projection method for solving strongly monotone VIs with an almost-sure convergence guarantee.~\citet{Koshal-2012-Regularized} and~\citet{Yousefian-2013-Regularized} proposed various regularized iterative stochastic approximation methods for solving monotone VIs and also established almost-sure convergence. A survey of these methods, as well as applications and the theory behind stochastic VI, can be found in~\citet{Shanbhag-2013-Stochastic}.~\citet{Juditsky-2011-Solving} was among the first to establish an iteration complexity bound for stochastic VI methods by extending the mirror-prox method~\citep{Nemirovski-2004-Prox} to stochastic setting.~\citet{Yousefian-2014-Optimal} further extended the stochastic mirror-prox method with a more general step size choice and proved the same iteration complexity. They also proved an improved complexity bound for the stochastic extragradient method for solving strongly monotone VIs.~\citet{Chen-2017-Accelerated} studied a specific class of VIs and proposed a method that combines the stochastic mirror-prox method with Nesterov's acceleration~\citep{Nesterov-2018-Lectures}, resulting in an optimal iteration complexity for such problem class.~\citet{Kannan-2019-Optimal} analyzed a general variant of an extragradient method (which uses general distance-generating functions) and proved an optimal iteration bound under a slightly weaker assumption than strong monotonicity. Several other stochastic methods have also been shown to yield an optimal iteration bound for solving strongly monotone VIs~\citep{Kotsalis-2022-Simple, Huang-2022-New}. In recent years, there have been developments in variance-reduction-based methods~\citep{Balamurugan-2016-Stochastic, Iusem-2017-Extragradient, Iusem-2019-Variance, Jalilzadeh-2019-Proximal, Yu-2022-Fast, Alacaoglu-2022-Stochastic, Jin-2022-Sharper, Huang-2022-Accelerated}. In the line of research aiming to model multistage stochastic VI (as compared to the single-stage VI considered in the above literature), the dynamics between the actions and the arrival of future information plays a central role. For details regarding multistage stochastic VI, we refer to~\citet{Rockafellar-2017-Stochastic} and~\citet{Rockafellar-2019-Solving}. 

We are also aware of two related topics: (1) parameter-free online learning~\citep{Foster-2015-Adaptive, Foster-2017-Parameter, Orabona-2016-Coin, Cutkosky-2018-Black, Jun-2019-Parameter, Cutkosky-2020-Better, Cutkosky-2020-Parameter} and (2) online learning with adaptive/dynamic regret~\citep{Zinkevich-2003-Online, Hazan-2007-Adaptive, Daniely-2015-Strongly, Hazan-2016-Introduction}.  The former line of works consider an alternative form of regret: $\text{P-Regret}_T(x) = \sum_{t=1}^T f_t(x^t) - \sum_{t=1}^T f_t(x)$, where $x \in \XCal$ is an unknown \textit{competitor}. The goal is to achieve expected regret bounds that have optimal dependency not only on $T$ but also $\|x\|$. Notably, their framework provides a way to design algorithms that achieve minimax-optimal regret with respect to any competitor, without imposing a bounded set for the competitor nor any parameter to tune in online convex optimization. However, the strong convexity parameter and the exp-concavity parameter both characterize \textit{the lower bound for curvature information}. Estimating these quantities will require new techniques. The second line of works consider two other form of regrets: $\text{D-Regret}_T(x_1, \ldots, x_T) = \sum_{t=1}^T f_t(x^t) - \sum_{t=1}^T f_t(x_t)$ (dynamic) where $x_1, \ldots, x_T \in \XCal$ are any competitors, and $\text{A-Regret}_T(\tau) = \max_{1 \leq s \leq T+1-\tau}\left\{\sum_{t=s}^{s+\tau-1} f_t(x^t) - \min_{x \in \XCal} \left\{\sum_{t=s}^{s+\tau-1} f_t(x)\right\}\right\}$ (adaptive), which is the maximum static regret over an interval with the length $\tau$. In this context, it is impossible to obtain a sublinear dynamic regret with arbitrarily varying sequences and we often assume certain conditions, such as an upper bound for path-length $P_T = \sum_{t=2}^T \|x_{t-1} - x_t\|$. The seminar work of~\citet{Zinkevich-2003-Online} established the first general-case bound of $O(\sqrt{T}(1+P_T))$ for OGD and~\citet{Zhang-2018-Dynamic, Zhang-2018-Adaptive} improved it to $O(\sqrt{T(1+P_T)})$. Several recent studies have further investigated the dynamic regret by leveraging the curvature of loss functions, such as exponential concavity~\citep{Baby-2021-Optimal} and strong convexity~\citep{Baby-2022-Optimal}. There is also the worst-case variant of dynamic regret: $\text{D-Regret}_T(x_1, \ldots, x_T) = \sum_{t=1}^T f_t(x^t) - \sum_{t=1}^T f_t(x_t^\star)$ where $x_t^\star = \argmin_{x \in \XCal} f_t(x_t)$ and has been extensively studied in operations research and computer science~\citep{Besbes-2015-Non, Jadbabaie-2015-Online, Mokhtari-2016-Online, Yang-2016-Tracking}. Moreover, the previous works in adaptive regret minimization mainly focus on the setting of online convex optimization~\citep{Hazan-2007-Adaptive, Daniely-2015-Strongly, Jun-2017-Online, Jun-2017-Improved, Zhang-2019-Adaptive}. The seminar work of~\citet{Hazan-2007-Adaptive} first introduced the notion of adaptive regret in a weak form and proposed an algorithm with an $O(d\log^2(T))$ adaptive regret bound for exponentially concave losses. However, the weak form of adaptive regret could be dominated by long intervals and hence, cannot respect short intervals well. This issue was addressed by~\citet{Daniely-2015-Strongly} who put forth the strongly adaptive regret $\text{A-Regret}_T(\cdot)$ and design a two-layer algorithm which combines OGD and the geometric covering intervals construction. In addition, a few projection-free online algorithms were developed to investigate adaptive/dynamic regret minimization~\citep{Kalhan-2021-Dynamic, Wan-2021-Projection, Wan-2023-Improved, Garber-2022-New, Lu-2023-Projection}. However, all of these algorithms require knowing the strong convexity parameter and the exp-concavity parameters. 

In summary, the possibility of designing an online algorithm that is both \textit{doubly optimal} and \textit{feasible} under strong convexity still remains open. 

\subsection{Our Contributions}
We present a feasible variant of OGD that we refer to as \textsf{AdaOGD} that does not require knowing any problem parameter. It is guaranteed to achieve a minimax optimal (up to a log factor) regret bound of $O(\log^2(T))$ in the single-agent setting with strongly convex cost functions. Further, in a strongly monotone game, if each agent employs \textsf{AdaOGD}, we show that the joint action converges to the unique Nash equilibrium in a \textit{last-iterate sense} at a rate of $O(\frac{\log^3(T)}{T})$, again optimal up to log factors. In comparison, the existing single-agent OGD~\citep{Hazan-2016-Introduction} and multi-agent OGD~\citep{Zhou-2021-Robust} methods require prior knowledge of the strong convexity/monotonicity parameter (respectively) to perform the step-size design with theoretical guarantees. It is worth noting that if the game has only a single agent, the strong monotonicity parameter degenerates to the strong convexity parameter. However, when there are multiple agents, the strong monotonicity parameter depends on all agents' cost functions. As such, one would naturally think that these two settings would require two different adaptive schemes. Surprisingly, our \textsf{AdaOGD} algorithm, which is based on a single adaptive principle, works in both settings, achieving optimal regret in the single-agent setting and optimal last-iterate convergence in the multi-agent setting (up to log factors). A particularly important application of our results is the problem of learning to order in the classical newsvendor problem, where due to lost sales, only (noisy) gradient feedback can be observed. Our results immediately yield the first feasible near-optimal algorithm---both in the single-retailer setting~\citep{Huh-2009-Nonparametric} and in the multi-retailer setting~\citep{Netessine-2003-Centralized}.  This is in contrast to previous work that requires problem parameters to be known. Indeed, the direct application of our results to the VI setting yields a decentralized, feasible optimization algorithm for finding a solution of a strongly monotone VI. 

Additionally, we provide a feasible variant of ONS (that we refer to as \textsf{AdaONS}) that again does not require prior knowledge of any problem parameter. It is also guaranteed to achieve a minimax optimal regret bound of $O(d\log^2(T))$ (up to a log factor) in the single-agent setting with exp-concave cost functions. Further, we propose a new class of exp-concave (EC) games and show that if each agents employs \textsf{AdaONS}, an optimal time-average convergence rate of $O(\frac{d\log^2(T)}{T})$ is obtained. Much like strongly monotone games that provide a multi-agent generalization of strongly convex cost functions, the EC games that we introduce are a natural generalization of exp-concave cost functions from the single-agent to a multi-agent setting. Again, in this case, a single adaptive scheme works for both of these two settings. One thing to note here is that we first establish an $O(\frac{d\log(T)}{T})$ time-average convergence rate for the multi-agent version of classical ONS that is non-adaptive (and hence not feasible). To the best of our knowledge, results of this kind have not appeared in the game-theoretic literature and they yield a decentralized, feasible optimization algorithm for finding a solution of a new class of VIs in that line of literature. 

Perhaps the most surprising takeaway from our work is that both \textsf{AdaOGD} and \textsf{AdaONS} are based on a simple and unifying randomized strategy that selects the step size based on a set of independent and identically distributed geometric random variables.

%!TEX root = main.tex
\section{Feasible Single-Agent Online Learning under Strongly Convex Costs}\label{sec:SA}
In this section, we present adaptive OGD (\textsf{AdaOGD}), a feasible single-agent online learning algorithm, and prove that \textsf{AdaOGD} achieves a near-optimal regret of $O(\log^2(T))$ for a class of strongly convex cost functions. We also show that our algorithm can be used to solve the problem of adaptive ordering in newsvendor problems with lost sales. To our knowledge, this is the first feasible no-regret learning algorithm for the newsvendor-with-lost-sales problem with strong regret guarantees. 

\subsection{Algorithmic Scheme}
We continue with the setup in the introduction, focusing on strongly convex cost functions $f_t$:
\begin{definition}\label{def:SC-function}
A function $f: \br^d \mapsto \br$ is $\beta$-strongly convex if $f(\cdot) - 0.5\beta\|\cdot\|^2$ is convex. 
\end{definition}
We work with a more general (and relaxed) model of gradient feedback~\citep{Flaxman-2005-Online}:
\begin{enumerate}
\item At each round $t$, an unbiased and bounded gradient is observed. That is, the observed noisy gradient $\xi^t$ satisfies $\EE[\xi^t \mid x^t] = \nabla f_t(x^t)$ and $\EE[\|\xi^t\|^2 \mid x^t] \leq G^2$ for all $t \geq 1$.
\item The action set $\XCal$ is bounded by a diameter $D > 0$, i.e., $\|x - x'\| \leq D$ for all $x, x' \in \XCal$.
\end{enumerate}
A lower bound of $\Omega(\log(T))$ has been established in~\citet[Theorem~18]{Hazan-2014-Beyond} under an assumption of perfect gradient feedback. However, even with noisy gradient feedback, OGD with a particular step size can achieve the minimax-optimal regret bound of $\Theta(\log(T))$~\citep{Hazan-2007-Logarithmic}. In particular, we write OGD as $x^{t+1} \leftarrow \PCal_\XCal(x^t - \frac{1}{\beta(t+1)}\nabla f_t(x^t))$, which is equivalent to 
\begin{equation*}
\eta^{t+1} \leftarrow \beta(t+1), \qquad x^{t+1} \leftarrow \argmin_{x \in \XCal}\{(x - x^t)^\top \nabla f_t(x^t) + \tfrac{\eta^{t+1}}{2}\|x - x^t\|^2\}.
\end{equation*}
The value of $\beta(t+1)$ comes from the key inequality for $\beta$-strongly convex functions:
\begin{equation}\label{def:SC-inequality}
f(x') \geq f(x) + (x' - x)^\top\nabla f(x) + \tfrac{\beta}{2}\|x' - x\|^2. 
\end{equation}
Despite the elegance of OGD (and its optimal regret guarantee), however, it is inadequate since it requires knowledge of the problem parameter $\beta$. We can address this issue by a simple randomization strategy based on independent, identically distributed geometric random variables, $M^t \sim \textsf{Geometric}(p_0)$, for $p_0 = \frac{1}{\log(T+10)}$; i.e., $\PP(M^t = k) = (1-p_0)^{k-1}p_0$ for $k \in \{1, 2, \ldots\}$. See Algorithm~\ref{alg:AdaOGD-SA}. 
\begin{algorithm}[!t]
\caption{\textsf{AdaOGD}($x^1$, $T$)}\label{alg:AdaOGD-SA}
\begin{algorithmic}[1]
\STATE \textbf{Input:} initial point $x^1 \in \XCal$ and the total number of rounds $T$.  
\STATE \textbf{Initialization:} $p_0 = \frac{1}{\log(T+10)}$. 
\FOR{$t = 1, 2, \ldots, T$}
\STATE sample $M^t \sim \textsf{Geometric}(p_0)$. 
\STATE set $\eta^{t+1} \leftarrow \tfrac{t+1}{\sqrt{1 + \max\{M^1, \ldots, M^t\}}}$. 
\STATE update $x^{t+1} \leftarrow \argmin_{x \in \XCal}\{(x - x^t)^\top \xi^t + \tfrac{\eta^{t+1}}{2}\|x - x^t\|^2\}$. 
\ENDFOR
\end{algorithmic}
\end{algorithm} 
\begin{remark}[Compared with doubling trick] 
In the context of online learning, the doubling trick~\citep{Shalev-2012-Online} is commonly used to make OGD adaptive to \textbf{specific} unknown parameters under \textbf{convex} costs. In particular, for any algorithm that enjoys a regret bound of $O(\sqrt{T})$ but requires the knowledge of $T$ (such as OGD under convex costs), the doubling trick converts such an algorithm into an algorithm that does not require the knowledge of $T$. The idea is to divide the time into periods of increasing size and run the original algorithm on each period: for $m = 0, 1, 2, \ldots$, we run the original algorithm on the $2^m$ rounds $t=2^m, \ldots, 2^{m+1}-1$. The resulting algorithm still enjoys a regret bound of $O(\sqrt{T})$. However, it is nontrivial to apply the doubling trick to make OGD adaptive to the strongly convex parameter $\beta$ under \textbf{strongly convex} costs. In particular, a natural adaptation of the doubling trick under strongly convex costs is as follows: for $m = 0, 1, 2, \ldots$, we run OGD with $\eta^t = \frac{t}{2^m}$ on the rounds $t=2^m, \ldots, 2^{m+1}-1$. The analysis contains two parts: (i) for $0 \leq m \leq \lfloor \log_2(1/\beta)\rfloor$, the regret for each round is $O(2^m)$. This leads to a total constant regret of $O(1/\beta)$; (ii) for $\lceil \log_2(1/\beta)\rceil\leq m \leq \lceil \log_2(T)\rceil$, the regret for each round is $O(2^m)$. This unfortunately leads to a linear regret of $O(T)$. This argument of course does not eliminate the possibility that some variant of doubling could lead to doubly optimal learning algorithms for strongly monotone games. Our results do suggest that a fruitful way to search for such a procedure would be via some form of randomization. 
\end{remark}
\begin{remark}[Comparison with online learning with adaptive regret] 
The static regret defined in Eq.~\eqref{def:regret} can be unsuitable as the environments are non-stationary and the best action is drifting over time~\citep{Hazan-2016-Introduction}. To tackle this issue, there are two independently proposed notions: dynamic regret~\citep{Zinkevich-2003-Online, Besbes-2015-Non} and adaptive regret~\citep{Hazan-2007-Adaptive, Daniely-2015-Strongly}. Several recent studies have further investigated the adaptive/dynamic regret minimization by leveraging the curvature of loss functions, such as exponential concavity~\citep{Hazan-2007-Adaptive, Baby-2021-Optimal} and strong convexity~\citep{Mokhtari-2016-Online, Baby-2022-Optimal}. However, the adaptivity in regret is \textit{different} from that in learning rates and all of these algorithms require knowing the strong convexity parameter and the exp-concavity parameters. 
\end{remark}
\begin{remark}[Comparison with geometrization]
In the offline finite-sum optimization, the geometrization trick~\citep{Lei-2017-Less, Lei-2020-Adaptivity}---which sets the length of each epoch as a geometric random variable---has been used to make the stochastic variance-reduced gradient (SVRG) algorithm of~\citet{Johnson-2013-Accelerating} adaptive to both strong convexity parameter and target accuracy. While similar in spirit, our approach is not a straightforward application of the geometrization trick. The key difference between OGD and SVRG is their dependence on the strongly convexity parameter. The former needs that parameter to set the stepsize while the latter algorithm needs it to set the length of each epoch. The intuition behind the geometrization trick is to randomly set the length of each epoch using geometric random variables, allowing terms to telescope across the outer and inner loops; such telescoping does not happen in SVRG, a fact which leads to the loss of adaptivity for SVRG (see~\citet[Section~3.1]{Lei-2020-Adaptivity}). In contrast, our technique is designed to randomly set the stepsize using geometric random variables, implementing a trade-off for bounding the regret and last-iterate convergence rate. The telescoping from~\citep{Lei-2020-Adaptivity} occurs due to the specific nature of SVRG (or more broadly, the setting of offline finite-sum optimization), and does not appear in our analysis for \textsf{AdaOGD} and its multi-agent generalization. 
\end{remark}
\subsection{Regret Guarantees}
 We present our result on the regret minimization property in the following theorem. 
\begin{theorem}\label{Thm:AdaOGD-regret}
For an arbitrary fixed sequence of $\beta$-strongly convex functions $f_1, \ldots, f_T$, where each $f_t$ satisfies $\EE[\xi^t \mid x^t] = \nabla f_t(x^t)$ and $\EE[\|\xi^t\|^2 \mid x^t] \leq G^2$ for all $t \geq 1$, and $\|x - x'\| \leq D$ for all $x, x' \in \XCal$. If the agent employs Algorithm~\ref{alg:AdaOGD-SA}, we have 
\begin{equation*}
\EE[\text{Regret}(T)] \leq \tfrac{D^2}{2}(1 + e^{\frac{1}{\beta^2\log(T+10)}}) + \tfrac{G^2\log(T+1)}{2}\sqrt{1 + \log(T+10) + \log(T)\log(T+10)}. 
\end{equation*}
As a consequence, we have $\EE[\text{Regret}(T)] = O(\log^2(T))$. 
\end{theorem}
\begin{remark}
Theorem~\ref{Thm:AdaOGD-regret} demonstrates that Algorithm~\ref{alg:AdaOGD-SA} achieves a near-optimal regret since the upper bound matches the lower bound up to a log factor; indeed,~\citet{Hazan-2014-Beyond} proved the lower bound of $\Omega(\log(T))$ for this setting. Further, Algorithm~\ref{alg:AdaOGD-SA} dynamically adjusts $\eta^{t+1}$ without any prior knowledge of problem parameters, only utilizing the noisy feedback $\{\xi^t\}_{t \geq 1}$. 
\end{remark}
To prove Theorem~\ref{Thm:AdaOGD-regret}, we present a descent inequality for the iterates generated by Algorithm~\ref{alg:AdaOGD-SA}. 

\begin{lemma}\label{Lemma:AdaOGD}
For an arbitrary fixed sequence of $\beta$-strongly convex functions $f_1, \ldots, f_T$, where each $f_t$ satisfies $\EE[\xi^t \mid x^t] = \nabla f_t(x^t)$ and $\EE[\|\xi^t\|^2 \mid x^t] \leq G^2$ for all $t \geq 1$, and $\|x - x'\| \leq D$ for all $x, x' \in \XCal$. Letting the iterates $\{x^t\}_{t \geq 1}$ be generated by Algorithm~\ref{alg:AdaOGD-SA}, we have
\begin{equation*}
\sum_{t=1}^T \EE[f_t(x^t) - f_t(x)] \leq \tfrac{\eta^1}{2}\|x^1 - x\|^2 + \sum_{t=1}^T \EE\left[\left(\tfrac{\eta^{t+1} - \eta^t}{2} - \tfrac{\beta}{2}\right)\|x^t - x\|^2\right] + \tfrac{G^2}{2}\left(\sum_{t=1}^T \EE\left[\tfrac{1}{\eta^{t+1}}\right]\right), \textnormal{ for all } x \in \XCal. 
\end{equation*}
\end{lemma}
The proof of this lemma is deferred to Appendix~\ref{app:AdaOGD}. 

\paragraph{Proof of Theorem~\ref{Thm:AdaOGD-regret}.} Recall that $\XCal$ is convex and bounded with a diameter $D > 0$ and we have $\eta^{t+1} = \frac{t+1}{\sqrt{1 + \max\{M^1, \ldots, M^t\}}}$ in Algorithm~\ref{alg:AdaOGD-SA}, we have
\begin{equation*}
\tfrac{\eta^1}{2}\|x^1 - x\|^2 \leq \tfrac{D^2}{2}, \qquad \eta^{t+1} - \eta^t \leq \tfrac{1}{\sqrt{1 + \max\{M^1, \ldots, M^t\}}}. 
\end{equation*}
By Lemma~\ref{Lemma:AdaOGD}, we have
\begin{equation}\label{inequality:AdaOGD-regret-first}
\sum_{t=1}^T \EE[f_t(x^t) - f_t(x)] \leq \tfrac{D^2}{2} + \sum_{t=1}^T \EE\left[\left(\tfrac{1}{2\sqrt{1 + \max\{M^1, \ldots, M^t\}}} - \tfrac{\beta}{2}\right)\|x^t - x\|^2\right] + \tfrac{G^2}{2}\left(\sum_{t=1}^T \EE\left[\tfrac{1}{\eta^{t+1}}\right]\right).
\end{equation}
Further, we have 
\begin{equation}\label{inequality:AdaOGD-regret-second}
\sum_{t=1}^T \tfrac{1}{\eta^{t+1}} \leq \sqrt{1 + \max\{M^1, \ldots, M^T\}}\left(\sum_{t=1}^T \tfrac{1}{t+1}\right) \leq \sqrt{1 + \max\{M^1, \ldots, M^T\}}\log(T+1). 
\end{equation}
Plugging Eq.~\eqref{inequality:AdaOGD-regret-second} into Eq.~\eqref{inequality:AdaOGD-regret-first} yields that 
\begin{equation*}
\sum_{t=1}^T \EE[f_t(x^t) - f_t(x)] \leq \tfrac{D^2}{2} + \sum_{t=1}^T \EE\left[\left(\tfrac{1}{2\sqrt{1 + \max\{M^1, \ldots, M^t\}}} - \tfrac{\beta}{2}\right)\|x^t - x\|^2\right] + \tfrac{G^2\log(T+1)}{2}\EE\left[\sqrt{1 + \max\{M^1, \ldots, M^T\}}\right]. 
\end{equation*}
Since $\XCal$ is convex and bounded with a diameter $D > 0$, we have
\begin{equation*}
\sum_{t=1}^T \left(\tfrac{1}{2\sqrt{1 + \max\{M^1, \ldots, M^t\}}} - \tfrac{\beta}{2}\right)\|x^t - x\|^2 \leq \tfrac{D^2}{2}\left(\sum_{t = 1}^T \max\left\{0, \tfrac{1}{\sqrt{1 + \max\{M^1, \ldots, M^t\}}} - \beta\right\}\right). 
\end{equation*}
This implies that 
\begin{equation*}
\EE[\text{Regret}(T)] \leq \tfrac{D^2}{2} + \tfrac{D^2}{2}\underbrace{\EE\left[\sum_{t = 1}^T \max\left\{0, \tfrac{1}{\sqrt{1 + \max\{M^1, \ldots, M^t\}}} - \beta\right\}\right]}_{\textbf{I}} + \tfrac{G^2\log(T+1)}{2}\underbrace{\EE\left[\sqrt{1 + \max\{M^1, \ldots, M^T\}}\right]}_{\textbf{II}}. 
\end{equation*}
It remains to bound the terms $\textbf{I}$ and $\textbf{II}$ using Proposition~\ref{Prop:GRV} (cf. Appendix~\ref{app:RV}). Indeed, we have
\begin{eqnarray*}
\textbf{I} & = & \sum_{t = 1}^T \EE\left[\max\left\{0, \tfrac{1}{\sqrt{1 + \max\{M^1, \ldots, M^t\}}} - \beta\right\}\right] \ \leq \ \sum_{t = 1}^T \PP\left(\tfrac{1}{\sqrt{1 + \max\{M^1, \ldots, M^t\}}} - \beta \geq 0\right) \\
& = & \sum_{t = 1}^T \PP\left(\sqrt{1 + \max\{M^1, \ldots, M^t\}} \leq \tfrac{1}{\beta}\right) \ \leq \ \sum_{t = 1}^T \PP\left(\max\{M^1, \ldots, M^t\} \leq \tfrac{1}{\beta^2}\right). 
\end{eqnarray*}
Since $M^1, \ldots, M^t$ are i.i.d.\ geometric random variables with $p_0 = \frac{1}{\log(T+10)}$,  Proposition~\ref{Prop:GRV} implies that 
\begin{equation*}
\sum_{t = 1}^T \PP\left(\max\{M^1, \ldots, M^t\} \leq \tfrac{1}{\beta^2}\right) \leq e^{\frac{p_0}{\beta^2}} = e^{\frac{1}{\beta^2\log(T+10)}}. 
\end{equation*}
Putting these pieces together yields that $\textbf{I} \leq e^{\frac{1}{\beta^2\log(T+10)}}$. 

By using Jensen's inequality and the concavity of $g(x) = \sqrt{x}$, we have
\begin{equation*}
\textbf{II} \leq \sqrt{1 + \EE\left[\max\{M^1, \ldots, M^T\}\right]}. 
\end{equation*}
Using Proposition~\ref{Prop:GRV} and $p_0 = \frac{1}{\log(T+10)}$, we have
\begin{equation*}
\EE\left[\max\{M^1, \ldots, M^T\}\right] \leq \tfrac{1 + \log(T)}{p_0} = \log(T+10) + \log(T)\log(T+10). 
\end{equation*}
Putting these pieces together yields that $\textbf{II} \leq \sqrt{1 + \log(T+10) + \log(T)\log(T+10)}$. Therefore, we conclude that 
\begin{eqnarray*}
\EE[\text{Regret}(T)] & \leq & \tfrac{D^2}{2}(1 + e^{\frac{1}{\beta^2\log(T+10)}}) + \tfrac{G^2\log(T+1)}{2}\sqrt{1 + \log(T+10) + \log(T)\log(T+10)} \\
& = & \tfrac{D^2}{2}(1 + e^{\frac{1}{\beta^2\log(T+10)}}) + \tfrac{G^2\log(T+1)}{2}\sqrt{1 + \log(T+10) + \log(T)\log(T+10)}. 
\end{eqnarray*}
This completes the proof.

\subsection{Application: Feasible Learning for Newsvendors with Lost Sales}
The single-retailer version of the newsvendor problem is a well known model for perishable inventory control~\citep{Huh-2009-Nonparametric}.  The assumption in this model is that unsold inventory perishes at the end of each period. A retailer sells a product over a time horizon $T$ and then makes inventory-ordering decisions $x^t \in [0, \bar{x}]$ at the beginning of each period $t$ to maximize the profit. The unknown demand $D^t$ is random and only realized with a value $d^t$ after the retailer makes her decision. It is often assumed in the inventory control literature that the $D^t$ are independent, corresponding to a stationary
environment. Here, we do not need to make this assumption and we allow $D^t$ to be arbitrary.
Further, in the lost-sales setting, any unmet demand is lost and hence the retailer does \textit{not} observe $d^t$; instead, she only observes the sales quantity $\min\{x^t, d^t\}$. The retailer's cost functions are defined by 
\begin{equation}\label{Eq:RIG-cost}
f_t(x^t) = (p-c) \cdot \EE[\max\{0, D^t - x^t\}] + c \cdot \EE[\max\{0, x^t - D^t\}],
\end{equation}
where the unit purchase cost is $c > 0$ and the unit selling price is $p \geq c$. It is known that minimizing this cost is equivalent to maximizing the profit, $\EE[p \cdot \min\{x^t, D^t\} - c \cdot x^t]$, where:
\begin{eqnarray*}
\lefteqn{\EE[p \cdot \min\{x^t, D^t\} - c \cdot x^t]} \\
& = & \EE\left[p \cdot \left(D^t - \max\{0, D^t - x^t\}\right) - c \cdot \left(D^t - \max\{0, D^t - x^t\} + \max\{0, x^t - D^t\}\right)\right] \\
& = & (p-c) \cdot \EE[D^t] - (p-c) \cdot \EE[\max\{0, D^t - x^t\}] - c \cdot \EE[\max\{0, x^t - D^t\}] \\
& = & (p-c) \cdot \EE[D^t] - f_t(x^t). 
\end{eqnarray*}
Note that the first term $(p-c) \cdot \EE[D^t]$ is independent of $x^t$. Thus, the maximization of the profit $\EE[p \cdot \min\{x^t, D^t\} - c \cdot x^t]$ is equivalent to minimizing the cost $f_t(\cdot)$ in Eq.~\eqref{Eq:RIG-cost}. 

In this context,~\citet{Huh-2009-Nonparametric} have shown that the cost function $f_t(\cdot)$ is convex in general and $\alpha p$-strongly convex if the demand is a random variable with a continuous density function $q$ such that $\inf_{d \in [0, \bar{x}]} q(d) \geq \alpha > 0$. Moreover, the retailer only observes the sales quantity $\min\{x^t, d^t\}$ where $d^t$ is a realization of $D^t$. Thus, the (noisy) bandit feedback is not observable. However, a noisy gradient feedback signal can be obtained:
\begin{equation*}
\xi^t = \begin{cases}
c, & \text{if } x^t \geq \min\{x^t, d^t\}, \\
c - p, & \text{otherwise},
\end{cases}
\end{equation*}
which is an unbiased and bounded gradient estimator. In this setting, the parameter $\alpha > 0$ is not available since the distribution of the demand is unknown. However, the retailer can apply our \textsf{AdaOGD} algorithm (cf. Algorithm~\ref{alg:AdaOGD-SA}) and obtain a near-optimal regret of $O(\log^2(T))$. 
\begin{algorithm}[!t]
\caption{\textsf{Newsvendor-AdaOGD}($x^1$, $T$)}\label{alg:AdaOGD-Newsvendor}
\begin{algorithmic}[1]
\STATE \textbf{Input:} initial point $x^1 \in \XCal$ and the total number of rounds $T$. 
\STATE \textbf{Initialization:} $p_0 = \frac{1}{\log(T+10)}$.  
\FOR{$t = 1, 2, \ldots, T$}
\STATE sample $M^t \sim \textsf{Geometric}(p_0)$. 
\STATE set $\eta^{t+1} \leftarrow \tfrac{t+1}{\sqrt{1 + \max\{M^1, \ldots, M^t\}}}$. 
\STATE observe the sales quantity of $\min\{x^t, d^t\}$. 
\STATE update $x^{t+1} \leftarrow \left\{\begin{array}{ll}
\argmin_{x \in [0, \bar{x}]}\{(x - x^t)c + \tfrac{\eta^{t+1}}{2}(x - x^t)^2\}, & \text{if } x^t \geq \min\{x^t, d^t\}, \\ \argmin_{x \in [0, \bar{x}]}\{(x - x^t)(c - p) + \tfrac{\eta^{t+1}}{2}(x - x^t)^2\}, & \text{otherwise}. 
\end{array} \right. $
\ENDFOR
\end{algorithmic}
\end{algorithm} 

We specialize Algorithm~\ref{alg:AdaOGD-SA} to the newsvendor problem in Algorithm~\ref{alg:AdaOGD-Newsvendor} and we present the corresponding result on the regret minimization property in the following corollary. 
\begin{corollary}\label{Thm:AdaOGD-Newsvendor-regret}
In the single-retailer newsvendor problem, the retailer's cost functions are defined by Eq.~\eqref{Eq:RIG-cost} where the unit purchase cost is $c > 0$ and the unit selling price is $p \geq c$. Also, the demand is a random variable with a continuous density function $q$ such that $\inf_{d \in [0, \bar{x}]} q(d) \geq \alpha > 0$. If the agent employs Algorithm~\ref{alg:AdaOGD-Newsvendor}, we have 
\begin{equation*}
\EE[\text{Regret}(T)] \leq \tfrac{\bar{x}^2}{2}(1 + e^{\frac{1}{(\alpha p)^2\log(T+10)}}) + \tfrac{p^2\log(T+1)}{2}\sqrt{1 + \log(T+10) + \log(T)\log(T+10)}. 
\end{equation*}
As a consequence, we have $\EE[\text{Regret}(T)] = O(\log^2(T))$. 
\end{corollary}
\begin{proof}
Recall that $f_t(x^t) = (p-c) \cdot \EE[\max\{0, D^t - x^t\}] + c \cdot \EE[\max\{0, x^t - D^t\}]$ and the noisy gradient feedback is given by
\begin{equation*}
\xi^t = \begin{cases}
c, & \text{if } x^t \geq \min\{x^t, d^t\}, \\
c - p, & \text{otherwise}. 
\end{cases}
\end{equation*}
So we have $\EE[\xi^t \mid x^t] = \nabla f_t(x^t)$ and $\EE[\|\xi^t\|^2 \mid x^t] \leq p^2$ for all $t \geq 1$. Since the demand has a continuous density function $q$ such that $\inf_{d \in [0, \bar{x}]} q(d) \geq \alpha > 0$, we have $f_t(\cdot)$ is $\alpha p$-strongly convex. In addition, $\XCal = [0, \bar{x}]$ implies that $D = \bar{x}$. Thus, Theorem~\ref{Thm:AdaOGD-regret} can be applied and implies the desired result. 
\end{proof}

%!TEX root = main.tex
\section{Feasible Multi-Agent Online Learning in Strongly Monotone Games}\label{sec:MA}
In this section, we consider feasible multi-agent learning in monotone games. Our main result is that if each agent applies \textsf{AdaOGD} in a strongly monotone game (the multi-agent generalization of strongly convex costs), the joint action of all agents converges in a last-iterate sense to the unique Nash equilibrium at a near-optimal rate. In contrast to previous work, our results provide the first feasible no-regret learning algorithm that is doubly optimal; in particular, in addition to not requiring any prior knowledge of the problem parameters, one does not need to adjust the step-size schedule based on whether an agent is in the single-agent setting or the multi-agent setting. It is important to note that these are two different merits, and our algorithm enjoys both of them.

\subsection{Basic Definitions and Notations}
We first review the definition of continuous games and consider a class of monotone games. In particular, we focus on continuous games played by a set of agents, $\NCal = \{1, 2, \ldots, N\}$. Each agent selects an \textit{action} $x_i$ from a convex and bounded $\XCal_i \subseteq \br^{d_i}$. The incurred cost for each agent is determined by the joint action $x = (x_i; x_{-i}) = (x_1, x_2, \ldots, x_N)$. We let $\|\cdot\|$ denote the Euclidean norm (Other norms can also be accommodated here and different $\XCal_i$'s can have different norms). 
\begin{definition}\label{def:game}
A continuous game is a tuple $\GCal=\{\NCal, \XCal=\Pi_{i=1}^N \XCal_i, \{u_i\}_{i=1}^N\}$, where $\NCal$ is a set of $N$ agents, $\XCal_i \subseteq \br^{d_i}$ is the $i^{\textnormal{th}}$ agent's action set that is both convex and bounded, and $u_i: \XCal \rightarrow \br$ is the $i^{\textnormal{th}}$ agent's cost function satisfying: (i) $u_i(x_i; x_{-i})$ is continuous in $x$ and continuously differentiable in $x_i$; (ii) $v_i(x)=\nabla_{x_i} u_i(x_i; x_{-i})$ is continuous in $x$. For simplicity, we denote $v(\cdot) = (v_1(\cdot), v_2(\cdot), \ldots, v_N(\cdot))$ as the joint profile of all agents' individual gradients.  
\end{definition}
We work with an analogous model of gradient feedback:
\begin{enumerate}
\item At each round $t$, an unbiased and bounded gradient is observed as a feedback signal. In particular, the observed noisy gradient $\xi^t$ satisfies $\EE[\xi^t \mid x^t] = v(x^t)$ and $\EE[\|\xi^t\|^2 \mid x^t] \leq G^2$ for all $t \geq 1$.
\item The action set $\XCal$ is bounded by a diameter $D > 0$, i.e., $\|x - x'\| \leq D$ for all $x, x' \in \XCal$.
\end{enumerate}
The study of monotone games dates to~\citet{Rosen-1965-Existence} who considered a class of games that satisfy the \textit{diagonal strict concavity} (DSC) condition.\footnote{This condition is equivalent to the notion of strict monotonicity in convex analysis~\citep{Bauschke-2011-Convex}; see~\citet{Facchinei-2007-Finite} for further discussion.} Further contributions appeared in~\citet{Sandholm-2015-Population} and~\citet{Sorin-2016-Finite}, where games that satisfy DSC are called ``contractive" and ``dissipative." In this context, a game $\GCal$ is \emph{monotone} if $(x'-x)^\top(v(x') - v(x)) \geq 0$ for any $x, x' \in \XCal$. Intuitively, the notion of monotonicity generalize the notion of convexity; indeed, the gradient operator of a convex function $f$ satisfies that $(x'-x)^\top(\nabla f(x') - \nabla f(x)) \geq 0$ for any $x, x' \in \XCal$. 

We now define the class of strongly monotone games:
\begin{definition}\label{def:SM-game}
A continuous game $\GCal$ is said to be \emph{$\beta$-strongly monotone} if we have $\langle x' - x, v(x') - v(x)\rangle \geq \beta\|x' - x\|^2$ for all $x, x' \in \XCal$. 
\end{definition}
A standard solution concept for non-cooperative games is the \textit{Nash equilibrium} (NE), where no agent has an incentive to deviate from her strategy~\citep{Osborne-1994-Course}. For the continuous games considered in this paper, we are interested in pure-strategy Nash equilibria since the randomness introduced by mixed strategies is unnecessary when each action lives in a continuum. 
\begin{definition}
An action profile $x^\star \in \XCal$ is called a \emph{Nash equilibrium} of $\GCal$ if it is resilient to unilateral deviations; that is, $u_i(x_i^\star; x_{-i}^\star) \leq u_i(x_i; x_{-i}^\star)$ for all $x_i \in \XCal_i$ and $i \in \NCal$. 
\end{definition}
\citet{Debreu-1952-Social} proved that any continuous game admits at least one Nash equilibrium if all action sets are convex and bounded, and all cost functions are individually convex (i.e., $u_i(x_i; x_{-i})$ is convex in $x_i$ for a fixed $x_{-i}$). Moreover, there is a variational characterization that forms the basis of equilibrium computation under an individual convexity condition~\citep{Facchinei-2007-Finite}. We summarize this characterization in the following proposition. 
\begin{proposition}\label{Prop:VC-games}
If all cost functions are in a continuous game $\GCal$ are individually convex, the joint action $x^\star \in \XCal$ is a Nash equilibrium if and only if $(x - x^\star)^\top v(x^\star) \geq 0$ for all $x \in \XCal$.
\end{proposition}
The notion of strong monotonicity arises in various application domains. Examples include strongly-convex-strongly-concave zero-sum games, atomic splittable congestion games in networks with parallel links~\citep{Orda-1993-Competitive, Sorin-2016-Finite, Mertikopoulos-2019-Learning}, wireless network optimization~\citep{Weeraddana-2012-Weighted, Tan-2014-Wireless, Zhou-2021-Robust} and classical online decision-making problems~\citep{Cesa-2006-Prediction}.

Strongly monotone games satisfy the individual convexity condition and hence the existence of at least one Nash equilibrium is ensured. Moreover, every strongly monotone game admits a unique Nash equilibrium~\citep{Zhou-2021-Robust}. Thus, one appealing feature of strongly monotone games is that  finite-time convergence can be derived in terms of $\|\hat{x} - x^\star\|^2$ where $x^\star \in \XCal$ is a unique Nash equilibrium. Despite some recent progress on last-iterate convergence rates for non-strongly monotone games~\citep{Lin-2020-Finite, Golowich-2020-Tight, Cai-2022-Finite}, last-iterate convergence rates in terms of $\|\hat{x} - x^\star\|^2$ are only available for strongly monotone games~\citep{Bravo-2018-Bandit, Zhou-2021-Robust, Lin-2021-Optimal}. An important gap in the literature is that there currently do not exist \textit{doubly optimal} and \textit{feasible} learning algorithms for strongly monotone games.

\subsection{Algorithmic Scheme}
We review the multi-agent OGD method that is a generalization of single-agent OGD. Letting $x_i^1 \in \XCal_i$ for all $i \in \NCal$, the multi-agent version of OGD (cf.\ Definition~\ref{def:SM-game}) performs the following step at each round: 
\begin{equation}\label{def:MA-OGD}
\eta_i^{t+1} \leftarrow \beta(t+1), \qquad x_i^{t+1} \leftarrow \argmin_{x_i \in \XCal_i}\{(x_i - x_i^t)^\top \xi_i^t + \tfrac{\eta_i^{t+1}}{2}\|x_i - x_i^t\|^2\}. 
\end{equation}
In the following theorem, we summarize the results from~\citet{Zhou-2021-Robust} on the optimal last-iterate convergence rate of multi-agent OGD in Eq.~\eqref{def:MA-OGD} using a squared Euclidean distance function. 
\begin{theorem}\label{Thm:OGD-rate}
Suppose that a continuous game $\GCal$ is $\beta$-strongly monotone and let $G, D > 0$ be problem parameters satisfying $\EE[\xi^t \mid x^t] = v(x^t)$ and let $\EE[\|\xi^t\|^2 \mid x^t] \leq G^2$ for all $t \geq 1$, and $\|x - x'\| \leq D$ for all $x, x' \in \XCal$. If all agents employ multi-agent OGD in Eq.~\eqref{def:MA-OGD}, we have
\begin{equation*}
\EE[\|x^T - x^\star\|^2] \leq \tfrac{4G^2}{\beta^2 T},
\end{equation*}
where $x^\star \in \XCal$ denotes the unique Nash equilibrium.
As a consequence, we have $\EE[\|x^T - x^\star\|^2] = O(\frac{1}{T})$. 
\end{theorem}
\begin{remark}
Theorem~\ref{Thm:OGD-rate} demonstrates that multi-agent OGD can achieve a near-optimal convergence rate in strongly monotone games; indeed, the convergence rate of multi-agent OGD matches the lower bound of $\Omega(\frac{1}{T})$ proved in~\citet{Nemirovski-1983-Problem} for strongly convex optimization. 
\end{remark}
A drawback of multi-agent OGD is that it requires  knowledge of the strong monotonicity parameters in order to update $\eta^{t+1}$ and is thus not feasible in practice. We are not aware of any research that addresses this key issue. This is possibly because existing adaptive techniques are specialized to \textit{upper curvature information}, e.g., the Lipschitz constant of function values or gradients~\citep{Duchi-2011-Adaptive, Kingma-2015-Adam, Mukkamala-2017-Variants, Levy-2017-Online, Levy-2018-Online, Bach-2019-Universal, Antonakopoulos-2021-Adaptive, Hsieh-2021-Adaptive}, and are not suitable for estimating \textit{lower curvature information} such as that encoded by the strong monotonicity parameter. The goal of this section is to extend Algorithm~\ref{alg:AdaOGD-SA} to multi-agent learning in games, showing that our adaptive variant of OGD is both doubly optimal and feasible.   

We again employ a randomization strategy such that the resulting algorithm is adaptive to the strong monotonicity parameter and other problem parameters. We again choose independently identical distributed geometric random variables, i.e., $M_i^t \sim \textsf{Geometric}(p_0)$, for $p_0 = \frac{1}{\log(T + 10)}$.\footnote{We can define agent-specific probabilities $p_0^i \in (0, 1)$ and prove the same finite-time convergence guarantee. For simplicity, we use an agent-independent probability $p_0 = \frac{1}{\log(T + 10)}$.} Note that all the other updates are analogous to those of Algorithm~\ref{alg:AdaOGD-SA} and the resulting algorithm is decentralized. This defines our adaptive multi-agent variant of OGD, as detailed in Algorithm~\ref{alg:AdaOGD-MA}. 
\begin{algorithm}[!t]
\caption{\textsf{MA-AdaOGD}($x_1^1, x_2^1, \ldots, x_N^1, T$)}\label{alg:AdaOGD-MA}
\begin{algorithmic}[1]
\STATE \textbf{Input:} initial points $x_i^1 \in \XCal_i$ for all $i \in \NCal$ and the total number of rounds $T$. 
\STATE \textbf{Initialization:} $p_0 = \frac{1}{\log(T+10)}$. 
\FOR{$t = 1, 2, \ldots, T$}
\FOR{$i = 1, 2, \ldots, N$}
\STATE sample $M_i^t \sim \textsf{Geometric}(p_0)$. 
\STATE set $\eta_i^{t+1} \leftarrow \tfrac{t + 1}{\sqrt{1 + \max\{M_i^1, \ldots, M_i^t\}}}$. 
\STATE update $x_i^{t+1} \leftarrow \argmin_{x_i \in \XCal_i}\{(x_i - x_i^t)^\top \xi_i^t + \tfrac{\eta_i^{t+1}}{2}\|x_i - x_i^t\|^2\}$. 
\ENDFOR
\ENDFOR
\end{algorithmic}
\end{algorithm}

\subsection{Finite-Time Last-Iterate Convergence Guarantee}
In the following theorem, we summarize our main results on the last-iterate convergence rate of Algorithm~\ref{alg:AdaOGD-MA} using a distance function based on squared Euclidean norm. 
\begin{theorem}\label{Thm:AdaOGD-rate}
Suppose that a continuous game $\GCal$ is $\beta$-strongly monotone and let $G, D > 0$ be problem parameters satisfying $\EE[\xi^t \mid x^t] = v(x^t)$ and let $\EE[\|\xi^t\|^2 \mid x^t] \leq G^2$ for all $t \geq 1$, and $\|x - x'\| \leq D$ for all $x, x' \in \XCal$. If all agents employ Algorithm~\ref{alg:AdaOGD-MA}, we have 
\begin{eqnarray*}
\EE[\|x^T - x^\star\|^2] & \leq & \tfrac{D^2}{T}(1 + e^{\frac{1}{4\beta^2\log(T+10)}})\sqrt{1 + \log(T+10) + \log(NT)\log(T+10)} \\
& & + \tfrac{G^2}{T}\log(T + 1)(1 + \log(T+10) + \log(NT)\log(T+10)),
\end{eqnarray*}
where $x^\star \in \XCal$ is the unique Nash equilibrium.
As a consequence, we have $\EE[\|x^T - x^\star\|^2] = O(\frac{\log^3(T)}{T})$. 
\end{theorem}
\begin{remark}
Theorem~\ref{Thm:AdaOGD-rate} demonstrates that Algorithm~\ref{alg:AdaOGD-MA} achieves a near-optimal convergence rate since the upper bound matches the lower bound~\citep{Nemirovski-1983-Problem} up to a log factor. This result together with Theorem~\ref{Thm:AdaOGD-regret} shows that our adaptive variant of OGD is a doubly optimal and feasible learning algorithm for strongly monotone games. 
\end{remark}
\begin{remark}[Importance of doubly optimality]
To further appreciate the concept of double optimality note that our starting point is single-agent that is learning in an online manner in an arbitrarily non-stationary and possibly adversarial environment. It is well known that online gradient descent achieves the minimax optimal regret bounds ($\Theta(\log) T$ in this setting, for strongly convex cost functions and $\Theta(\sqrt{T})$). As such, to an agent engaged in an online decision making process, it is natural to employ OGD in such an environment where statistical regularity may be lacking. In particular, the most common instantiation of such a non-stationary environment is one that comprises other agents who are simultaneously engaged in the online decision making process and whose actions impact all others' costs/rewards. In other words, each agent is acting in an environment whose cost/reward is determined by an \textit{opaque} game: the cost/reward is determined by an unknown underlying game, where  even the number of agents that comprise of the game may be unknown. Consequently, it is natural to ask if each agent adopts OGD to maximize its finite-horizon cumulative reward, would the system jointly converge to a Nash equilibrium, a multi-agent optimal outcome where no agent has any incentive to unilaterally deviate? If not, then in the long run an agent would be able to do better by deviating from what is prescribed by the online learning algorithm, given what all the other agents are doing, thereby producing ``regret.''  Consequently, an online learning algorithm that is doubly optimal---where a single agent adopting it attains optimal finite-time no-regret guarantees and when all agents adopt it they converge to a Nash equilibrium---effectively bridges optimal transient performance (i.e., finite-horizon performance) with optimal long-run performance (i.e., the equilibrium outcome). 
\end{remark}
To prove Theorem~\ref{Thm:AdaOGD-rate}, we provide another descent inequality for the iterates generated by Algorithm~\ref{alg:AdaOGD-MA}. 
\begin{lemma}\label{Lemma:AdaOGD-MA}
Suppose that a continuous game $\GCal$ is $\beta$-strongly monotone and let $G, D > 0$ be problem parameters satisfying $\EE[\xi^t \mid x^t] = v(x^t)$ and $\EE[\|\xi^t\|^2 \mid x^t] \leq G^2$ for all $t \geq 1$, and let $\|x - x'\| \leq D$ for all $x, x' \in \XCal$. Letting the iterates $\{x^t\}_{t \geq 1}$ be generated by Algorithm~\ref{alg:AdaOGD-MA}, we have
\begin{eqnarray*}
\lefteqn{\sum_{i = 1}^N \eta_i^T\EE\left[\|x_i^T - x_i^\star\|^2 \mid \{\eta_i^t\}_{1 \leq i \leq N, 1 \leq t \leq T}\right] \leq \sum_{i = 1}^N \eta_i^1\|x_i^1 - x_i^\star\|^2} \\ 
& & + D^2 \left(\sum_{t=1}^{T-1} \left(\max\left\{0, \max_{1 \leq i \leq N} \left\{\eta_i^{t+1} - \eta_i^t\right\} - 2\beta\right\}\right)\right) + G^2 \left(\sum_{t=1}^{T-1} \left(\max_{1 \leq i \leq N} \left\{\tfrac{1}{\eta_i^{t+1}}\right\}\right)\right),  
\end{eqnarray*}
where $x^\star \in \XCal$ is the unique Nash equilibrium.
\end{lemma}
We defer the proof of this lemma to Appendix~\ref{app:AdaOGD-MA}. 

\paragraph{Proof of Theorem~\ref{Thm:AdaOGD-rate}.} Since $D > 0$ satisfies that $\|x - x'\| \leq D$ for all $x, x' \in \XCal$ and $\eta_i^{t+1} = \tfrac{t+1}{\sqrt{1 + \max\{M_i^1, \ldots, M_i^t\}}}$ in Algorithm~\ref{alg:AdaOGD-MA}, we have
\begin{equation*}
\sum_{i = 1}^N \eta_i^1\|x_i^1 - x_i^\star\|^2 \leq D^2, \qquad \eta_i^{t+1} - \eta_i^t \leq \tfrac{1}{\sqrt{1 + \max\{M_i^1, \ldots, M_i^t\}}}.
\end{equation*}
By Lemma~\ref{Lemma:AdaOGD-MA}, we have
\begin{eqnarray}\label{inequality:AdaOGD-rate-first}
\lefteqn{\sum_{i = 1}^N \eta_i^T\EE\left[\|x_i^T - x_i^\star\|^2 \mid \{\eta_i^t\}_{1 \leq i \leq N, 1 \leq t \leq T}\right] \leq D^2} \\
& & + D^2 \left(\sum_{t=1}^{T-1} \left(\max\left\{0, \max_{1 \leq i \leq N} \left\{\tfrac{1}{\sqrt{1 + \max\{M_i^1, \ldots, M_i^t\}}}\right\} - 2\beta\right\}\right)\right) + G^2 \left(\sum_{t=1}^{T-1} \left(\max_{1 \leq i \leq N} \left\{\tfrac{1}{\eta_i^{t+1}}\right\}\right)\right). \nonumber 
\end{eqnarray}
Further, we have 
\begin{equation}\label{inequality:AdaOGD-rate-second}
\sum_{t=1}^{T-1} \left(\max_{1 \leq i \leq N} \left\{\tfrac{1}{\eta_i^{t+1}}\right\}\right) \leq \sqrt{1 + \max_{1 \leq i \leq N, 1 \leq t \leq T} \{M_i^t\}}\left(\sum_{t=1}^{T-1} \tfrac{1}{t+1}\right) \leq \log(T+1)\sqrt{1 + \max_{1 \leq i \leq N, 1 \leq t \leq T} \{M_i^t\}}. 
\end{equation}
Plugging Eq.~\eqref{inequality:AdaOGD-rate-second} into Eq.~\eqref{inequality:AdaOGD-rate-first} yields that 
\begin{eqnarray*}
\lefteqn{\sum_{i = 1}^N \eta_i^T\EE\left[\|x_i^T - x_i^\star\|^2 \mid \{\eta_i^t\}_{1 \leq i \leq N, 1 \leq t \leq T}\right] \leq D^2} \\
& & + D^2 \left(\sum_{t=1}^{T-1} \left(\max\left\{0, \max_{1 \leq i \leq N} \left\{\tfrac{1}{\sqrt{1 + \max\{M_i^1, \ldots, M_i^t\}}}\right\} - 2\beta\right\}\right)\right) + G^2\log(T+1)\sqrt{1 + \max_{1 \leq i \leq N, 1 \leq t \leq T} \{M_i^t\}}. 
\end{eqnarray*}
By the definition of $\eta_i^T$, we have
\begin{equation*}
\eta_i^T \geq \tfrac{T}{\sqrt{1 + \max\{M_i^1, \ldots, M_i^T\}}} \geq \tfrac{T}{\sqrt{1 + \max_{1 \leq i \leq N, 1 \leq t \leq T} \{M_i^t\}}}, 
\end{equation*}
which implies that 
\begin{equation*}
\sum_{i = 1}^N \eta_i^T\EE\left[\|x_i^T - x_i^\star\|^2 \mid \{\eta_i^t\}_{1 \leq i \leq N, 1 \leq t \leq T}\right] \geq \tfrac{T}{\sqrt{1 + \max_{1 \leq i \leq N, 1 \leq t \leq T} \{M_i^t\}}} \cdot \EE\left[\|x^T - x^\star\|^2 \mid \{\eta_i^t\}_{1 \leq i \leq N, 1 \leq t \leq T}\right]. 
\end{equation*}
Putting these pieces together yields that  
\begin{eqnarray*}
\lefteqn{\left(\tfrac{T}{\sqrt{1 + \max_{1 \leq i \leq N, 1 \leq t \leq T} \{M_i^t\}}}\right)\EE\left[\|x^T - x^\star\|^2 \mid \{\eta_i^t\}_{1 \leq i \leq N, 1 \leq t \leq T}\right] \leq D^2} \\ 
& & + D^2\left(\sum_{t=1}^{T-1} \left(\max\left\{0, \max_{1 \leq i \leq N} \left\{\tfrac{1}{\sqrt{1 + \max\{M_i^1, \ldots, M_i^t\}}}\right\} - 2\beta\right\}\right)\right) + G^2\log(T+1)\sqrt{1 + \max_{1 \leq i \leq N, 1 \leq t \leq T} \{M_i^t\}}. 
\end{eqnarray*}
Rearranging and taking the expectation of both sides, we have
\begin{eqnarray}\label{inequality:AdaOGD-rate-third}
\lefteqn{T \cdot \EE[\|x^T - x^\star\|^2] \leq D^2\underbrace{\EE\left[\sqrt{1 + \max_{1 \leq i \leq N, 1 \leq t \leq T} \{M_i^t\}}\right]}_{\textbf{I}} + G^2\log(T+1)\underbrace{\EE\left[1 + \max_{1 \leq i \leq N, 1 \leq t \leq T} \{M_i^t\}\right]}_{\textbf{II}}} \nonumber \\ 
& & + D^2\underbrace{\EE\left[\sum_{t = 1}^{T-1} \max\left\{0, \max_{1 \leq i \leq N} \left\{\tfrac{\sqrt{1 + \max_{1 \leq i \leq N, 1 \leq t \leq T} \{M_i^t\}}}{\sqrt{1 + \max\{M_i^1, \ldots, M_i^t\}}}\right\} - 2\beta\sqrt{1 + \max_{1 \leq i \leq N, 1 \leq t \leq T} \{M_i^t\}}\right\}\right]}_{\textbf{III}}.  
\end{eqnarray}
By using the previous argument with Proposition~\ref{Prop:GRV} (cf. Appendix~\ref{app:RV}) and $p = \frac{1}{\log(T+10)}$, we have
\begin{equation}\label{inequality:AdaOGD-rate-I}
\textbf{I} \leq \sqrt{1 + \tfrac{1 + \log(NT)}{p}} = \sqrt{1 + \log(T+10) + \log(NT)\log(T+10)}. 
\end{equation}
and 
\begin{equation}\label{inequality:AdaOGD-rate-II}
\textbf{II} \leq 1 + \tfrac{1 + \log(NT)}{p} = 1 + \log(T+10) + \log(NT)\log(T+10). 
\end{equation}
It remains to bound the term $\textbf{III}$ using Proposition~\ref{Prop:GRV} and $p = \frac{1}{\log(T+10)}$. Indeed, we have
\begin{eqnarray*}
\textbf{III} & = & \EE\left[\sum_{t = 1}^{T-1} \max\left\{0, \max_{1 \leq i \leq N} \left\{\tfrac{\sqrt{1 + \max_{1 \leq i \leq N, 1 \leq t \leq T} \{M_i^t\}}}{\sqrt{1 + \max\{M_i^1, \ldots, M_i^t\}}}\right\} - 2\beta\sqrt{1 + \max_{1 \leq i \leq N, 1 \leq t \leq T} \{M_i^t\}}\right\}\right] \\ 
& \leq & \EE\left[\sum_{t = 1}^{T-1} \max_{1 \leq i \leq N} \left\{\tfrac{\sqrt{1 + \max_{1 \leq i \leq N, 1 \leq t \leq T} \{M_i^t\}}}{\sqrt{1 + \max\{M_i^1, \ldots, M_i^t\}}}\right\} \cdot \mathbb{I}\left(\max_{1 \leq i \leq N} \left\{\tfrac{1}{\sqrt{1 + \max\{M_i^1, \ldots, M_i^t\}}}\right\} - 2\beta \geq 0\right)\right]. 
\end{eqnarray*}
Defining $i^t = \argmax_{1 \leq i \leq N} \left\{\tfrac{1}{\sqrt{1 + \max\{M_i^1, \ldots, M_i^t\}}}\right\}$ as a random variable and recalling that $\{\max\{M_i^1, \ldots, M_i^t\}\}_{1 \leq i \leq N}$ are independent and identically distributed, we have that $i^t \in \{1, \ldots, N\}$ is uniformly distributed. This implies that 
\begin{eqnarray*}
\textbf{III} & \leq & \tfrac{1}{N} \sum_{t = 1}^{T-1} \sum_{j=1}^N \EE\left[\max_{1 \leq i \leq N} \left\{\tfrac{\sqrt{1 + \max_{1 \leq i \leq N, 1 \leq t \leq T} \{M_i^t\}}}{\sqrt{1 + \max\{M_i^1, \ldots, M_i^t\}}}\right\} \cdot \mathbb{I}\left(\max_{1 \leq i \leq N} \left\{\tfrac{1}{\sqrt{1 + \max\{M_i^1, \ldots, M_i^t\}}}\right\} - 2\beta \geq 0\right) \mid i^t = j\right] \\
& = & \tfrac{1}{N} \sum_{t = 1}^{T-1} \sum_{j=1}^N \EE\left[\tfrac{\sqrt{1 + \max_{1 \leq i \leq N, 1 \leq t \leq T} \{M_i^t\}}}{\sqrt{1 + \max\{M_j^1, \ldots, M_j^t\}}} \cdot \mathbb{I}\left(\tfrac{1}{\sqrt{1 + \max\{M_j^1, \ldots, M_j^t\}}} - 2\beta \geq 0\right)\right] \\
& \leq & \tfrac{1}{N} \sum_{t = 1}^{T-1} \sum_{j=1}^N \EE\left[\sqrt{1 + \max_{1 \leq i \leq N, 1 \leq t \leq T, i \neq j} \{M_i^t\}} \cdot \mathbb{I}\left(\tfrac{1}{\sqrt{1 + \max\{M_j^1, \ldots, M_j^t\}}} - 2\beta \geq 0\right)\right]. 
\end{eqnarray*}
Since $\max_{1 \leq i \leq N, 1 \leq t \leq T, i \neq j} \{M_i^t\}$ is independent of $\max\{M_j^1, \ldots, M_j^t\}$, we have
\begin{equation*}
\textbf{III} \leq \tfrac{1}{N} \sum_{t = 1}^{T-1} \sum_{j=1}^N \EE\left[\sqrt{1 + \max_{1 \leq i \leq N, 1 \leq t \leq T, i \neq j} \{M_i^t\}}\right] \cdot \PP\left(\tfrac{1}{\sqrt{1 + \max\{M_j^1, \ldots, M_j^t\}}} - 2\beta \geq 0\right).
\end{equation*}
Since $\{M_i^t\}_{1 \leq i \leq N, 1 \leq t \leq T}$ is a sequence of independent and identically distributed geometric random variables with $p = \frac{1}{\log(T+10)}$, Proposition~\ref{Prop:GRV} (cf. Appendix~\ref{app:RV}) implies that
\begin{equation*}
\EE\left[\max_{1 \leq i \leq N, 1 \leq t \leq T, i \neq j} \{M_i^t\}\right] \leq \tfrac{1 + \log(NT)}{p} = \log(T+10) + \log(NT)\log(T+10). 
\end{equation*}
and 
\begin{equation*}
\sum_{t=1}^{T-1} \PP\left(\tfrac{1}{\sqrt{1 + \max\{M_j^1, \ldots, M_j^t\}}} - 2\beta \geq 0\right) \leq e^{\frac{p}{4\beta^2}} = e^{\frac{1}{4\beta^2\log(T+10)}}. 
\end{equation*}
Putting these pieces together with Jensen's inequality yields that 
\begin{equation}\label{inequality:AdaOGD-rate-III}
\textbf{III} \leq e^{\frac{1}{4\beta^2\log(T+10)}}\sqrt{1 + \log(T+10) + \log(NT)\log(T+10)}. 
\end{equation}
Plugging Eq.~\eqref{inequality:AdaOGD-rate-I}, Eq.~\eqref{inequality:AdaOGD-rate-II} and Eq.~\eqref{inequality:AdaOGD-rate-III} into Eq.~\eqref{inequality:AdaOGD-rate-third} yields that 
\begin{eqnarray*}
\EE[\|x^T - x^\star\|^2] & \leq & \tfrac{D^2}{T}(1 + e^{\frac{1}{4\beta^2\log(T+10)}})\sqrt{1 + \log(T+10) + \log(NT)\log(T+10)} \\
& & + \tfrac{G^2}{T}\log(T + 1)(1 + \log(T+10) + \log(NT)\log(T+10)). 
\end{eqnarray*}
This completes the proof. 

It is worth remarking that our proof does not use the structure of $v(\cdot)$ (i.e., $v_i(x) = \nabla_{x_i} u_i(x_i; x_{-i})$). Thus, Theorem~\ref{Thm:AdaOGD-rate}, when extended to the VI setting, yields a decentralized, feasible optimization algorithm for finding a solution of a strongly monotone VI, contributing to that line of literature. 
\begin{corollary}\label{Cor:AdaOGD-rate}
Suppose that the variational inequality defined by $v(\cdot)$ (without the structure that $v_i(x) = \nabla_{x_i} u_i(x_i; x_{-i})$) and $\XCal$ (any convex set) is $\beta$-strongly monotone and let $G, D > 0$ be problem parameters satisfying $\EE[\xi^t \mid x^t] = v(x^t)$ and $\EE[\|\xi^t\|^2 \mid x^t] \leq G^2$ for all $t \geq 1$, and $\|x - x'\| \leq D$ for all $x, x' \in \XCal$. Then, by employing Algorithm~\ref{alg:AdaOGD-MA}, we have 
\begin{eqnarray*}
\EE[\|x^T - x^\star\|^2] & \leq & \tfrac{D^2}{T}(1 + e^{\frac{1}{4\beta^2\log(T+10)}})\sqrt{1 + \log(T+10) + \log(NT)\log(T+10)} \\
& & + \tfrac{G^2}{T}\log(T + 1)(1 + \log(T+10) + \log(NT)\log(T+10)),
\end{eqnarray*}
where $x^\star \in \XCal$ is the unique solution of he VI. As a consequence, we have $\EE[\|x^T - x^\star\|^2] = O(\frac{\log^3(T)}{T})$. 
\end{corollary}

\subsection{Applications: Feasible Multi-Agent Learning for Power Management and Newsvendors with Lost Sales}\label{subsec:examples}
We present two typical examples that satisfy the conditions in Definition~\ref{def:SM-game}, where (noisy) gradient feedback is available even when (noisy) bandit feedback is not available. 
\begin{example}[Power Management in Wireless Networks]\label{Example:PM}
Target-rate power management problems are well known in operations research and wireless communications~\citep{Rappaport-2001-Wireless, Goldsmith-2005-Wireless}. We consider a wireless network of $N$ communication links, each link consisting of a transmitter (e.g., phone, tablet and sensor) and an intended receiver (each transmitter consumes power to send signals to their intended receivers). Assume further that the transmitter in the $i^\textnormal{th}$ link transmits with power $a_i \geq 0$ and let $a = (a_1, \ldots, a_N) \in \br^N$ denote the joint power profile of all transmitters in the network. 
In this context, the quality-of-service rate of the $i^\textnormal{th}$ link depends not only on how
much power its transmitter is employing but also on how much power all the other transmitters are concurrently employing. Formally, the $i^\textnormal{th}$ link's quality-of-service rate is given by $r_i(a) = \tfrac{G_{ii} a_i}{\sum_{j \neq i}G_{ij}a_j+\eta_i}$ where $\eta_i$ is the thermal noise associated with the receiver of the $i^\textnormal{th}$ link and $G_{ij} \geq 0$ is the unit power gain between the transmitter in the $j^\textnormal{th}$ link and the receiver in the $i^\textnormal{th}$ link, which is determined by the network topology but is unknown in practice. Note that $\sum_{j \neq i}G_{ij}a_j$ is the interference caused by other links to link $i$---all else being equal, the larger the powers of transmitters in other links, the lower the service rate the $i^\textnormal{th}$ link. Intuitively, the transmitter in the $i^\textnormal{th}$ link aims to balance between two objectives: maintaining a target service rate $r_i^\star$ while consuming as little power as possible. This consideration leads to the following standard cost function~\citep{Rappaport-2001-Wireless, Goldsmith-2005-Wireless}:
\begin{equation*}
u_i(a) = \tfrac{a_i^2}{2}\left(1 - \tfrac{r_i^\star}{r_i(a)}\right)^2 = \tfrac{1}{2}\left(a_i - \tfrac{r_i^\star(\sum_{j \neq i}G_{ij} a_j+\eta_i)}{G_{ii}}\right)^2. 
\end{equation*}
Notably, $u_i(a) = 0$ if the realized service rate $r_i(a)$ is equal to the target rate $r_i^\star$. Otherwise, there will be a cost either due to not meeting $r_i^\star$ or due to consuming unnecessary power. Our prior work has shown that this is a $\beta$-strongly monotone game~\citep{Zhou-2021-Robust} with $\beta = \lambda_{\min}(I - \frac{1}{2}(W + W^\top))> 0$, where $W_{ii} = 0$ for all $1 \leq i \leq N$ and $W_{ij} = \frac{r_i^\star G_{ij}}{G_{ii}}$ for $i \neq j$. 

Finally, gradient feedback is available for the aforementioned class of strongly monotone games; indeed, the $i^\textnormal{th}$ link's quality-of-service rate $r_i(\sa)$ is available to the transmitter in the $i^\textnormal{th}$ link, who can use it to compute the gradient of $u_i(\sa)$ with respect to $a_i$. In this setting, the parameter $\beta > 0$ is not available in practice since finding the minimal eigenvalue is computationally prohibitive. However, the retailer can apply our \textsf{MA-AdaOGD} algorithm (cf. Algorithm~\ref{alg:AdaOGD-MA}) and obtain a near-optimal last-iterate convergence rate of $O(\frac{\log^3(T)}{T})$. 
\end{example}

\begin{example}[Newsvendor with Lost Sales]\label{Example:RIG}
We consider the multi-retailer generalization of newsvendor problem~\citep{Netessine-2003-Centralized}. For simplicity, we focus on a single product with same per-unit price $p > 0$ and perishable inventory control. Given the set of retailers $\NCal = \{1, 2, \ldots, N\}$, the $i^\textnormal{th}$ retailer's action $x_i$ is assumed to lie in the interval $[0, \bar{x}_i]$. For the $i^\textnormal{th}$ retailer, the demand is random and depends on the inventory levels of other retailers, therefore we denote it as $D_i(x_{-i})$ and let $d_i \geq 0$ denote a realization of this random variable. The $i^\textnormal{th}$ retailer only observes the sales quantity $\min\{x_i, d_i\}$. Using the previous argument, the $i^\textnormal{th}$ retailer's cost function is defined by 
\begin{equation}\label{Eq:RIG-cost-MA}
u_i(x) = (p - c_i) \cdot \EE[\max\{0, D_i(x_{-i}) - x_i\}] + c_i \cdot \EE[\max\{0, x_i - D_i(x_{-i})\}], 
\end{equation}
and the noisy gradient feedback signal is given by 
\begin{equation*}
\xi_i = \begin{cases}
c_i, & \text{if } x_i \geq \min\{x_i, d_i\}, \\
c_i - p, & \text{otherwise}.
\end{cases}
\end{equation*}
It remains to extend the analysis from~\citet[Section~3.5]{Huh-2009-Nonparametric} and provide a condition on the distribution of $D_i(x_{-i})$ that can guarantee that the multi-retailer newsvendor problem is $\beta$-strongly monotone for some constant $\beta > 0$. Indeed, after simple calculations, we have
\begin{equation*}
v_i(x) = \nabla_{x_i} u_i(x) = p \cdot \PP(D_i(x_{-i}) \leq x_i) - p + c_i = p \cdot F_i(x) - p + c_i. 
\end{equation*}
Letting $F = (F_1, F_2, \ldots, F_N)$ be an operator from $\prod_{i=1}^N [0, \bar{x}_i]$ to $[0, 1]^N$, if $F$ is $\alpha$-strongly monotone, we have 
\begin{equation*}
\langle x' - x, v(x') - v(x)\rangle = p \cdot \langle x' - x, F(x') - F(x)\rangle \geq \alpha p\|x' - x\|^2, \quad \textnormal{for all } x, x' \in \prod_{i=1}^N [0, \bar{x}_i]. 
\end{equation*}
As a concrete example, we can let the distribution of a demand $D_i(x_{-i})$ be given by 
\begin{equation*}
\PP(D_i(x_{-i}) \leq z) = 1 - \tfrac{1 + \sum_{j \neq i} x_j}{(1 + z + \sum_{j \neq i} x_j)^2}. 
\end{equation*}
It is clear that $\lim_{z \rightarrow +\infty} \PP(D_i(x_{-i}) \leq z) = 1$ for all $x_{-i} \in \prod_{j \neq i} [0, \bar{x}_j]$. Also, we have 
\begin{equation*}
\PP(D_i(x_{-i}) = 0) = 1 - \tfrac{1}{1 + \sum_{j \neq i} x_j}, 
\end{equation*}
which characterizes the dependence of the distribution of the demand for the $i^\textnormal{th}$ retailer on other retailers' actions $x_{-i}$. Indeed, the demand $D_i(x_{-i})$ is likely to be small if the total inventory provided by other retailers $\sum_{j \neq i} x_j$ is large. Two extreme cases are (i) $\PP(D_i(x_{-i}) = 0) \rightarrow 1$ as $\sum_{j \neq i} x_j \rightarrow +\infty$ and (ii) $\PP(D_i(x_{-i}) = 0) = 0$ as $\sum_{j \neq i} x_j = 0$. Finally, we can prove that this is a $\beta$-strongly monotone game with $\beta = \frac{p}{(1 + \sum_{k = 1}^N \bar{x}_k)^3} > 0$; see Appendix~\ref{app:example} for the details. 

Thus, if all the retailers agree to apply our \textsf{MA-AdaOGD} algorithm (cf. Algorithm~\ref{alg:AdaOGD-MA}), we can obtain a near-optimal last-iterate convergence rate of $O(\frac{\log^3(T)}{T})$. 
\end{example}
\begin{algorithm}[!t]
\caption{\textsf{Newsvendor-MA-AdaOGD}($x_1^1, x_2^1, \ldots, x_N^1, T$)}\label{alg:AdaOGD-MA-Newsvendor}
\begin{algorithmic}[1]
\STATE \textbf{Input:} initial point $x_i^1 \in \XCal_i$ for all $i \in \NCal$ and the total number of rounds $T$. 
\STATE \textbf{Initialization:} $p_0 = \frac{1}{\log(T+10)}$. 
\FOR{$t = 1, 2, \ldots, T$}
\FOR{$i = 1, 2, \ldots, N$}
\STATE sample $M_i^t \sim \textsf{Geometric}(p_0)$. 
\STATE set $\eta_i^{t+1} \leftarrow \tfrac{t+1}{\sqrt{1 + \max\{M_i^1, \ldots, M_i^t\}}}$. 
\STATE observe the sales quantity of $\min\{x^t, d^t\}$. 
\STATE update $x_i^{t+1} \leftarrow \left\{\begin{array}{ll}
\argmin_{x_i \in [0, \bar{x}_i]}\{(x_i - x_i^t)c_i + \tfrac{\eta_i^{t+1}}{2}(x_i - x_i^t)^2\}, & \text{if } x_i^t \geq \min\{x_i^t, d_i^t\}, \\ \argmin_{x_i \in [0, \bar{x}_i]}\{(x_i - x_i^t)(c_i - p) + \tfrac{\eta_i^{t+1}}{2}(x_i - x_i^t)^2\}, & \text{otherwise}. 
\end{array} \right. $
\ENDFOR
\ENDFOR
\end{algorithmic}
\end{algorithm} 

We summarize Algorithm~\ref{alg:AdaOGD-MA} specialized to the multi-retailer generalization of the newsvendor problem in Algorithm~\ref{alg:AdaOGD-MA-Newsvendor} and we present the corresponding result on the last-iterate convergence rate guarantee in the following corollary. 
\begin{corollary}\label{Thm:AdaOGD-Newsvendor-rate}
In the multi-retailer generalization of the newsvendor problem, each retailer's cost functions are defined by Eq.~\eqref{Eq:RIG-cost-MA} where the unit purchase cost is $c_i > 0$ and the unit selling price is $p \geq c$. Also, the demand is a random variable with a continuous density function defined in Example~\ref{Example:RIG}. If the agent employs Algorithm~\ref{alg:AdaOGD-MA-Newsvendor}, we have 
\begin{eqnarray*}
\EE[\|x^T - x^\star\|^2] & \leq & \tfrac{\sum_{i=1}^N \bar{x}_i^2}{T}(1 + e^{\frac{1}{4\beta^2\log(T+10)}})\sqrt{1 + \log(T+10) + \log(NT)\log(T+10)} \\
& & + \tfrac{Np^2}{T}\log(T + 1)(1 + \log(T+10) + \log(NT)\log(T+10)),
\end{eqnarray*}
As a consequence, we have $\EE[\|x^T - x^\star\|^2] = O(\frac{\log^3(T)}{T})$. 
\end{corollary}
\begin{proof}
Recall that $u_i(x^t) = (p - c_i) \cdot \EE[\max\{0, D_i(x_{-i}^t) - x_i^t\}] + c_i \cdot \EE[\max\{0, x_i^t - D_i(x_{-i}^t)\}]$ and the noisy gradient feedback is given by
\begin{equation*}
\xi_i^t = \begin{cases}
c_i, & \text{if } x_i^t \geq \min\{x_i^t, d_i^t\}, \\
c_i - p, & \text{otherwise}. 
\end{cases}
\end{equation*}
So we have $\EE[\xi^t \mid x^t] = v(x^t)$ and $\EE[\|\xi^t\|^2 \mid x^t] \leq Np^2$ for all $t \geq 1$. In Appendix~\ref{app:example}, we show that this is a $\beta$-strongly monotone game with $\beta = \frac{p}{(1 + \sum_{k = 1}^N \bar{x}_k)^3} > 0$. In addition, $\XCal_i = [0, \bar{x}_i]$ implies that $D^2 = \sum_{i=1}^N \bar{x}_i^2$. Thus, Theorem~\ref{Thm:AdaOGD-rate} can be applied and implies the desired result. 
\end{proof}

%!TEX root = paper.tex
\section{Extensions to Exp-Concave Cost Functions and Games}
In this section, we extend our results---for both single-agent and multi-agent learning---to a broader class of cost functions and game structures defined in terms of exp-concavity. For exp-concave cost functions, our adaptive variant of an online Newton step (\textsf{AdaONS}) can achieve a near-optimal regret of $O(d\log^2(T))$. For exp-concave games, we propose a multi-agent online Newton step with a near-optimal time-average convergence rate of $O(\frac{d\log(T)}{T})$. We also extend it to a multi-agent adaptive online Newton step (\textsf{MA-AdaONS}) that achieves a near-optimal rate of $O(\frac{d\log^2(T)}{T})$. 

\subsection{Single-Agent Learning with Exp-Concave Cost}
We start by recalling the definition of exp-concave functions. 
\begin{definition}\label{def:EC-function}
A function $f: \br^d \mapsto \br$ is $\alpha$-exp-concave if $e^{-\alpha f(\cdot)}$ is concave. 
\end{definition}
Exp-Concave (EC) functions, as a strict subclass of convex functions and as a strict superclass of strongly convex functions, have found applications in many fields, including information theory~\citep{Cover-1999-Elements}, stochastic portfolio theory~\citep{Fernholz-2002-Stochastic}, optimal transport~\citep{Villani-2009-Optimal}, optimization~\citep{Hazan-2014-Logistic, Mahdavi-2015-Lower}, probability theory~\citep{Pal-2018-Exponentially} and statistics~\citep{Juditsky-2008-Learning, Koren-2015-Fast,Yang-2018-Simple}. 

We work with the following model of perfect gradient feedback:\footnote{Can we generalize the aforementioned results to the setting with noisy gradient feedback? We leave the answer to this question to future work.}:
\begin{enumerate}
\item At each round $t$, an exact gradient is observed. That is, we have $\nabla f_t(x^t)$ for all $t \geq 1$.
\item The action set $\XCal$ is bounded by a diameter $D > 0$, i.e., $\|x - x'\| \leq D$ for all $x, x' \in \XCal$.
\end{enumerate}
For this class of cost functions, a lower bound on regret of $\Omega(d\log(T))$ is known \citep[see Section~3 in][]{Abernethy-2008-Optimal}. Two typical examples which admit EC cost functions are online linear regression~\citep{Kivinen-1999-Averaging} and universal portfolio management~\citep{Cover-1991-Universal}. 
\begin{enumerate}
\item \textbf{In universal portfolio management}, we have $f_t(x) = -\log(a_t^\top x)$ where $x \in \Delta^d$ is a probability vector and $a^t \in \br_+^d$ is the $d$-dimensional vector of all $d$ assets' relative prices between period $t$ and period $t-1$. Here, $f_t$ is EC and~\citet{Ordentlich-1998-Cost} established a lower bound of $\Omega(d\log(T))$ for the cumulative regret.
\item \textbf{In online linear regression}, we have  $f_t(x) = (a_t^\top x - b_t)^2$, which is in general not a strongly convex function, but is indeed EC for any $a_t, b_t$. Formally,~\citet{Vovk-1997-Competitive} proved that there exists a randomized strategy of the adversary for choosing the vectors $a_t, b_t$ such that the expected cumulative regret scales as $\Omega(d\log(T))$.
\end{enumerate}
This lower bound was established without making the connectio to the general EC function class. Note also that the $\Omega(d\log(T))$ lower bound is interesting given that $d \geq 1$ does not enter the lower bound for strongly convex functions~\citep{Hazan-2014-Beyond}.

On the upper-bound side, there are many algorithms that achieve the minimax-optimal regret, such as the online Newton step (ONS) of~\citet{Hazan-2007-Logarithmic}. Formally, let $\{f_t\}_{t \geq 1}$ be $\alpha$-exp-concave cost functions and let $G > 0$ be problem parameters satisfying $\|\nabla f_t(x)\| \leq G$ for all $t \geq 1$. Then, the ONS algorithm is given by
\begin{equation}\label{def:ONS}
\begin{array}{ll}
& \eta \leftarrow \tfrac{1}{2}\min\left\{\tfrac{1}{4GD}, \alpha\right\}, \qquad A^{t+1} \leftarrow A^t + \nabla f_t(x^t)\nabla f_t(x^t)^\top,  \\
& x^{t+1} \leftarrow \argmin_{x \in \XCal}\{(x - x^t)^\top \nabla f_t(x^t) + \tfrac{\eta}{2}(x - x^t)^\top A^{t+1}(x - x^t)\}.
\end{array}
\end{equation}
The choice $\frac{1}{2}\min\{\tfrac{1}{4GD}, \alpha\}$ comes from~\citet[Lemma~3]{Hazan-2007-Logarithmic}: if $\|\nabla f(x)\| \leq G$ and $\|x - x'\| \leq D$ for all $x, x' \in \XCal$, the $\alpha$-exp-concave function $f$ satisfies 
\begin{equation}\label{def:EC-inequality}
f(x') \geq f(x) + (x' - x)^\top\nabla f(x) + \tfrac{1}{4}\min\{\tfrac{1}{4GD}, \alpha\}(x' - x)^\top(\nabla f(x)\nabla f(x)^\top)(x' - x). 
\end{equation}
Intuitively, this equation implies that any $\alpha$-exp-concave function can be approximated by a local quadratic function with a matrix $\eta \nabla f(x)\nabla f(x)^\top$ and moreover ONS can exploit such structure. Although ONS suffers from a quadratic dependence on the dimension in terms of per-iteration cost, there has been progress in alleviating this complexity. For example,~\citet{Luo-2016-Efficient} proposed a variant of ONS that attains linear dependence on the dimension using sketching techniques. 
\begin{algorithm}[!t]
\caption{\textsf{AdaONS}($x^1$, $T$)}\label{alg:AdaONS-SA}
\begin{algorithmic}[1]
\STATE \textbf{Input:} initial point $x^1 \in \XCal$ and the total number of rounds $T$. 
\STATE \textbf{Initialization:} $A^1 = I_d$ where $I_d \in \br^{d \times d}$ is an identity matrix and $p_0 = \frac{1}{\log(T+10)}$. 
\FOR{$t = 1, 2, \ldots, T$}
\STATE sample $M^t \sim \textsf{Geometric}(p_0)$. 
\STATE set $\eta^{t+1} \leftarrow \tfrac{1}{\sqrt{1 + \max\{M^1, \ldots, M^t\}}}$. 
\STATE update $A^{t+1} \leftarrow A^t + \nabla f_t(x^t)\nabla f_t(x^t)^\top$. 
\STATE update $x^{t+1} \leftarrow \argmin_{x \in \XCal}\{(x - x^t)^\top \nabla f_t(x^t) + \tfrac{\eta^{t+1}}{2}(x - x^t)^\top A^{t+1} (x - x^t)\}$.  
\ENDFOR
\end{algorithmic}
\end{algorithm}

\subsection{Feasible Single-Agent Online Learning under Exp-Concave Cost}
The scheme of ONS in Eq.~\eqref{def:ONS} is infeasible in practice since it requires prior knowledge of the problem parameters $\alpha$, $G$ and $D$. To address this issue, we invoke the same randomization strategy as described in Algorithm~\ref{alg:AdaOGD-SA}. This results in the single-agent adaptive variant of ONS (cf. Algorithm~\ref{alg:AdaONS-SA}). 

We summarize our main results on the regret minimization properties of the algorithm in the following theorem. 
\begin{theorem}\label{Thm:AdaONS-regret}
Consider an arbitrary fixed sequence of $\alpha$-exp-concave functions $f_1, \ldots, f_T$, where each $f_t$ satisfies $\|\nabla f_t(x)\| \leq G$ for all $t \geq 1$ and $\|x - x'\| \leq D$ for all $x, x' \in \XCal$. If the agent employs Algorithm~\ref{alg:AdaONS-SA}, we have 
\begin{equation*}
\EE[\text{Regret}(T)] \leq \tfrac{D^2}{2}(1 + G^2 e^{\frac{(\max\{8GD, 2\alpha^{-1}\})^2}{\log(T+10)}}) + \tfrac{d\log(TG^2+1)}{2}\sqrt{1 + \log(T+10) + \log(T)\log(T+10)}. 
\end{equation*}
Thus, we have $\EE[\text{Regret}(T)] = O(d\log^2(T))$. 
\end{theorem}
\begin{remark}
Theorem~\ref{Thm:AdaONS-regret} demonstrates that Algorithm~\ref{alg:AdaONS-SA} achieves a near-optimal regret since the upper bound matches the lower bound up to a log factor; indeed, the lower bound of $\Omega(d\log(T))$ has been established for online linear regression and universal portfolio management. Further, Algorithm~\ref{alg:AdaONS-SA} dynamically adjusts $\eta^{t+1}$ without any prior knowledge of problem parameters. 
\end{remark}
To prove Theorem~\ref{Thm:AdaONS-regret}, we again make use of a key descent inequality. 
\begin{lemma}\label{Lemma:AdaONS}
Consider an arbitrary fixed sequence of $\alpha$-exp-concave functions $f_1, \ldots, f_T$, where each $f_t$ satisfies $\|\nabla f_t(x)\| \leq G$ for all $t \geq 1$ and $\|x - x'\| \leq D$ for all $x, x' \in \XCal$. Letting the iterates $\{x^t\}_{t \geq 1}$ be generated by Algorithm~\ref{alg:AdaONS-SA}, we have
\begin{eqnarray*}
\lefteqn{\sum_{t=1}^T f_t(x^t) - \sum_{t=1}^T f_t(x) \leq \tfrac{\eta^1}{2}(x^1 - x)^\top A^1(x^1 - x) + \tfrac{1}{2}\left(\sum_{t=1}^T \tfrac{1}{\eta^{t+1}}\nabla f_t(x^t)^\top (A^{t+1})^{-1}\nabla f_t(x^t)\right)} \\ 
& & + \sum_{t=1}^T (x^t - x)^\top \left(\tfrac{\eta^{t+1}}{2} A^{t+1} - \tfrac{\eta^t}{2}A^t - \tfrac{1}{4}\min\{\tfrac{1}{4GD}, \alpha\}\nabla f_t(x^t)\nabla f_t(x^t)^\top\right)(x^t - x), \quad \textnormal{for all } x \in \XCal. 
\end{eqnarray*}
\end{lemma}
The proofs of Lemma~\ref{Lemma:AdaONS} and Theorem~\ref{Thm:AdaONS-regret} are given in Appendix~\ref{app:AdaONS} and~\ref{app:AdaONS-regret}.

\subsection{Exp-Concave (EC) Games}
We define a class of exp-concave (EC) games via a game-theoretic generalization of the optimization framework in the previous section.  Letting $f$ be $\alpha$-exp-concave with $\|\nabla f(x)\| \leq G$ and $\|x - x'\| \leq D$ for all $x, x' \in \XCal$,~\citet[Lemma~3]{Hazan-2007-Logarithmic} guarantees that 
\begin{equation*}
f(x') \geq f(x) + (x' - x)^\top\nabla f(x) + \tfrac{1}{4}\min\{\tfrac{1}{4GD}, \alpha\}(x' - x)^\top(\nabla f(x)\nabla f(x)^\top)(x' - x). 
\end{equation*}
Equivalently, we have
\begin{equation*}
(x' - x)^\top(\nabla f(x') - \nabla f(x)) \geq \tfrac{1}{4}\min\{\tfrac{1}{4GD}, \alpha\}(x' - x)^\top(\nabla f(x')\nabla f(x')^\top + \nabla f(x)\nabla f(x)^\top)(x' - x). 
\end{equation*}
This leads to the following formal definition in which $v(\cdot)$ replaces $\nabla f(\cdot)$.   
\begin{definition}\label{def:EC-game}
A continuous game $\GCal$ is said to be $\alpha$-exp-concave if we have $\langle x' - x, v(x') - v(x)\rangle \geq \tfrac{1}{4}\min\{\tfrac{1}{4GD}, \alpha\}(\sum_{i = 1}^N (x'_i - x_i)^\top(v_i(x')v_i(x')^\top + v_i(x)v_i(x)^\top)(x'_i - x_i))$ for all $x, x' \in \XCal$. 
\end{definition}
Since EC games satisfy the individual convexity condition, we have the existence of at least one Nash equilibrium. However, multiple Nash equilibria might exist for EC games. For example, letting $d \geq 2$, we consider a single-agent game with the cost function $f(x) =\frac{1}{2}(a^\top x)^2$ and $\XCal = \{x \in \br^d: \|x\| \leq 1\}$. It is easy to verify that this is a $(2/\|a\|^2)$-exp-concave game and any $x \in \XCal$ satisfying $a^\top x = 0$ will be a Nash equilibrium. It is thus natural to ask how to measure the optimality of a point $\hat{x} \in \XCal$. Proposition~\ref{Prop:VC-games} inspires us to use a gap function, $\textsc{gap}(\cdot): \XCal \mapsto \br_+$, given by 
\begin{equation}\label{def:gap}
\textsc{gap}(\hat{x}) = \sup_{x \in \XCal} \ (\hat{x} - x)^\top v(x). 
\end{equation}
The above gap function is well defined and is nonnegative for all $\hat{x} \in \XCal$, given that at least one Nash equilibrium exists in EC games. Note that such a function has long been a standard example in the optimization literature~\citep{Facchinei-2007-Finite, Mertikopoulos-2019-Learning}. 

All strongly monotone games are clearly EC games if $\|v(x)\| \leq G$ for all $x \in \XCal$. Indeed, we have $\sum_{i = 1}^N (x'_i - x_i)^\top(v_i(x')v_i(x')^\top + v_i(x)v_i(x)^\top)(x'_i - x_i) \leq 2G^2\|x' - x\|^2$ and
\begin{equation*}
\langle x' - x, v(x') - v(x)\rangle \geq \beta\|x' - x\|^2 \geq \tfrac{\beta}{2G^2}\left(\sum_{i = 1}^N (x'_i - x_i)^\top(v_i(x')v_i(x')^\top + v_i(x)v_i(x)^\top)(x'_i - x_i)\right). 
\end{equation*}
This implies that a $\beta$-strongly monotone game is $\alpha$-exp-concave if we have $\alpha \leq \frac{2\beta}{G^2}$. More generally, we remark that strongly monotone games are a rich class of games that arise in many real-world application problems~\citep{Bravo-2018-Bandit, Mertikopoulos-2019-Learning, Lin-2021-Optimal}. Moreover, there are also applications that can be cast as EC games rather than strongly monotone games, such as empirical risk minimization~\citep{Koren-2015-Fast,Yang-2018-Simple} and universal portfolio management~\citep{Cover-1991-Universal, Fernholz-2002-Stochastic}. 

\subsection{Multi-Agent Online Learning in EC Games}
We start by extending the single-agent ONS algorithm in Eq.~\eqref{def:ONS} to the multi-agent online learning in EC games (cf.\ Definition~\ref{def:EC-game}). We again work with a model of exact gradient feedback:
\begin{enumerate}
\item At each round $t$, an exact gradient is observed. That is, we have $v(x^t)$ for all $t \geq 1$.
\item The action set $\XCal = \Pi_{i=1}^N \XCal_i$ is bounded by a diameter $D > 0$, i.e., $\|x - x'\| \leq D$ for all $x, x' \in \XCal$.
\end{enumerate}
Letting $x_i^1 \in \XCal_i$ for all $i \in \NCal$, the multi-agent version of ONS performs the following step at each round ($A_i^1 = I_{d_i}$ for all $i \in \NCal$): 
\begin{equation}\label{def:MA-ONS}
\begin{array}{ll}
& \eta_i \leftarrow \tfrac{1}{2}\min\left\{\tfrac{1}{4GD}, \alpha\right\}, \quad A_i^{t+1} \leftarrow A_i^t + v_i(x^t)v_i(x^t)^\top, \\
& x_i^{t+1} \leftarrow \argmin_{x_i \in \XCal_i}\{(x_i - x_i^t)^\top v_i(x^t) + \tfrac{\eta_i}{2}(x_i - x_i^t)^\top A_i^{t+1}(x_i - x_i^t)\}. 
\end{array}
\end{equation}
We see from Eq.~\eqref{def:MA-ONS} that our multi-agent learning algorithm is a generalization of single-agent ONS in Eq.~\eqref{def:ONS}. Notably, it is a decentralized algorithm since each agent does not need to know any other agents' gradients. In the following theorem, we summarize our main results on the time-average convergence rate using the gap function in Eq.~\eqref{def:gap}. 
\begin{theorem}\label{Thm:ONS-rate}
Suppose that a continuous game $\GCal$ is $\alpha$-exp-concave and let $G, D > 0$ be problem parameters satisfying $\|v(x)\| \leq G$ and $\|x - x'\| \leq D$ for all $x, x' \in \XCal$. If all agents agree to employ multi-agent ONS in Eq.~\eqref{def:MA-ONS}, we have 
\begin{equation*}
\textsc{gap}(\bar{x}^T) \leq \tfrac{\alpha D^2}{4T} + \tfrac{d\log(TG^2+1)}{T}\max\left\{4GD, \tfrac{1}{\alpha}\right\},
\end{equation*}
where $\bar{x}^T = \frac{1}{T}\sum_{t=1}^T x^t$ denotes a time-average iterate. Thus, we have $\textsc{gap}(\bar{x}^T) = O(\frac{d\log(T)}{T})$. 
\end{theorem}
\begin{remark}
Theorem~\ref{Thm:ONS-rate} shows that the multi-agent ONS algorithm can achieve a near-optimal convergence rate in EC games; indeed, the convergence rate of multi-agent ONS matches up to a log factor the lower bound of $\Omega(\frac{d}{T})$ proved in~\citet{Mahdavi-2015-Lower} for exp-concave optimization. 
\end{remark}
The proof of Theorem~\ref{Thm:ONS-rate} is again based on a descent inequality. 
\begin{lemma}\label{Lemma:AdaONS-MA}
Suppose that a continuous game $\GCal$ is $\alpha$-exp-concave and let $G, D > 0$ be problem parameters satisfying $\|v(x)\| \leq G$ and $\|x - x'\| \leq D$ for all $x, x' \in \XCal$. Letting the iterates $\{x^t\}_{t \geq 1}$ be generated by the multi-agent ONS in Eq.~\eqref{def:MA-ONS}, we have
\begin{eqnarray*}
\lefteqn{\sum_{t = 1}^T (x^t - x)^\top v(x) \leq \sum_{i=1}^N \tfrac{\eta_i}{2} (x_i^1 - x_i)^\top A_i^1(x_i^1 - x_i) + \tfrac{1}{2}\left(\sum_{t = 1}^T \sum_{i = 1}^N \tfrac{1}{\eta_i}v_i(x^t)^\top (A_i^{t+1})^{-1}v_i(x^t)\right)} \\ 
& & + \sum_{t = 1}^T \sum_{i=1}^N (x_i^t - x_i)^\top\left(\tfrac{\eta_i}{2}A_i^{t+1} - \tfrac{\eta_i}{2}A_i^t - \tfrac{1}{4}\min\{\tfrac{1}{4GD}, \alpha\} v_i(x^t)v_i(x^t)^\top\right)(x_i^t - x_i). 
\end{eqnarray*}
\end{lemma}
See Appendices~\ref{app:AdaONS-MA} and~\ref{app:ONS-rate} for the proofs of Lemma~\ref{Lemma:AdaONS-MA} and Theorem~\ref{Thm:ONS-rate}. 
\begin{algorithm}[!t]
\caption{\textsf{MA-AdaONS}($x_1^1, x_2^1, \ldots, x_N^1, T$)}\label{alg:AdaONS-MA}
\begin{algorithmic}[1]
\STATE \textbf{Input:} initial points $x_i^1 \in \XCal_i$ for all $i \in \NCal$ and the total number of rounds $T$. 
\STATE \textbf{Initialization:} $A_i^1 = I_{d_i}$ where $I_{d_i} \in \br^{d_i \times d_i}$ is an identity matrix and $p_0 = \frac{1}{\log(T+10)}$. 
\FOR{$t = 1, 2, \ldots, T$}
\FOR{$i = 1, 2, \ldots, N$}
\STATE sample $M_i^t \sim \textsf{Geometric}(p_0)$. 
\STATE set $\eta_i^{t+1} \leftarrow \tfrac{1}{\sqrt{1 + \max\{M_i^1, \ldots, M_i^t\}}}$. 
\STATE update $A_i^{t+1} \leftarrow A_i^t +v_i(x^t) v_i(x^t)^\top$. 
\STATE update $x_i^{t+1} \leftarrow \argmin_{x_i \in \XCal_i}\{(x_i - x_i^t)^\top v_i(x^t) + \tfrac{\eta_i^{t+1}}{2}(x_i - x_i^t)^\top A_i^{t+1}(x_i - x_i^t)\}$. 
\ENDFOR
\ENDFOR
\end{algorithmic}
\end{algorithm}
\subsection{Feasible Multi-Agent Online Learning in EC Games}
We extend Algorithm~\ref{alg:AdaONS-SA} to multi-agent learning in EC games, showing that our adaptive variant of ONS is not only feasible but achieves a near-optimal convergence rate in terms of a gap function. 

By employing a randomization strategy, we obtain an algorithm that is adaptive to exp-concavity parameter and other problem parameters. All the other updates are analogous to Algorithm~\ref{alg:AdaONS-SA} and the overall multi-agent framework, shown in Algorithm~\ref{alg:AdaONS-MA}, is decentralized.  we summarize our main results on the time-average convergence rate of Algorithm~\ref{alg:AdaONS-MA} using the gap function in Eq.~\eqref{def:gap}. 
\begin{theorem}\label{Thm:AdaONS-rate}
Suppose that a continuous game $\GCal$ is $\alpha$-exp-concave and let $G, D > 0$ be problem parameters satisfying $\|v(x)\| \leq G$ and $\|x - x'\| \leq D$ for all $x, x' \in \XCal$. If all agents agree to employ Algorithm~\ref{alg:AdaONS-MA}, we have 
\begin{equation*}
\EE[\textsc{gap}(\bar{x}^T)] \leq \tfrac{D^2}{2T}(1 + G^2 e^{\frac{(\max\{8GD, 2\alpha^{-1}\})^2}{\log(T+10)}}) + \tfrac{d\log(TG^2+1)}{2T}\sqrt{1 + \log(T+10) + \log(T)\log(T+10)},
\end{equation*}
where $\bar{x}^T = \frac{1}{T}\sum_{t=1}^T x^t$ denotes a time-average iterate.
Thus, we have $\EE[\textsc{gap}(\bar{x}^T)] = O(\frac{d\log^2(T)}{T})$. 
\end{theorem}
Directly extending Lemma~\ref{Lemma:AdaONS-MA} to multi-agent learning in EC games, we prove that the iterates $\{x^t\}_{t \geq 1}$ generated by Algorithm~\ref{alg:AdaONS-MA} satisfy  
\begin{eqnarray}\label{inequality:AdaONS-MA-key}
\lefteqn{\sum_{t = 1}^T (x^t - x)^\top v(x) \leq \sum_{i=1}^N \tfrac{\eta_i^1}{2} (x_i^1 - x_i)^\top A_i^1(x_i^1 - x_i) + \tfrac{1}{2}\left(\sum_{t = 1}^T \sum_{i = 1}^N \tfrac{1}{\eta_i^{t+1}}v_i(x^t)^\top (A_i^{t+1})^{-1}v_i(x^t)\right)} \nonumber \\ 
& & + \sum_{t = 1}^T \sum_{i=1}^N (x_i^t - x_i)^\top\left(\tfrac{\eta_i^{t+1}}{2}A_i^{t+1} - \tfrac{\eta_i^t}{2}A_i^t - \tfrac{1}{4}\min\{\tfrac{1}{4GD}, \alpha\} v_i(x^t)v_i(x^t)^\top\right)(x_i^t - x_i). 
\end{eqnarray}
This inequality is crucial to the proof of Theorem~\ref{Thm:AdaONS-rate}, which can be found in Appendix~\ref{app:AdaONS-rate}. 

%!TEX root = paper.tex
\section{Discussion}\label{sec:conclu}
Our results open up several directions for further research. First, we have assumed that the gradient feedback is always received at the end of each period. In practice, there may be delays. For instance, in the  power-control example, the signal-to-noise ratio sent by the receiver to the transmitter through the feedback channel, from which the gradient can be recovered, is often received with a delay. As such, it is important to understand how delays impact the performance as well as to design the delay-adaptive algorithms that can operate robustly in such environments. Second, our paper has focused on (noisy) gradient feedback while another important yet more challenging type of feedback is bandit feedback, where we only observe (noisy) function values at the chosen action. This problem domain has been less explored than that of learning with (noisy) gradient feedback. For instance, it remains unknown what the minimax optimal regret bound for convex cost functions is; the optimal dependence on $T$ is $\sqrt{T}$~\citep{Bubeck-2021-Kernel} but the optimal dependence on the dimension $d$ is unknown (note that the OGD regret is dimension-independent if gradient feedback can be observed). Thus, it is desirable to develop the feasible and optimal bandit learning algorithms in both single-agent and multi-agent settings. Can our randomization technique be applicable in that setting and improve the existing algorithms~\citep{Bravo-2018-Bandit, Lin-2021-Optimal}? This is a natural question for future work.

% Acknowledgements should only appear in the accepted version.
\section*{Acknowledgments}
This work was supported in part by the Mathematical Data Science program of the Office of Naval Research under grant number N00014-18-1-2764 and by the Vannevar Bush Faculty Fellowship program
under grant number N00014-21-1-2941.

\bibliographystyle{plainnat}
\bibliography{ref}

\begin{thebibliography}{128}
\providecommand{\natexlab}[1]{#1}
\providecommand{\url}[1]{\texttt{#1}}
\expandafter\ifx\csname urlstyle\endcsname\relax
  \providecommand{\doi}[1]{doi: #1}\else
  \providecommand{\doi}{doi: \begingroup \urlstyle{rm}\Url}\fi

\bibitem[Abernethy et~al.(2008)Abernethy, Bartlett, Rakhlin, and
  Tewari]{Abernethy-2008-Optimal}
J.~Abernethy, P.~L. Bartlett, A.~Rakhlin, and A.~Tewari.
\newblock Optimal strategies and minimax lower bounds for online convex games.
\newblock In \emph{COLT}, page 414–424. Omnipress, 2008.

\bibitem[Alacaoglu and Malitsky(2022)]{Alacaoglu-2022-Stochastic}
A.~Alacaoglu and Y.~Malitsky.
\newblock Stochastic variance reduction for variational inequality methods.
\newblock In \emph{COLT}, pages 778--816. PMLR, 2022.

\bibitem[Antonakopoulos et~al.(2021)Antonakopoulos, Belmega, and
  Mertikopoulos]{Antonakopoulos-2021-Adaptive}
K.~Antonakopoulos, V.~Belmega, and P.~Mertikopoulos.
\newblock Adaptive extra-gradient methods for min-max optimization and games.
\newblock In \emph{ICLR}, pages 1--12, 2021.
\newblock URL \url{https://openreview.net/forum?id=R0a0kFI3dJx}.

\bibitem[Baby and Wang(2021)]{Baby-2021-Optimal}
D.~Baby and Y-X. Wang.
\newblock Optimal dynamic regret in exp-concave online learning.
\newblock In \emph{COLT}, pages 359--409. PMLR, 2021.

\bibitem[Baby and Wang(2022)]{Baby-2022-Optimal}
D.~Baby and Y-X. Wang.
\newblock Optimal dynamic regret in proper online learning with strongly convex
  losses and beyond.
\newblock In \emph{AISTATS}, pages 1805--1845. PMLR, 2022.

\bibitem[Bach and Levy(2019)]{Bach-2019-Universal}
F.~Bach and K.~Y. Levy.
\newblock A universal algorithm for variational inequalities adaptive to
  smoothness and noise.
\newblock In \emph{COLT}, pages 164--194. PMLR, 2019.

\bibitem[Balamurugan and Bach(2016)]{Balamurugan-2016-Stochastic}
P.~Balamurugan and F.~Bach.
\newblock Stochastic variance reduction methods for saddle-point problems.
\newblock In \emph{NIPS}, pages 1416--1424, 2016.

\bibitem[Balandat et~al.(2016)Balandat, Krichene, Tomlin, and
  Bayen]{Balandat-2016-Minimizing}
M.~Balandat, W.~Krichene, C.~Tomlin, and A.~Bayen.
\newblock Minimizing regret on reflexive {B}anach spaces and {N}ash equilibria
  in continuous zero-sum games.
\newblock In \emph{NIPS}, pages 154--162, 2016.

\bibitem[Bauschke and Combettes(2011)]{Bauschke-2011-Convex}
H.~H. Bauschke and P.~L. Combettes.
\newblock \emph{Convex Analysis and Monotone Operator Theory in Hilbert
  Spaces}, volume 408.
\newblock Springer, 2011.

\bibitem[Besbes et~al.(2015)Besbes, Gur, and Zeevi]{Besbes-2015-Non}
O.~Besbes, Y.~Gur, and A.~Zeevi.
\newblock Non-stationary stochastic optimization.
\newblock \emph{Operations Research}, 63\penalty0 (5):\penalty0 1227--1244,
  2015.

\bibitem[Bloembergen et~al.(2015)Bloembergen, Tuyls, Hennes, and
  Kaisers]{Bloembergen-2015-Evolutionary}
D.~Bloembergen, K.~Tuyls, D.~Hennes, and M.~Kaisers.
\newblock Evolutionary dynamics of multi-agent learning: A survey.
\newblock \emph{Journal of Artificial Intelligence Research}, 53:\penalty0
  659--697, 2015.

\bibitem[Blum and Mansour(2007)]{Blum-2007-External}
A.~Blum and Y.~Mansour.
\newblock From external to internal regret.
\newblock \emph{Journal of Machine Learning Research}, 8\penalty0 (6):\penalty0
  1307--1324, 2007.

\bibitem[Blum(1998)]{Blum-1998-Online}
Avrim Blum.
\newblock Online algorithms in machine learning.
\newblock In \emph{Online Algorithms}, pages 306--325. Springer, 1998.

\bibitem[Bravo et~al.(2018)Bravo, Leslie, and Mertikopoulos]{Bravo-2018-Bandit}
M.~Bravo, D.~Leslie, and P.~Mertikopoulos.
\newblock Bandit learning in concave {N}-person games.
\newblock In \emph{NeurIPS}, pages 5661--5671, 2018.

\bibitem[Bubeck et~al.(2021)Bubeck, Eldan, and Lee]{Bubeck-2021-Kernel}
S.~Bubeck, R.~Eldan, and Y.~T. Lee.
\newblock Kernel-based methods for bandit convex optimization.
\newblock \emph{Journal of the ACM (JACM)}, 68\penalty0 (4):\penalty0 1--35,
  2021.

\bibitem[Cai et~al.(2022)Cai, Oikonomou, and Zheng]{Cai-2022-Finite}
Y.~Cai, A.~Oikonomou, and W.~Zheng.
\newblock Finite-time last-iterate convergence for learning in multi-player
  games.
\newblock In \emph{NeurIPS}, pages 33904--33919, 2022.

\bibitem[Cesa-Bianchi and Lugosi(2006)]{Cesa-2006-Prediction}
N.~Cesa-Bianchi and G.~Lugosi.
\newblock \emph{Prediction, Learning, and Games}.
\newblock Cambridge University Press, 2006.

\bibitem[Chen et~al.(2017)Chen, Lan, and Ouyang]{Chen-2017-Accelerated}
Y.~Chen, G.~Lan, and Y.~Ouyang.
\newblock Accelerated schemes for a class of variational inequalities.
\newblock \emph{Mathematical Programming}, 165:\penalty0 113--149, 2017.

\bibitem[Cover(1991)]{Cover-1991-Universal}
T.~M. Cover.
\newblock Universal portfolios.
\newblock \emph{Mathematical Finance}, 1\penalty0 (1):\penalty0 1--29, 1991.

\bibitem[Cover(1999)]{Cover-1999-Elements}
T.~M. Cover.
\newblock \emph{Elements of Information Theory}.
\newblock John Wiley \& Sons, 1999.

\bibitem[Cutkosky(2020{\natexlab{a}})]{Cutkosky-2020-Better}
A.~Cutkosky.
\newblock Better full-matrix regret via parameter-free online learning.
\newblock In \emph{NeurIPS}, pages 8836--8846, 2020{\natexlab{a}}.

\bibitem[Cutkosky(2020{\natexlab{b}})]{Cutkosky-2020-Parameter}
A.~Cutkosky.
\newblock Parameter-free, dynamic, and strongly-adaptive online learning.
\newblock In \emph{ICML}, pages 2250--2259. PMLR, 2020{\natexlab{b}}.

\bibitem[Cutkosky and Orabona(2018)]{Cutkosky-2018-Black}
A.~Cutkosky and F.~Orabona.
\newblock Black-box reductions for parameter-free online learning in banach
  spaces.
\newblock In \emph{COLT}, pages 1493--1529. PMLR, 2018.

\bibitem[Daniely et~al.(2015)Daniely, Gonen, and
  Shalev-Shwartz]{Daniely-2015-Strongly}
Amit Daniely, Alon Gonen, and Shai Shalev-Shwartz.
\newblock Strongly adaptive online learning.
\newblock In \emph{ICML}, pages 1405--1411, 2015.

\bibitem[Debreu(1952)]{Debreu-1952-Social}
G.~Debreu.
\newblock A social equilibrium existence theorem.
\newblock \emph{Proceedings of the National Academy of Sciences}, 38\penalty0
  (10):\penalty0 886--893, 1952.

\bibitem[Defazio et~al.(2014)Defazio, Bach, and
  Lacoste-Julien]{Defazio-2014-Saga}
A.~Defazio, F.~Bach, and S.~Lacoste-Julien.
\newblock {SAGA}: a fast incremental gradient method with support for
  non-strongly convex composite objectives.
\newblock In \emph{NIPS}, pages 1646--1654, 2014.

\bibitem[Duchi et~al.(2011)Duchi, Hazan, and Singer]{Duchi-2011-Adaptive}
J.~Duchi, E.~Hazan, and Y.~Singer.
\newblock Adaptive subgradient methods for online learning and stochastic
  optimization.
\newblock \emph{Journal of Machine Learning Research}, 12\penalty0 (7), 2011.

\bibitem[Facchinei and Pang(2007)]{Facchinei-2007-Finite}
F.~Facchinei and J-S. Pang.
\newblock \emph{Finite-dimensional Variational Inequalities and Complementarity
  Problems}.
\newblock Springer Science \& Business Media, 2007.

\bibitem[Fernholz(2002)]{Fernholz-2002-Stochastic}
E.~R. Fernholz.
\newblock \emph{Stochastic Portfolio Theory}, volume~48.
\newblock Springer Science \& Business Media, 2002.

\bibitem[Flaxman et~al.(2005)Flaxman, Kalai, and McMahan]{Flaxman-2005-Online}
A.~D. Flaxman, A.~T. Kalai, and H.~B. McMahan.
\newblock Online convex optimization in the bandit setting: gradient descent
  without a gradient.
\newblock In \emph{SODA}, pages 385--394, 2005.

\bibitem[Foster et~al.(2015)Foster, Rakhlin, and
  Sridharan]{Foster-2015-Adaptive}
D.~J. Foster, A.~Rakhlin, and K.~Sridharan.
\newblock Adaptive online learning.
\newblock In \emph{NeurIPS}, pages 3375--3383, 2015.

\bibitem[Foster et~al.(2017)Foster, Kale, Mohri, and
  Sridharan]{Foster-2017-Parameter}
D.~J. Foster, S.~Kale, M.~Mohri, and K.~Sridharan.
\newblock Parameter-free online learning via model selection.
\newblock In \emph{NeurIPS}, pages 6022--6032, 2017.

\bibitem[Garber and Kretzu(2022)]{Garber-2022-New}
D.~Garber and B.~Kretzu.
\newblock New projection-free algorithms for online convex optimization with
  adaptive regret guarantees.
\newblock In \emph{COLT}, pages 2326--2359. PMLR, 2022.

\bibitem[Goldsmith(2005)]{Goldsmith-2005-Wireless}
A.~Goldsmith.
\newblock \emph{Wireless Communications}.
\newblock Cambridge University Press, 2005.

\bibitem[Golowich et~al.(2020)Golowich, Pattathil, and
  Daskalakis]{Golowich-2020-Tight}
N.~Golowich, S.~Pattathil, and C.~Daskalakis.
\newblock Tight last-iterate convergence rates for no-regret learning in
  multi-player games.
\newblock In \emph{NeurIPS}, pages 20766--20778, 2020.

\bibitem[Hazan(2016)]{Hazan-2016-Introduction}
E.~Hazan.
\newblock Introduction to online convex optimization.
\newblock \emph{Foundations and Trends in Optimization}, 2\penalty0
  (3-4):\penalty0 157--325, 2016.

\bibitem[Hazan and Kakade(2019)]{Hazan-2019-Revisiting}
E.~Hazan and S.~Kakade.
\newblock Revisiting the {P}olyak step size.
\newblock \emph{ArXiv Preprint: 1905.00313}, 2019.

\bibitem[Hazan and Kale(2014)]{Hazan-2014-Beyond}
E.~Hazan and S.~Kale.
\newblock Beyond the regret minimization barrier: Optimal algorithms for
  stochastic strongly-convex optimization.
\newblock \emph{Journal of Machine Learning Research}, 15\penalty0
  (1):\penalty0 2489--2512, 2014.

\bibitem[Hazan and Seshadhri(2007)]{Hazan-2007-Adaptive}
E.~Hazan and C.~Seshadhri.
\newblock Adaptive algorithms for online decision problems.
\newblock \emph{ECCC}, 14\penalty0 (088), 2007.

\bibitem[Hazan et~al.(2007)Hazan, Agarwal, and Kale]{Hazan-2007-Logarithmic}
E.~Hazan, A.~Agarwal, and S.~Kale.
\newblock Logarithmic regret algorithms for online convex optimization.
\newblock \emph{Machine Learning}, 69\penalty0 (2-3):\penalty0 169--192, 2007.

\bibitem[Hazan et~al.(2014)Hazan, Koren, and Levy]{Hazan-2014-Logistic}
E.~Hazan, T.~Koren, and K.~Y. Levy.
\newblock Logistic regression: Tight bounds for stochastic and online
  optimization.
\newblock In \emph{COLT}, pages 197--209. PMLR, 2014.

\bibitem[Hsieh et~al.(2021)Hsieh, Antonakopoulos, and
  Mertikopoulos]{Hsieh-2021-Adaptive}
Y-G. Hsieh, K.~Antonakopoulos, and P.~Mertikopoulos.
\newblock Adaptive learning in continuous games: Optimal regret bounds and
  convergence to {N}ash equilibrium.
\newblock In \emph{COLT}, pages 2388--2422. PMLR, 2021.

\bibitem[Huang and Zhang(2022)]{Huang-2022-New}
K.~Huang and S.~Zhang.
\newblock New first-order algorithms for stochastic variational inequalities.
\newblock \emph{SIAM Journal on Optimization}, 32\penalty0 (4):\penalty0
  2745--2772, 2022.

\bibitem[Huang et~al.(2022)Huang, Wang, and Zhang]{Huang-2022-Accelerated}
K.~Huang, N.~Wang, and S.~Zhang.
\newblock An accelerated variance reduced extrapoint approach to finite-sum
  {VI} and optimization.
\newblock \emph{ArXiv Preprint: 2211.03269}, 2022.

\bibitem[Huh and Rusmevichientong(2009)]{Huh-2009-Nonparametric}
W.~T. Huh and P.~Rusmevichientong.
\newblock A nonparametric asymptotic analysis of inventory planning with
  censored demand.
\newblock \emph{Mathematics of Operations Research}, 34\penalty0 (1):\penalty0
  103--123, 2009.

\bibitem[Iusem et~al.(2017)Iusem, Jofr{\'e}, Oliveira, and
  Thompson]{Iusem-2017-Extragradient}
A.~N. Iusem, A.~Jofr{\'e}, R.~I. Oliveira, and P.~Thompson.
\newblock Extragradient method with variance reduction for stochastic
  variational inequalities.
\newblock \emph{SIAM Journal on Optimization}, 27\penalty0 (2):\penalty0
  686--724, 2017.

\bibitem[Iusem et~al.(2019)Iusem, Jofr{\'e}, Oliveira, and
  Thompson]{Iusem-2019-Variance}
A.~N. Iusem, A.~Jofr{\'e}, R.~I. Oliveira, and P.~Thompson.
\newblock Variance-based extragradient methods with line search for stochastic
  variational inequalities.
\newblock \emph{SIAM Journal on Optimization}, 29\penalty0 (1):\penalty0
  175--206, 2019.

\bibitem[Jadbabaie et~al.(2015)Jadbabaie, Rakhlin, Shahrampour, and
  Sridharan]{Jadbabaie-2015-Online}
A.~Jadbabaie, A.~Rakhlin, S.~Shahrampour, and K.~Sridharan.
\newblock Online optimization: Competing with dynamic comparators.
\newblock In \emph{AISTATS}, pages 398--406. PMLR, 2015.

\bibitem[Jalilzadeh and Shanbhag(2019)]{Jalilzadeh-2019-Proximal}
A.~Jalilzadeh and U.~V. Shanbhag.
\newblock A proximal-point algorithm with variable sample-sizes {(PPAWSS)} for
  monotone stochastic variational inequality problems.
\newblock In \emph{WSC}, pages 3551--3562. IEEE, 2019.

\bibitem[Jiang and Xu(2008)]{Jiang-2008-Stochastic}
H.~Jiang and H.~Xu.
\newblock Stochastic approximation approaches to the stochastic variational
  inequality problem.
\newblock \emph{IEEE Transactions on Automatic Control}, 53\penalty0
  (6):\penalty0 1462--1475, 2008.

\bibitem[Jin et~al.(2022)Jin, Sidford, and Tian]{Jin-2022-Sharper}
Y.~Jin, A.~Sidford, and K.~Tian.
\newblock Sharper rates for separable minimax and finite sum optimization via
  primal-dual extragradient methods.
\newblock In \emph{COLT}, pages 4362--4415. PMLR, 2022.

\bibitem[Johnson and Zhang(2013)]{Johnson-2013-Accelerating}
R.~Johnson and T.~Zhang.
\newblock Accelerating stochastic gradient descent using predictive variance
  reduction.
\newblock In \emph{NIPS}, pages 315--323, 2013.

\bibitem[Juditsky et~al.(2008)Juditsky, Rigollet, Tsybakov,
  et~al.]{Juditsky-2008-Learning}
A.~Juditsky, P.~Rigollet, A.~B. Tsybakov, et~al.
\newblock Learning by mirror averaging.
\newblock \emph{The Annals of Statistics}, 36\penalty0 (5):\penalty0
  2183--2206, 2008.

\bibitem[Juditsky et~al.(2011)Juditsky, Nemirovski, and
  Tauvel]{Juditsky-2011-Solving}
A.~Juditsky, A.~Nemirovski, and C.~Tauvel.
\newblock Solving variational inequalities with stochastic mirror-prox
  algorithm.
\newblock \emph{Stochastic Systems}, 1\penalty0 (1):\penalty0 17--58, 2011.

\bibitem[Jun and Orabona(2019)]{Jun-2019-Parameter}
K-S. Jun and F.~Orabona.
\newblock Parameter-free online convex optimization with sub-exponential noise.
\newblock In \emph{COLT}, pages 1802--1823. PMLR, 2019.

\bibitem[Jun et~al.(2017{\natexlab{a}})Jun, Orabona, Wright, and
  Willett]{Jun-2017-Improved}
K-S. Jun, F.~Orabona, S.~Wright, and R.~Willett.
\newblock Improved strongly adaptive online learning using coin betting.
\newblock In \emph{AISTATS}, pages 943--951. PMLR, 2017{\natexlab{a}}.

\bibitem[Jun et~al.(2017{\natexlab{b}})Jun, Orabona, Wright, and
  Willett]{Jun-2017-Online}
K-S. Jun, F.~Orabona, S.~Wright, and R.~Willett.
\newblock Online learning for changing environments using coin betting.
\newblock \emph{Electronic Journal of Statistics}, 11:\penalty0 5282--5310,
  2017{\natexlab{b}}.

\bibitem[Kalhan et~al.(2021)Kalhan, Bedi, Koppel, Rajawat, Hassani, Gupta, and
  Banerjee]{Kalhan-2021-Dynamic}
D.~S. Kalhan, A.~S. Bedi, A.~Koppel, K.~Rajawat, H.~Hassani, A.~K. Gupta, and
  A.~Banerjee.
\newblock Dynamic online learning via {F}rank-{W}olfe algorithm.
\newblock \emph{IEEE Transactions on Signal Processing}, 69:\penalty0 932--947,
  2021.

\bibitem[Kannan and Shanbhag(2019)]{Kannan-2019-Optimal}
A.~Kannan and U.~V. Shanbhag.
\newblock Optimal stochastic extragradient schemes for pseudomonotone
  stochastic variational inequality problems and their variants.
\newblock \emph{Computational Optimization and Applications}, 74\penalty0
  (3):\penalty0 779--820, 2019.

\bibitem[Kingma and Ba(2015)]{Kingma-2015-Adam}
D.~P. Kingma and J.~Ba.
\newblock Adam: A method for stochastic optimization.
\newblock In \emph{ICLR}, pages 1--15, 2015.
\newblock URL \url{https://openreview.net/forum?id=8gmWwjFyLj}.

\bibitem[Kivinen and Warmuth(1999)]{Kivinen-1999-Averaging}
J.~Kivinen and M.~K. Warmuth.
\newblock Averaging expert predictions.
\newblock In \emph{EuroCOLT}, pages 153--167. Springer, 1999.

\bibitem[Koren and Levy(2015)]{Koren-2015-Fast}
T.~Koren and K.~Levy.
\newblock Fast rates for exp-concave empirical risk minimization.
\newblock In \emph{NeurIPS}, pages 1477--1485, 2015.

\bibitem[Koshal et~al.(2012)Koshal, Nedi{\'c}, and
  Shanbhag]{Koshal-2012-Regularized}
J.~Koshal, A.~Nedi{\'c}, and U.~V. Shanbhag.
\newblock Regularized iterative stochastic approximation methods for stochastic
  variational inequality problems.
\newblock \emph{IEEE Transactions on Automatic Control}, 58\penalty0
  (3):\penalty0 594--609, 2012.

\bibitem[Kotsalis et~al.(2022)Kotsalis, Lan, and Li]{Kotsalis-2022-Simple}
G.~Kotsalis, G.~Lan, and T.~Li.
\newblock Simple and optimal methods for stochastic variational inequalities,
  {I}: Operator extrapolation.
\newblock \emph{SIAM Journal on Optimization}, 32\penalty0 (3):\penalty0
  2041--2073, 2022.

\bibitem[Krichene et~al.(2015)Krichene, Krichene, Dong, and
  Bayen]{Krichene-2015-Convergence}
S.~Krichene, W.~Krichene, R.~Dong, and A.~Bayen.
\newblock Convergence of heterogeneous distributed learning in stochastic
  routing games.
\newblock In \emph{Allerton}, pages 480--487. IEEE, 2015.

\bibitem[Lei and Jordan(2017)]{Lei-2017-Less}
L.~Lei and M.~I. Jordan.
\newblock Less than a single pass: Stochastically controlled stochastic
  gradient.
\newblock In \emph{AISTATS}, pages 148--156. PMLR, 2017.

\bibitem[Lei and Jordan(2020)]{Lei-2020-Adaptivity}
L.~Lei and M.~I. Jordan.
\newblock On the adaptivity of stochastic gradient-based optimization.
\newblock \emph{SIAM Journal on Optimization}, 30\penalty0 (2):\penalty0
  1473--1500, 2020.

\bibitem[Levy(2017)]{Levy-2017-Online}
K.~Y. Levy.
\newblock Online to offline conversions, universality and adaptive minibatch
  sizes.
\newblock In \emph{NeurIPS}, pages 1612--1621, 2017.

\bibitem[Levy et~al.(2018)Levy, Yurtsever, and Cevher]{Levy-2018-Online}
K.~Y. Levy, A.~Yurtsever, and V.~Cevher.
\newblock Online adaptive methods, universality and acceleration.
\newblock In \emph{NeurIPS}, pages 6501--6510, 2018.

\bibitem[Li and Orabona(2019)]{Li-2019-Convergence}
X.~Li and F.~Orabona.
\newblock On the convergence of stochastic gradient descent with adaptive
  stepsizes.
\newblock In \emph{AISTATS}, pages 983--992. PMLR, 2019.

\bibitem[Lin et~al.(2020)Lin, Zhou, Mertikopoulos, and Jordan]{Lin-2020-Finite}
T.~Lin, Z.~Zhou, P.~Mertikopoulos, and M.~I. Jordan.
\newblock Finite-time last-iterate convergence for multi-agent learning in
  games.
\newblock In \emph{ICML}, pages 6161--6171. PMLR, 2020.

\bibitem[Lin et~al.(2021)Lin, Zhou, Ba, and Zhang]{Lin-2021-Optimal}
T.~Lin, Z.~Zhou, W.~Ba, and J.~Zhang.
\newblock Optimal no-regret learning in strongly monotone games with bandit
  feedback.
\newblock \emph{ArXiv Preprint: 2112.02856}, 2021.

\bibitem[Loizou et~al.(2021)Loizou, Berard, Gidel, Mitliagkas, and
  Lacoste-Julien]{Loizou-2021-Stochastic}
N.~Loizou, H.~Berard, G.~Gidel, I.~Mitliagkas, and S.~Lacoste-Julien.
\newblock Stochastic gradient descent-ascent and consensus optimization for
  smooth games: Convergence analysis under expected co-coercivity.
\newblock In \emph{NeurIPS}, pages 19095--19108, 2021.

\bibitem[Lu et~al.(2023)Lu, Brukhim, Gradu, and Hazan]{Lu-2023-Projection}
Z.~Lu, N.~Brukhim, P.~Gradu, and E.~Hazan.
\newblock Projection-free adaptive regret with membership oracles.
\newblock In \emph{ALT}, pages 1055--1073. PMLR, 2023.

\bibitem[Luo et~al.(2016)Luo, Agarwal, Cesa-Bianchi, and
  Langford]{Luo-2016-Efficient}
H.~Luo, A.~Agarwal, N.~Cesa-Bianchi, and J.~Langford.
\newblock Efficient second order online learning by sketching.
\newblock In \emph{NeurIPS}, pages 902--910, 2016.

\bibitem[Mahdavi et~al.(2015)Mahdavi, Zhang, and Jin]{Mahdavi-2015-Lower}
M.~Mahdavi, L.~Zhang, and R.~Jin.
\newblock Lower and upper bounds on the generalization of stochastic
  exponentially concave optimization.
\newblock In \emph{COLT}, pages 1305--1320. PMLR, 2015.

\bibitem[Mertikopoulos and Zhou(2019)]{Mertikopoulos-2019-Learning}
P.~Mertikopoulos and Z.~Zhou.
\newblock Learning in games with continuous action sets and unknown payoff
  functions.
\newblock \emph{Mathematical Programming}, 173\penalty0 (1-2):\penalty0
  465--507, 2019.

\bibitem[Mertikopoulos et~al.(2018)Mertikopoulos, Papadimitriou, and
  Piliouras]{Mertikopoulos-2018-Cycles}
P.~Mertikopoulos, C.~Papadimitriou, and G.~Piliouras.
\newblock Cycles in adversarial regularized learning.
\newblock In \emph{SODA}, pages 2703--2717. SIAM, 2018.

\bibitem[Mertikopoulos et~al.(2019)Mertikopoulos, Lecouat, Zenati, Foo,
  Chandrasekhar, and Piliouras]{Mertikopoulos-2018-Optimistic}
P.~Mertikopoulos, B.~Lecouat, H.~Zenati, C-S. Foo, V.~Chandrasekhar, and
  G.~Piliouras.
\newblock Optimistic mirror descent in saddle-point problems: Going the
  extra(-gradient) mile.
\newblock In \emph{ICLR}, pages 1--23, 2019.
\newblock URL \url{https://openreview.net/forum?id=Bkg8jjC9KQ}.

\bibitem[Mokhtari et~al.(2016)Mokhtari, Shahrampour, Jadbabaie, and
  Ribeiro]{Mokhtari-2016-Online}
A.~Mokhtari, S.~Shahrampour, A.~Jadbabaie, and A.~Ribeiro.
\newblock Online optimization in dynamic environments: Improved regret rates
  for strongly convex problems.
\newblock In \emph{CDC}, pages 7195--7201. IEEE, 2016.

\bibitem[Monnot and Piliouras(2017)]{Monnot-2017-Limits}
B.~Monnot and G.~Piliouras.
\newblock Limits and limitations of no-regret learning in games.
\newblock \emph{The Knowledge Engineering Review}, 32, 2017.

\bibitem[Mukkamala and Hein(2017)]{Mukkamala-2017-Variants}
M.~C. Mukkamala and M.~Hein.
\newblock Variants of {RMSP}rop and {A}dagrad with logarithmic regret bounds.
\newblock In \emph{ICML}, pages 2545--2553. PMLR, 2017.

\bibitem[Nemirovski(2004)]{Nemirovski-2004-Prox}
A.~Nemirovski.
\newblock Prox-method with rate of convergence o(1/t) for variational
  inequalities with {L}ipschitz continuous monotone operators and smooth
  convex-concave saddle point problems.
\newblock \emph{SIAM Journal on Optimization}, 15\penalty0 (1):\penalty0
  229--251, 2004.

\bibitem[Nemirovski and Yudin(1983)]{Nemirovski-1983-Problem}
A.~Nemirovski and D.~Yudin.
\newblock \emph{Problem Complexity and Method Efficiency in Optimization}.
\newblock Wiley, 1983.

\bibitem[Nesterov(2018)]{Nesterov-2018-Lectures}
Y.~Nesterov.
\newblock \emph{Lectures on Convex Optimization}, volume 137.
\newblock Springer, 2018.

\bibitem[Netessine and Rudi(2003)]{Netessine-2003-Centralized}
S.~Netessine and N.~Rudi.
\newblock Centralized and competitive inventory models with demand
  substitution.
\newblock \emph{Operations Research}, 51\penalty0 (2):\penalty0 329--335, 2003.

\bibitem[Nguyen et~al.(2022)Nguyen, van Dijk, Phan, Nguyen, Weng, and
  Kalagnanam]{Nguyen-2022-Finite}
L.~M. Nguyen, M.~van Dijk, D.~T. Phan, P.~H. Nguyen, T-W. Weng, and J.~R.
  Kalagnanam.
\newblock Finite-sum smooth optimization with {SARAH}.
\newblock \emph{Computational Optimization and Applications}, 82\penalty0
  (3):\penalty0 561--593, 2022.

\bibitem[Orabona and P{\'a}l(2016)]{Orabona-2016-Coin}
F.~Orabona and D.~P{\'a}l.
\newblock Coin betting and parameter-free online learning.
\newblock In \emph{NeurIPS}, pages 577--585, 2016.

\bibitem[Orda et~al.(1993)Orda, Rom, and Shimkin]{Orda-1993-Competitive}
A.~Orda, R.~Rom, and N.~Shimkin.
\newblock Competitive routing in multiuser communication networks.
\newblock \emph{IEEE/ACM Transactions on Networking}, 1\penalty0 (5):\penalty0
  510--521, 1993.

\bibitem[Ordentlich and Cover(1998)]{Ordentlich-1998-Cost}
E.~Ordentlich and T.~M. Cover.
\newblock The cost of achieving the best portfolio in hindsight.
\newblock \emph{Mathematics of Operations Research}, 23\penalty0 (4):\penalty0
  960--982, 1998.

\bibitem[Osborne and Rubinstein(1994)]{Osborne-1994-Course}
M.~J. Osborne and A.~Rubinstein.
\newblock \emph{A Course in Game Theory}.
\newblock MIT Press, 1994.

\bibitem[Pal et~al.(2018)Pal, Wong, et~al.]{Pal-2018-Exponentially}
S.~Pal, T-K.~L. Wong, et~al.
\newblock Exponentially concave functions and a new information geometry.
\newblock \emph{The Annals of Probability}, 46\penalty0 (2):\penalty0
  1070--1113, 2018.

\bibitem[Polyak(1987)]{Polyak-1987-Introduction}
B.~Polyak.
\newblock \emph{Introduction to Optimization}.
\newblock Optimization Software, 1987.

\bibitem[Rappaport(2001)]{Rappaport-2001-Wireless}
T.~Rappaport.
\newblock \emph{Wireless Communications: Principles and Practice}.
\newblock Prentice Hall PTR, 2001.

\bibitem[Reddi et~al.(2018)Reddi, Kale, and Kumar]{Reddi-2018-Adam}
S.~J. Reddi, S.~Kale, and S.~Kumar.
\newblock On the convergence of {A}dam and beyond.
\newblock In \emph{ICLR}, pages 1--23, 2018.
\newblock URL \url{https://openreview.net/forum?id=ryQu7f-RZ}.

\bibitem[Rockafellar and Sun(2019)]{Rockafellar-2019-Solving}
R.~T. Rockafellar and J.~Sun.
\newblock Solving monotone stochastic variational inequalities and
  complementarity problems by progressive hedging.
\newblock \emph{Mathematical Programming}, 174\penalty0 (1-2):\penalty0
  453--471, 2019.

\bibitem[Rockafellar and Wets(2017)]{Rockafellar-2017-Stochastic}
R.~T. Rockafellar and R.~J.~B. Wets.
\newblock Stochastic variational inequalities: Single-stage to multistage.
\newblock \emph{Mathematical Programming}, 165\penalty0 (1):\penalty0 331--360,
  2017.

\bibitem[Rosen(1965)]{Rosen-1965-Existence}
J.~B. Rosen.
\newblock Existence and uniqueness of equilibrium points for concave {N}-person
  games.
\newblock \emph{Econometrica: Journal of the Econometric Society}, pages
  520--534, 1965.

\bibitem[Roux et~al.(2012)Roux, Schmidt, and Bach]{Roux-2012-Stochastic}
N.~L. Roux, M.~Schmidt, and F.~Bach.
\newblock A stochastic gradient method with an exponential convergence rate for
  finite training sets.
\newblock In \emph{NIPS}, pages 2663--2671, 2012.

\bibitem[Sandholm(2015)]{Sandholm-2015-Population}
W.~H. Sandholm.
\newblock Population games and deterministic evolutionary dynamics.
\newblock In \emph{Handbook of Game Theory with Economic Applications},
  volume~4, pages 703--778. Elsevier, 2015.

\bibitem[Shalev-Shwartz(2012)]{Shalev-2012-Online}
S.~Shalev-Shwartz.
\newblock Online learning and online convex optimization.
\newblock \emph{Foundations and Trends{\textregistered} in Machine Learning},
  4\penalty0 (2):\penalty0 107--194, 2012.

\bibitem[Shanbhag(2013)]{Shanbhag-2013-Stochastic}
U.~V. Shanbhag.
\newblock Stochastic variational inequality problems: Applications, analysis,
  and algorithms.
\newblock In \emph{Theory Driven by Influential Applications}, pages 71--107.
  INFORMS, Catonsville, MD, 2013.

\bibitem[Shoham and Leyton-Brown(2008)]{Shoham-2008-Multiagent}
Y.~Shoham and K.~Leyton-Brown.
\newblock \emph{Multiagent Systems: Algorithmic, Game-Theoretic, and Logical
  Foundations}.
\newblock Cambridge University Press, 2008.

\bibitem[Sorin and Wan(2016)]{Sorin-2016-Finite}
S.~Sorin and C.~Wan.
\newblock Finite composite games: Equilibria and dynamics.
\newblock \emph{Journal of Dynamics and Games}, 3\penalty0 (1):\penalty0
  101--120, 2016.

\bibitem[Tan(2014)]{Tan-2014-Wireless}
C.~W. Tan.
\newblock Wireless network optimization by {P}erron-{F}robenius theory.
\newblock In \emph{CISS}, pages 1--6. IEEE, 2014.

\bibitem[Vaswani et~al.(2019)Vaswani, Mishkin, Laradji, Schmidt, Gidel, and
  Lacoste-Julien]{Vaswani-2019-Painless}
S.~Vaswani, A.~Mishkin, I.~Laradji, M.~Schmidt, G.~Gidel, and
  S.~Lacoste-Julien.
\newblock Painless stochastic gradient: interpolation, line-search, and
  convergence rates.
\newblock In \emph{NIPS}, pages 3732--3745, 2019.

\bibitem[Villani(2009)]{Villani-2009-Optimal}
C.~Villani.
\newblock \emph{Optimal Transport: Old and New}, volume 338.
\newblock Springer Science \& Business Media, 2009.

\bibitem[Viossat and Zapechelnyuk(2013)]{Viossat-2013-Noregret}
Y.~Viossat and A.~Zapechelnyuk.
\newblock No-regret dynamics and fictitious play.
\newblock \emph{Journal of Economic Theory}, 148\penalty0 (2):\penalty0
  825--842, 2013.

\bibitem[Vovk(1997)]{Vovk-1997-Competitive}
V.~Vovk.
\newblock Competitive online linear regression.
\newblock In \emph{NIPS}, pages 364--370, 1997.

\bibitem[Wan et~al.(2021)Wan, Xue, and Zhang]{Wan-2021-Projection}
Y.~Wan, B.~Xue, and L.~Zhang.
\newblock Projection-free online learning in dynamic environments.
\newblock In \emph{AAAI}, pages 10067--10075, 2021.

\bibitem[Wan et~al.(2023)Wan, Zhang, and Song]{Wan-2023-Improved}
Y.~Wan, L.~Zhang, and M.~Song.
\newblock Improved dynamic regret for online {F}rank-{W}olfe.
\newblock \emph{ArXiv Preprint: 2302.05620}, 2023.

\bibitem[Wang et~al.(2020)Wang, Lu, Cheng, Tu, and Zhang]{Wang-2020-SAdam}
G.~Wang, S.~Lu, Q.~Cheng, W.~Tu, and L.~Zhang.
\newblock {SA}dam: A variant of {A}dam for strongly convex functions.
\newblock In \emph{ICLR}, pages 1--21, 2020.
\newblock URL \url{https://openreview.net/forum?id=rye5YaEtPr}.

\bibitem[Ward et~al.(2019)Ward, Wu, and Bottou]{Ward-2019-Adagrad}
R.~Ward, X.~Wu, and L.~Bottou.
\newblock Ada{G}rad stepsizes: Sharp convergence over nonconvex landscapes.
\newblock In \emph{ICML}, pages 6677--6686. PMLR, 2019.

\bibitem[Weeraddana et~al.(2012)Weeraddana, Codreanu, Latva-aho, Ephremides,
  and Fischione]{Weeraddana-2012-Weighted}
P.~C. Weeraddana, M.~Codreanu, M.~Latva-aho, A.~Ephremides, and C.~Fischione.
\newblock \emph{Weighted Sum-Rate Maximization in Wireless Networks: A Review}.
\newblock Now Foundations and Trends, 2012.

\bibitem[Xu et~al.(2017)Xu, Lin, and Yang]{Xu-2017-Adaptive}
Y.~Xu, Q.~Lin, and T.~Yang.
\newblock Adaptive {SVRG} methods under error bound conditions with unknown
  growth parameter.
\newblock In \emph{NIPS}, pages 3279--3289, 2017.

\bibitem[Yang et~al.(2016)Yang, Zhang, Jin, and Yi]{Yang-2016-Tracking}
T.~Yang, L.~Zhang, R.~Jin, and J.~Yi.
\newblock Tracking slowly moving clairvoyant: Optimal dynamic regret of online
  learning with true and noisy gradient.
\newblock In \emph{ICML}, pages 449--457. PMLR, 2016.

\bibitem[Yang et~al.(2018)Yang, Li, and Zhang]{Yang-2018-Simple}
T.~Yang, Z.~Li, and L.~Zhang.
\newblock A simple analysis for exp-concave empirical minimization with
  arbitrary convex regularizer.
\newblock In \emph{AISTATS}, pages 445--453, 2018.

\bibitem[Yousefian et~al.(2013)Yousefian, Nedi{\'c}, and
  Shanbhag]{Yousefian-2013-Regularized}
F.~Yousefian, A.~Nedi{\'c}, and U.~V. Shanbhag.
\newblock A regularized smoothing stochastic approximation {(RSSA)} algorithm
  for stochastic variational inequality problems.
\newblock In \emph{WSC}, pages 933--944. IEEE, 2013.

\bibitem[Yousefian et~al.(2014)Yousefian, Nedi{\'c}, and
  Shanbhag]{Yousefian-2014-Optimal}
F.~Yousefian, A.~Nedi{\'c}, and U.~V. Shanbhag.
\newblock Optimal robust smoothing extragradient algorithms for stochastic
  variational inequality problems.
\newblock In \emph{CDC}, pages 5831--5836. IEEE, 2014.

\bibitem[Yu et~al.(2022)Yu, Lin, Mazumdar, and Jordan]{Yu-2022-Fast}
Y.~Yu, T.~Lin, E.~Mazumdar, and M.~I. Jordan.
\newblock Fast distributionally robust learning with variance-reduced min-max
  optimization.
\newblock In \emph{AISTATS}, pages 1219--1250. PMLR, 2022.

\bibitem[Zhang et~al.(2018{\natexlab{a}})Zhang, Lu, and
  Zhou]{Zhang-2018-Adaptive}
L.~Zhang, S.~Lu, and Z-H. Zhou.
\newblock Adaptive online learning in dynamic environments.
\newblock In \emph{NIPS}, pages 1330--1340, 2018{\natexlab{a}}.

\bibitem[Zhang et~al.(2018{\natexlab{b}})Zhang, Yang, and
  Zhou]{Zhang-2018-Dynamic}
L.~Zhang, T.~Yang, and Z-H. Zhou.
\newblock Dynamic regret of strongly adaptive methods.
\newblock In \emph{ICML}, pages 5882--5891. PMLR, 2018{\natexlab{b}}.

\bibitem[Zhang et~al.(2019)Zhang, Liu, and Zhou]{Zhang-2019-Adaptive}
L.~Zhang, T-Y. Liu, and Z-H. Zhou.
\newblock Adaptive regret of convex and smooth functions.
\newblock In \emph{ICML}, pages 7414--7423. PMLR, 2019.

\bibitem[Zhou et~al.(2017)Zhou, Mertikopoulos, Moustakas, Bambos, and
  Glynn]{Zhou-2017-Mirror}
Z.~Zhou, P.~Mertikopoulos, A.~L. Moustakas, N.~Bambos, and P.~Glynn.
\newblock Mirror descent learning in continuous games.
\newblock In \emph{CDC}, pages 5776--5783. IEEE, 2017.

\bibitem[Zhou et~al.(2018)Zhou, Mertikopoulos, Athey, Bambos, Glynn, and
  Ye]{Zhou-2018-Learning}
Z.~Zhou, P.~Mertikopoulos, S.~Athey, N.~Bambos, P.~Glynn, and Y.~Ye.
\newblock Learning in games with lossy feedback.
\newblock In \emph{NIPS}, pages 1--11, 2018.

\bibitem[Zhou et~al.(2021)Zhou, Mertikopoulos, Moustakas, Bambos, and
  Glynn]{Zhou-2021-Robust}
Z.~Zhou, P.~Mertikopoulos, A.~L. Moustakas, N.~Bambos, and P.~Glynn.
\newblock Robust power management via learning and game design.
\newblock \emph{Operations Research}, 69\penalty0 (1):\penalty0 331--345, 2021.

\bibitem[Zinkevich(2003)]{Zinkevich-2003-Online}
M.~Zinkevich.
\newblock Online convex programming and generalized infinitesimal gradient
  ascent.
\newblock In \emph{ICML}, pages 928--936, 2003.

\bibitem[Zou et~al.(2019)Zou, Shen, Jie, Zhang, and Liu]{Zou-2019-Sufficient}
F.~Zou, L.~Shen, Z.~Jie, W.~Zhang, and W.~Liu.
\newblock A sufficient condition for convergences of {A}dam and {RMSP}rop.
\newblock In \emph{CVPR}, pages 11127--11135, 2019.

\end{thebibliography}

\clearpage
\appendix
%!TEX root = paper.tex
\section{Basic Proposition}\label{app:RV}
We provide a simple result on the maximum of independent identically distributed (i.i.d.) geometric random variables. 
\begin{proposition}\label{Prop:GRV}
Let $X_1, \ldots, X_n$ be $n$ i.i.d.\ geometric random variables: $X_i \sim \textsf{Geometric}(p_0)$ for all $i = 1, 2, \ldots, n$ and with $p_0 \in (0, 1)$. Defining $\bar{X}_n = \max_{1 \leq i \leq n} X_i$, we have
\begin{equation}\label{eq:GRV-first}
\sum_{n=1}^{+\infty} \PP(\bar{X}_n \leq x) \leq e^{xp_0}, 
\end{equation}
and
\begin{equation}\label{eq:GRV-second}
1 \leq \EE[\bar{X}_n] \leq \tfrac{1 + \log(n)}{p_0}. 
\end{equation}
\end{proposition}
\begin{proof}
Let $q_0 = 1-p_0$. Fixing a constant $x \geq 1$, we have
\begin{equation*}
\PP(\bar{X}_n \leq x) = \PP(\max\{X_1, X_2, \ldots, X_n\} \leq x) = \PP(X_1 \leq x, X_2 \leq x, \ldots, X_n \leq x). 
\end{equation*}
Since $X_1, X_2, \ldots, X_n$ are $n$ independently distributed random variables, we have
\begin{equation*}
\PP(\bar{X}_n \leq x) = \PP(X_1 \leq x)\PP(X_2 \leq x) \cdots \PP(X_n \leq x). 
\end{equation*}
Since $X_i \sim \textsf{Geometric}(p_0)$ for all $i = 1, 2, \ldots, n$,  we have
\begin{equation*}
\PP(X_i \leq x) = \sum_{k=1}^{\lfloor x\rfloor} \PP(X_i = k) = \sum_{k=1}^{\lfloor x\rfloor} (1-p_0)^{k-1}p_0 = p_0 \cdot \tfrac{1 - (1-p_0)^{\lfloor x\rfloor}}{p_0} = 1 - q_0^{\lfloor x\rfloor}, 
\end{equation*}
where $\lfloor x\rfloor$ is the largest integer that is smaller than $x$. As such, we have $\PP(\bar{X}_n \leq x) = (1 - q_0^{\lfloor x\rfloor})^n$. Given that $q_0 \in (0, 1)$, we have
\begin{equation*}
\sum_{n=1}^{+\infty} \PP(\bar{X}_n \leq x) = \sum_{n=1}^{+\infty} (1 - q_0^{\lfloor x\rfloor})^n = (1 - q_0^{\lfloor x\rfloor}) \cdot \frac{1}{q_0^{\lfloor x\rfloor}} = \frac{1}{q_0^{\lfloor x\rfloor}} - 1 \leq \frac{1}{q_0^x} = (1 - p_0)^{-x}. 
\end{equation*}
Since $1 + x \leq e^x$ for all $x \in \br$, we have $1 - p_0 \leq e^{-p_0}$. Putting these pieces together yields Eq.~\eqref{eq:GRV-first}. 

Moreover, it follows from the definition of $\bar{X}_n$ that $\bar{X}_n \geq 1$ and hence $\EE[\bar{X}_n] \geq 1$. We also have
\begin{equation*}
\EE[\bar{X}_n] = \sum_{k=0}^{+\infty} \PP(\bar{X}_n > k) = \sum_{k=0}^{+\infty} (1 - (1 - q_0^k)^n). 
\end{equation*}
By viewing the infinite sum as an Riemann sum approximation of an integral, we obtain 
\begin{equation*}
\sum_{k=0}^{+\infty} (1 - (1 - q_0^k)^n) \leq 1 + \int_0^{+\infty} (1 - (1 - q_0^x)^n) \; dx. 
\end{equation*}
We perform the change of variables $u = 1 - q_0^x$ and obtain 
\begin{eqnarray*}
\int_0^{+\infty} (1 - (1 - q_0^x)^n) dx & = & -\tfrac{1}{\log(q_0)} \int_0^1 \tfrac{1 - u^n}{1 - u} \; du \ = \ -\tfrac{1}{\log(q_0)} \int_0^1 (1 + u + \ldots + u^{n-1}) \; du \\
& = & -\tfrac{1}{\log(q_0)}(1 + \tfrac{1}{2} + \ldots + \tfrac{1}{n}) \ \leq \ -\tfrac{1}{\log(q_0)}\left(1 + \int_1^n \tfrac{1}{x}dx \right) \\ 
& = & -\tfrac{1 + \log(n)}{\log(1-p_0)}.  
\end{eqnarray*}
Recalling again that $1 + x \leq e^x$ for all $x \in \br$, we have $\log(1 - p_0) \leq -p_0$, which implies $-\tfrac{1}{\log(1-p_0)} \leq \frac{1}{p_0}$. Putting these pieces together yields Eq.~\eqref{eq:GRV-second}. 
\end{proof}

\section{Proof of Lemma~\ref{Lemma:AdaOGD}}\label{app:AdaOGD}
Recall that the update formula of $x^{t+1}$ in Algorithm~\ref{alg:AdaOGD-SA} is 
\begin{equation*}
x^{t+1} \leftarrow \argmin_{x \in \XCal}\{(x - x^t)^\top \xi^t + \tfrac{\eta^{t+1}}{2}\|x - x^t\|^2\}. 
\end{equation*}
The first-order optimality condition implies that 
\begin{equation*}
(x - x^{t+1})^\top \xi^t + \eta^{t+1} (x - x^{t+1})^\top(x^{t+1} - x^t) \geq 0, \quad \textnormal{for all } x \in \XCal. 
\end{equation*}
Equivalently, we have
\begin{eqnarray*}
\lefteqn{\tfrac{\eta^{t+1}}{2}(\|x^{t+1} - x\|^2 - \|x^t - x\|^2) \leq (x - x^{t+1})^\top \xi^t - \tfrac{\eta^{t+1}}{2}\|x^{t+1} - x^t\|^2} \\
& = & (x - x^t)^\top \xi^t + (x^t - x^{t+1})^\top \xi^t - \tfrac{\eta^{t+1}}{2}\|x^{t+1} - x^t\|^2 \ \leq \ (x - x^t)^\top \xi^t + \tfrac{1}{2\eta^{t+1}}\|\xi^t\|^2. 
\end{eqnarray*}
Since $\EE[\xi^t \mid x^t] = \nabla f_t(x^t)$ and $\EE[\|\xi^t\|^2 \mid x^t] \leq G^2$ for all $t \geq 1$, we have 
\begin{equation*}
\EE\left[\tfrac{\eta^{t+1}}{2}\left(\|x^{t+1} - x\|^2 - \|x^t - x\|^2\right) \mid x^t\right] \leq (x - x^t)^\top \nabla f_t(x^t) + \EE\left[\tfrac{G^2}{2\eta^{t+1}}\right]. 
\end{equation*}
Since $f_t$ is $\beta$-strongly convex, we have
\begin{equation*}
\EE\left[\tfrac{\eta^{t+1}}{2}\left(\|x^{t+1} - x\|^2 - \|x^t - x\|^2\right)\mid x^t\right] \leq f_t(x) - f_t(x^t) - \tfrac{\beta}{2}\|x^t - x\|^2 + \EE\left[\tfrac{G^2}{2\eta^{t+1}}\right], 
\end{equation*}
Taking the expectation of both sides and rearranging the resulting inequality yields  
\begin{equation*}
\EE\left[f_t(x^t) + \tfrac{\eta^{t+1}}{2}\|x^{t+1} - x\|^2 - \tfrac{\eta^t}{2}\|x^t - x\|^2\right] - f_t(x) \leq \EE\left[\left(\tfrac{\eta^{t+1} - \eta^t}{2} - \tfrac{\beta}{2}\right)\|x^t - x\|^2 + \tfrac{G^2}{2\eta^{t+1}}\right]. 
\end{equation*}
Summing up the above inequality over $t= 1, 2, \ldots, T$ yields the desired inequality. 

\section{Proofs for Example~\ref{Example:PM} and~\ref{Example:RIG}}\label{app:example}
We show that \textsc{Power Management in Wireless Networks} and \textsc{Newsvendor with Lost Sales} in Example~\ref{Example:PM} and~\ref{Example:RIG} are $\beta$-strongly monotone games (cf. Definition~\ref{def:SM-game}) for some $\beta > 0$.  

\paragraph{Power Management in Wireless Networks.} Example~\ref{Example:PM} satisfies Definition~\ref{def:SM-game} with $\beta = \lambda_{\min}(I - \frac{1}{2}(W + W^\top))> 0$ where $W_{ii} = 0$ for all $1 \leq i \leq N$ and $W_{ij} = \frac{r_i^\star G_{ij}}{G_{ii}}$ for $i \neq j$. An analysis of this setting has been given in~\citet{Zhou-2021-Robust}; we provide the details for  completeness. The cost function is given by 
\begin{equation*}
u_i(a) = \tfrac{1}{2}\left(a_i - \tfrac{r_i^\star(\sum_{j \neq i}G_{ij} a_j + \eta_i)}{G_{ii}}\right)^2. 
\end{equation*}
Taking the derivative of $u_i(a)$ with respect to $a_i$ yields that $v_i(a) = a_i - \frac{r_i^\star(\sum_{j \neq i}G_{ij} a_j + \eta_i)}{G_{ii}}$. Consequently, by the definition of $W_{ij}$, we have
\begin{eqnarray*}
\lefteqn{\langle a' - a, v(a') - v(a)\rangle = \|a' - a\|^2 - \sum_{i=1}^N \tfrac{r_i^\star}{G_{ii}} \sum_{j \neq i} G_{ij} \langle a'_i - a_i, a'_j - a_j\rangle} \\ 
& = & \|a' - a\|^2 - \sum_{i=1}^N \sum_{j=1}^N W_{ij} \langle a'_i - a_i, a'_j - a_j\rangle = \|a' - a\|^2 - \langle a' - a,  W(a' - a)\rangle = \langle a' - a,  (I - \tfrac{1}{2}(W + W^\top))(a' - a)\rangle \\ 
& \geq & \lambda_{\min}(I - \tfrac{1}{2}(W + W^\top))\|a' - a\|^2 = \beta\|a' - a\|^2. 
\end{eqnarray*}
This yields the desired result. 

\paragraph{Newsvendor with Lost Sales.} Example~\ref{Example:RIG} satisfies Definition~\ref{def:SM-game} with $\beta = \alpha p$. Indeed, each players' reward function is given by 
\begin{equation*}
u_i(x) = (p - c_i) \cdot \EE[\max\{0, D_i(x_{-i}) - x_i\}] + c_i \cdot \EE[\max\{0, x_i - D_i(x_{-i})\}], 
\end{equation*}
Equivalently, we have
\begin{equation*}
u_i(x) = (p - c_i) \cdot \int_{x_i}^{+\infty} (k - x_i) \cdot d\PP(D_i(x_{-i}) \leq k) + c_i \cdot \int_{-\infty}^{x_i} (x_i - k) \cdot d\PP(D_i(x_{-i}) \leq k). 
\end{equation*}
Using the Leibniz integral rule, we have
\begin{equation*}
v_i(x) = \nabla_{x_i} u_i(x) = p \cdot \PP(D_i(x_{-i}) \leq x_i) - p + c_i = p \cdot F_i(x) - p + c_i. 
\end{equation*}
Let $F = (F_1, F_2, \ldots, F_N)$ be an operator from $\prod_{i=1}^N [0, \bar{x}_i]$ to $[0, 1]^N$. Then, if $F$ is $\alpha$-strongly monotone, we have 
\begin{equation*}
\langle x' - x, v(x') - v(x)\rangle = p \cdot \langle x' - x, F(x') - F(x)\rangle \geq \alpha p\|x' - x\|^2, \quad \textnormal{for all } x, x' \in \prod_{i=1}^N [0, \bar{x}_i]. 
\end{equation*}
Considering the example where the distribution of a demand $D_i(x_{-i})$ is given by 
\begin{equation*}
\PP(D_i(x_{-i}) \leq z) = 1 - \tfrac{1 + \sum_{j \neq i} x_j}{(1 + z + \sum_{j \neq i} x_j)^2}. 
\end{equation*}
Then, we have 
\begin{equation*}
v_i(x) = p \cdot \left(1 - \tfrac{1 + \sum_{j \neq i} x_j}{(1 + \sum_{i = 1}^N x_i)^2}\right) - p + c_i. 
\end{equation*}
The following proposition, a modification of~\citet[Theorem~6]{Rosen-1965-Existence}, plays an important role in the subsequent analysis. 
\begin{proposition}\label{prop:DSC}
Given a continuous game $\GCal = (\NCal, \XCal = \prod_{i=1}^N \XCal_i, \{u_i\}_{i=1}^N)$, where each $u_i$ is twice continuously differentiable. For each $x \in \XCal$, we define the Hessian matrix $H(x)$ as follows:
\begin{equation*}
H_{ij}(x) = \tfrac{1}{2}\nabla_j v_i(x) + \tfrac{1}{2}(\nabla_i v_j(x))^\top. 
\end{equation*}
If $H(x)$ is positive definite for every $x \in \XCal$, we have $\langle x' - x, v(x') - v(x)\rangle \geq 0$ for all $x, x' \in \XCal$ where the equality holds true if and only if $x = x'$. 
\end{proposition}
As a consequence of Proposition~\ref{prop:DSC}, we have $\langle x' - x, v(x) - v(x)\rangle \geq \beta\|x' - x\|^2$ for all $x, x' \in \XCal$ if $H(x) \succeq \beta I_N$ for all $x \in \XCal$. For our example, it suffices to show that 
\begin{equation}\label{Eq:RIG-main}
H(x) \succeq \left(\tfrac{p}{(1 + \sum_{i = 1}^N \bar{x}_i)^3}\right)I_N, \quad \textnormal{for all } x \in \prod_{i=1}^N [0, \bar{x}_i].  
\end{equation}
Indeed, after a straightforward calculation, we have
\begin{equation*}
\nabla_i v_i(x) = \tfrac{2p(1 + \sum_{k \neq i} x_k)}{(1 + \sum_{k = 1}^N x_k)^3}, \qquad \nabla_j v_i(x) = \tfrac{p(1 + \sum_{k \neq i} x_k - x_i)}{(1 + \sum_{k = 1}^N x_k)^3}, \qquad \nabla_i v_j(x) = \tfrac{p(1 + \sum_{k \neq j} x_k - x_j)}{(1 + \sum_{k = 1}^N x_k)^3}. 
\end{equation*}
This implies that 
\begin{equation*}
H(x) = \tfrac{p}{(1 + \sum_{i = 1}^N x_i)^3} \cdot
\begin{pmatrix}
2 + 2\sum_{i \neq 1} x_i & 1 + \sum_{i \neq 1, 2} x_i & 1 + \sum_{i \neq 1, 3} x_i & \cdots & 1 + \sum_{i \neq 1, N} x_i \\
1 + \sum_{i \neq 1, 2} x_i & 2 + 2\sum_{i \neq 2} x_i & 1 + \sum_{i \neq 2, 3} x_i & \cdots & 1 + \sum_{i \neq 2, N} x_i \\
1 + \sum_{i \neq 1, 3} x_i & 1 + \sum_{i \neq 2, 3} x_i & 2 + 2\sum_{i \neq 3} x_i & \cdots & 1 + \sum_{i \neq 3, N} x_i \\
\vdots & \vdots & \vdots & \ddots & \vdots \\
1 + \sum_{i \neq 1, N} x_i & 1 + \sum_{i \neq 2, N} x_i & 1 + \sum_{i \neq 3, N} x_i & \cdots & 2 + 2\sum_{i \neq N} x_i
\end{pmatrix}.
\end{equation*}
Equivalently, we have 
\begin{eqnarray*}
H(x) & = & \tfrac{p}{(1 + \sum_{i = 1}^N x_i)^3} \cdot \left\{\left(1 + \sum_{i = 1}^N x_i\right) \cdot \begin{pmatrix}
2 & 1 & 1 & \cdots & 1 \\
1 & 2 & 1 & \cdots & 1 \\
1 & 1 & 2 & \cdots & 1 \\
\vdots & \vdots & \vdots & \ddots & \vdots \\
1 & 1 & 1 & \cdots & 2
\end{pmatrix} - \begin{pmatrix}
2x_1 & x_1 + x_2 & x_1 + x_3 & \cdots & x_1 + x_N \\
x_1 + x_2 & 2x_2 & x_2 + x_3 & \cdots & x_2 + x_N \\
x_1 + x_3 & x_2 + x_3 & 2x_3 & \cdots & x_3 + x_N \\
\vdots & \vdots & \vdots & \ddots & \vdots \\
x_1 + x_N & x_2 + x_N & x_3 + x_N & \cdots & 2x_N
\end{pmatrix}\right\} \\
& = & \tfrac{p}{(1 + \sum_{i = 1}^N x_i)^3} \cdot \left\{\left(1 + \sum_{i = 1}^N x_i\right) \cdot (I_N + \one_N \one_N^\top) -x\one_N^\top - \one_N x^\top\right\}. 
\end{eqnarray*}
Note that $x \in \prod_{i=1}^N [0, \bar{x}_i]$. If $x = \zero_N$, we have $H(x) = p \cdot (I_N + \one_N \one_N^\top)$ and thus satisfy Eq.~\eqref{Eq:RIG-main}. Otherwise, we let $y = \frac{x}{\sum_{i=1}^N x_k}$ and obtain that $\|y\| \leq 1$. It is clear that $(y - \one_N)(y - \one_N)^\top$ is positive semidefinite. Thus, we have 
\begin{equation*}
yy^\top + \one_N \one_N^\top \succeq y\one_N^\top + \one_N y^\top. 
\end{equation*}
Using the definition of $y$, we have $yy^\top \preceq I_N$.  This together with the above inequality implies that 
\begin{equation*}
\left(\sum_{i = 1}^N x_i\right) \cdot (I_N+ \one_N \one_N^\top) \succeq x\one_N^\top + \one_N x^\top. 
\end{equation*}
Putting these pieces together yields that 
\begin{equation*}
H(x) \succeq \tfrac{p}{(1 + \sum_{i = 1}^N x_i)^3} \cdot (I_N+ \one_N \one_N^\top) \succeq \left(\tfrac{p}{(1 + \sum_{i = 1}^N \bar{x}_i)^3}\right)I_N, \quad \textnormal{for all } x \in \prod_{i=1}^N [0, \bar{x}_i].  
\end{equation*}
This completes the proof. 

\section{Proof of Lemma~\ref{Lemma:AdaOGD-MA}}\label{app:AdaOGD-MA}
Recall that the update formula of $x_i^{t+1}$ in Algorithm~\ref{alg:AdaOGD-MA} is 
\begin{equation*}
x_i^{t+1} \leftarrow \argmin_{x_i \in \XCal_i}\{(x_i - x_i^t)^\top \xi_i^t + \tfrac{\eta_i^{t+1}}{2}\|x_i - x_i^t\|^2\}. 
\end{equation*}
The first-order optimality condition implies that 
\begin{equation*}
(x_i - x_i^{t+1})^\top \xi_i^t + \eta_i^{t+1} (x_i - x_i^{t+1})^\top(x_i^{t+1} - x_i^t) \geq 0, \quad \textnormal{for all } x_i \in \XCal_i. 
\end{equation*}
Equivalently, we have
\begin{eqnarray*}
\lefteqn{\tfrac{\eta_i^{t+1}}{2}(\|x_i^{t+1} - x_i\|^2 - \|x_i^t - x_i\|^2) \leq (x_i - x_i^{t+1})^\top \xi_i^t - \tfrac{\eta_i^{t+1}}{2}\|x_i^{t+1} - x_i^t\|^2} \\
& = & (x_i - x_i^t)^\top \xi_i^t + (x_i^t - x_i^{t+1})^\top \xi_i^t - \tfrac{\eta_i^{t+1}}{2}\|x_i^{t+1} - x_i^t\|^2 \ \leq \ (x_i - x_i^t)^\top \xi_i^t + \tfrac{1}{2\eta_i^{t+1}}\|\xi_i^t\|^2. 
\end{eqnarray*}
Letting $x = x^\star$ be a unique Nash equilibrium and rearranging the above inequality, we have
\begin{equation*}
\tfrac{\eta_i^{t+1}}{2}\|x_i^{t+1} - x_i^\star\|^2 - \tfrac{\eta_i^t}{2}\|x_i^t - x_i^\star\|^2 \leq (x_i^\star - x_i^t)^\top \xi_i^t + \left(\tfrac{\eta_i^{t+1}}{2} - \tfrac{\eta_i^t}{2}\right)\|x_i^t - x_i^\star\|^2 + \tfrac{1}{2\eta_i^{t+1}}\|\xi_i^t\|^2. 
\end{equation*}
Summing up the above inequality over $i = 1, 2, \ldots, N$ and rearranging, we have
\begin{equation*}
\sum_{i=1}^N \left(\eta_i^{t+1}\|x_i^{t+1} - x_i^\star\|^2 - \eta_i^t\|x_i^t - x_i^\star\|^2\right) \leq 2(x^\star - x^t)^\top \xi^t + \left(\sum_{i = 1}^N (\eta_i^{t+1} - \eta_i^t)\|x_i^t - x_i^\star\|^2\right) + \left(\max_{1 \leq i \leq N} \left\{\tfrac{1}{\eta_i^{t+1}}\right\}\right)\|\xi^t\|^2. 
\end{equation*}
Note that $\{\eta_i^t\}_{1 \leq i \leq N, 1 \leq t \leq T}$ are generated independently of any noisy gradient feedback, we have
\begin{eqnarray*}
\EE[(x^\star - x^t)^\top\xi^t \mid x^t, \{\eta_i^t\}_{1 \leq i \leq N, 1 \leq t \leq T}] & = & \EE[(x^\star - x^t)^\top v(x^t) \mid \{\eta_i^t\}_{1 \leq i \leq N, 1 \leq t \leq T}], \\ 
\EE[\|\xi^t\|^2 \mid x^t, \{\eta_i^t\}_{1 \leq i \leq N, 1 \leq t \leq T}] & \leq & G^2. 
\end{eqnarray*}
Thus, we have
\begin{eqnarray*}
\lefteqn{\sum_{i = 1}^N \left(\eta_i^{t+1}\EE\left[\|x_i^{t+1} - x_i^\star\|^2 \mid \{\eta_i^t\}_{1 \leq i \leq N, 1 \leq t \leq T}\right] - \eta_i^t\EE\left[\|x_i^t - x_i^\star\|^2 \mid \{\eta_i^t\}_{1 \leq i \leq N, 1 \leq t \leq T}\right]\right)} \\ 
& \leq & 2\EE\left[(x^\star - x^t)^\top v(x^t) \mid \{\eta_i^t\}_{1 \leq i \leq N, 1 \leq t \leq T}\right] + \sum_{i = 1}^N \left(\eta_i^{t+1} - \eta_i^t\right)\EE\left[\|x_i^t - x_i^\star\|^2 \mid \{\eta_i^t\}_{1 \leq i \leq N, 1 \leq t \leq T}\right] + G^2\left(\max_{1 \leq i \leq N} \left\{\tfrac{1}{\eta_i^{t+1}}\right\}\right). 
\end{eqnarray*}
Since a continuous game $\GCal$ is $\beta$-strongly monotone and $x^\star \in \XCal$ is a unique Nash equilibrium, we have
\begin{equation*}
(x^\star - x^t)^\top v(x^t) \leq (x^\star - x^t)^\top v(x^\star) - \beta\|x^\star - x^t\|^2 \leq - \beta\|x^\star - x^t\|^2. 
\end{equation*}
Putting these pieces together yields that 
\begin{eqnarray*}
\lefteqn{\sum_{i = 1}^N \left(\eta_i^{t+1}\EE\left[\|x_i^{t+1} - x_i^\star\|^2 \mid \{\eta_i^t\}_{1 \leq i \leq N, 1 \leq t \leq T}\right] - \eta_i^t\EE\left[\|x_i^t - x_i^\star\|^2 \mid \{\eta_i^t\}_{1 \leq i \leq N, 1 \leq t \leq T}\right]\right)} \\ 
& \leq & \sum_{i = 1}^N \left(\eta_i^{t+1} - \eta_i^t - 2\beta\right)\EE\left[\|x_i^t - x_i^\star\|^2 \mid \{\eta_i^t\}_{1 \leq i \leq N, 1 \leq t \leq T}\right] + G^2\left(\max_{1 \leq i \leq N} \left\{\tfrac{1}{\eta_i^{t+1}}\right\}\right). 
\end{eqnarray*}
Since $D > 0$ satisfies that $\|x - x'\| \leq D$ for all $x, x' \in \XCal$, we have
\begin{eqnarray*}
\lefteqn{\sum_{i = 1}^N \left(\eta_i^{t+1}\EE\left[\|x_i^{t+1} - x_i^\star\|^2 \mid \{\eta_i^t\}_{1 \leq i \leq N, 1 \leq t \leq T}\right] - \eta_i^t\EE\left[\|x_i^t - x_i^\star\|^2 \mid \{\eta_i^t\}_{1 \leq i \leq N, 1 \leq t \leq T}\right]\right)} \\
& \leq & \left(\max_{1 \leq i \leq N} \left\{\eta_i^{t+1} - \eta_i^t\right\} - 2\beta\right)\EE\left[\|x^t - x^\star\|^2 \mid \{\eta_i^t\}_{1 \leq i \leq N, 1 \leq t \leq T}\right] + G^2\left(\max_{1 \leq i \leq N} \left\{\tfrac{1}{\eta_i^{t+1}}\right\}\right) \\
& \leq & D^2\left(\max\left\{0, \max_{1 \leq i \leq N} \left\{\eta_i^{t+1} - \eta_i^t\right\} - 2\beta\right\}\right) + G^2\left(\max_{1 \leq i \leq N} \left\{\tfrac{1}{\eta_i^{t+1}}\right\}\right). 
\end{eqnarray*}
Summing over $t= 1, 2, \ldots, T-1$ yields the desired inequality. 

\section{Proof of Lemma~\ref{Lemma:AdaONS}}\label{app:AdaONS}
Recall that the update formula of $x^{t+1}$ in Algorithm~\ref{alg:AdaONS-SA} is 
\begin{equation*}
x^{t+1} \leftarrow \argmin_{x \in \XCal}\{(x - x^t)^\top \nabla f_t(x^t) + \tfrac{\eta^{t+1}}{2}(x - x^t)^\top A^{t+1} (x - x^t)\}. 
\end{equation*}
The first-order optimality condition implies that 
\begin{equation*}
(x - x^{t+1})^\top \nabla f_t(x^t) + \eta^{t+1} (x - x^{t+1})^\top A^{t+1}(x^{t+1} - x^t) \geq 0, \quad \textnormal{for all } x \in \XCal. 
\end{equation*}
Equivalently, we have
\begin{eqnarray*}
\lefteqn{\tfrac{\eta^{t+1}}{2}((x^{t+1} - x)^\top A^{t+1}(x^{t+1} - x) - (x^t - x)^\top A^{t+1}(x^t - x))} \\
& \leq & (x - x^{t+1})^\top \nabla f_t(x^t) - \tfrac{\eta^{t+1}}{2}(x^{t+1} - x^t)^\top A^{t+1}(x^{t+1} - x^t)   \\ 
& = & (x - x^t)^\top \nabla f_t(x^t) + (x^t - x^{t+1})^\top \nabla f_t(x^t) - \tfrac{\eta^{t+1}}{2}(x^{t+1} - x^t)^\top A^{t+1}(x^{t+1} - x^t) \\
& \leq & (x - x^t)^\top \nabla f_t(x^t) + \tfrac{1}{2\eta^{t+1}}\nabla f_t(x^t)^\top (A^{t+1})^{-1}\nabla f_t(x^t). 
\end{eqnarray*}
Since $f_t$ is $\alpha$-exp-concave and satisfies that $\|\nabla f_t(x)\| \leq G$ and $\|x - x'\| \leq D$ for all $x, x' \in \XCal$, we derive from~\citet[Lemma~3]{Hazan-2007-Logarithmic} that 
\begin{equation*}
f_t(x) \geq f_t(x^t) + (x - x^t)^\top\nabla f_t(x^t) + \tfrac{1}{4}\min\{\tfrac{1}{4GD}, \alpha\}(x - x^t)^\top(\nabla f_t(x^t)\nabla f_t(x^t)^\top)(x - x^t). 
\end{equation*}
For simplicity, we let $\gamma = \frac{1}{4}\min\{\tfrac{1}{4GD}, \alpha\}$. Putting these pieces together yields that 
\begin{eqnarray*}
\lefteqn{\tfrac{\eta^{t+1}}{2}((x^{t+1} - x)^\top A^{t+1}(x^{t+1} - x) - (x^t - x)^\top A^{t+1}(x^t - x)) \leq f_t(x) - f_t(x^t)} \\ 
& & - \gamma(x - x^t)^\top(\nabla f_t(x^t)\nabla f_t(x^t)^\top)(x - x^t) + \tfrac{1}{2\eta^{t+1}}\nabla f_t(x^t)^\top (A^{t+1})^{-1}\nabla f_t(x^t). 
\end{eqnarray*}
Rearranging the above inequality yields that 
\begin{eqnarray*}
\lefteqn{f_t(x^t) - f_t(x) + \tfrac{\eta^{t+1}}{2}(x^{t+1} - x)^\top A^{t+1}(x^{t+1} - x) - \tfrac{\eta^t}{2} (x^t - x)^\top A^t(x^t - x)} \\ 
& \leq & (x - x^t)^\top\left(\tfrac{\eta^{t+1}}{2} A^{t+1} - \tfrac{\eta^t}{2}A^t - \gamma\nabla f_t(x^t)\nabla f_t(x^t)^\top\right)(x - x^t) + \tfrac{1}{2\eta^{t+1}}\nabla f_t(x^t)^\top (A^{t+1})^{-1}\nabla f_t(x^t). 
\end{eqnarray*}
Summing over $t= 1, 2, \ldots, T$ yields the desired inequality. 

\section{Proof of Theorem~\ref{Thm:AdaONS-regret}} \label{app:AdaONS-regret}
By the update formula of $\eta^{t+1}$ in Algorithm~\ref{alg:AdaONS-SA}, we have $\eta^{t+1} = \tfrac{1}{\sqrt{1 + \max\{M^1, \ldots, M^t\}}}$. By the update formula of $A^{t+1}$ in Algorithm~\ref{alg:AdaONS-SA}, we have $A^1 = I_d$ where $I_d \in \br^{d \times d}$ is an identity matrix and $A^{t+1} = A^t + \nabla f_t(x^t)\nabla f_t(x^t)^\top$. Since $\XCal$ is convex and bounded with a diameter $D > 0$, we have
\begin{equation*}
\tfrac{\eta^1}{2}(x^1 - x)^\top A^1(x^1 - x) \leq \tfrac{D^2}{2}, \qquad \eta^{t+1} A^{t+1} - \eta^t A^t \preceq \tfrac{1}{\sqrt{1 + \max\{M^1, \ldots, M^t\}}}\nabla f_t(x^t)\nabla f_t(x^t)^\top. 
\end{equation*}
By Lemma~\ref{Lemma:AdaONS}, we have
\begin{eqnarray}\label{inequality:AdaONS-regret-first}
\lefteqn{\sum_{t=1}^T f_t(x^t) - \sum_{t=1}^T f_t(x) \leq \tfrac{D^2}{2} + \tfrac{1}{2}\left(\sum_{t=1}^T \tfrac{1}{\eta^{t+1}}\nabla f_t(x^t)^\top (A^{t+1})^{-1}\nabla f_t(x^t)\right)} \\
& & + \sum_{t=1}^T \left(\tfrac{1}{2\sqrt{1 + \max\{M^1, \ldots, M^t\}}} - \tfrac{1}{4}\min\{\tfrac{1}{4GD}, \alpha\}\right) (x^t - x)^\top \nabla f_t(x^t)\nabla f_t(x^t)^\top (x^t - x). \nonumber
\end{eqnarray}
Since $\|\nabla f_t(x)\| \leq G$, $A^1 = I_d$ where $I_d \in \br^{d \times d}$ is an identity matrix and $A^{t+1} = A^t + \nabla f_t(x^t)\nabla f_t(x^t)^\top$, we derive from~\citet[Lemma~11]{Hazan-2007-Logarithmic} that 
\begin{eqnarray}\label{inequality:AdaONS-regret-second}
\sum_{t=1}^T \tfrac{1}{\eta^{t+1}}\nabla f_t(x^t)^\top (A^{t+1})^{-1}\nabla f_t(x^t) & \leq & \sqrt{1 + \max\{M^1, \ldots, M^T\}}\left(\sum_{t=1}^T \nabla f_t(x^t)^\top (A^{t+1})^{-1}\nabla f_t(x^t)\right) \nonumber \\ 
& \leq & \sqrt{1 + \max\{M^1, \ldots, M^T\}}(d\log(TG^2+1)). 
\end{eqnarray}
Plugging Eq.~\eqref{inequality:AdaONS-regret-second} into Eq.~\eqref{inequality:AdaONS-regret-first} yields that 
\begin{eqnarray*}
\lefteqn{\sum_{t=1}^T f_t(x^t) - \sum_{t=1}^T f_t(x) \leq \tfrac{D^2}{2} + \tfrac{d\sqrt{1 + \max\{M^1, \ldots, M^T\}}}{2}\log(TG^2+1)} \\ 
& & + \sum_{t=1}^T \left(\tfrac{1}{2\sqrt{1 + \max\{M^1, \ldots, M^t\}}} - \tfrac{1}{4}\min\{\tfrac{1}{4GD}, \alpha\}\right)(x^t - x)^\top \nabla f_t(x^t)\nabla f_t(x^t)^\top (x^t - x). 
\end{eqnarray*}
Since $\|\nabla f_t(x)\| \leq G$ and $\XCal$ is convex and bounded with a diameter $D > 0$, we have
\begin{eqnarray*}
\lefteqn{\sum_{t=1}^T \left(\tfrac{1}{2\sqrt{1 + \max\{M^1, \ldots, M^t\}}} - \tfrac{1}{4}\min\{\tfrac{1}{4GD}, \alpha\}\right)(x^t - x)^\top \nabla f_t(x^t)\nabla f_t(x^t)^\top (x^t - x)} \\ 
& \leq & \tfrac{G^2D^2}{2}\left(\sum_{t = 1}^T \max\left\{0, \tfrac{1}{\sqrt{1 + \max\{M^1, \ldots, M^t\}}} - \tfrac{1}{2}\min\{\tfrac{1}{4GD}, \alpha\}\right\}\right). 
\end{eqnarray*}
Putting these pieces together yields that 
\begin{equation*}
\text{Regret}(T) \leq \tfrac{D^2}{2} + \tfrac{G^2D^2}{2}\left(\sum_{t = 1}^T \max\left\{0, \tfrac{1}{\sqrt{1 + \max\{M^1, \ldots, M^t\}}} - \tfrac{1}{2}\min\{\tfrac{1}{4GD}, \alpha\}\right\}\right) + \tfrac{d\sqrt{1 + \max\{M^1, \ldots, M^T\}}}{2}\log(TG^2 + 1). 
\end{equation*}
For simplicity, we let $\gamma = \frac{1}{2}\min\{\tfrac{1}{4GD}, \alpha\}$.  Taking the expectation of both sides, we have
\begin{equation*}
\EE[\text{Regret}(T)] \leq \tfrac{D^2}{2} + \tfrac{G^2 D^2}{2}\underbrace{\EE\left[\sum_{t = 1}^T \max\left\{0, \tfrac{1}{\sqrt{1 + \max\{M^1, \ldots, M^t\}}} - \gamma\right\}\right]}_{\textbf{I}} + \tfrac{d\log(TG^2+1)}{2}\underbrace{\EE\left[\sqrt{1 + \max\{M^1, \ldots, M^T\}}\right]}_{\textbf{II}}. 
\end{equation*}
It remains to bound the terms $\textbf{I}$ and $\textbf{II}$ using Proposition~\ref{Prop:GRV}. By using the same argument as applied in the proof of Theorem~\ref{Thm:AdaOGD-regret}, we have
\begin{equation*}
\textbf{I} \leq e^{\frac{1}{\gamma^2\log(T+10)}}, 
\end{equation*}
and 
\begin{equation*}
\textbf{II} \leq \sqrt{1 + \log(T+10) + \log(T)\log(T+10)}. 
\end{equation*}
Therefore, we conclude that 
\begin{equation*}
\EE[\text{Regret}(T)] \leq \tfrac{D^2}{2}(1 + e^{\frac{1}{\gamma^2\log(T+10)}}) + \tfrac{d\log(TG^2+1)}{2}\sqrt{1 + \log(T+10) + \log(T)\log(T+10)}. 
\end{equation*}
This completes the proof. 

\section{Proof of Lemma~\ref{Lemma:AdaONS-MA}}\label{app:AdaONS-MA}
Recall that the update formula of $x_i^{t+1}$ in either the multi-agent ONS is
\begin{equation*}
x_i^{t+1} \leftarrow \argmin_{x_i \in \XCal_i}\{(x_i - x_i^t)^\top v_i(x^t) + \tfrac{\eta_i}{2}(x_i - x_i^t)^\top A_i^{t+1}(x_i - x_i^t)\}. 
\end{equation*}
The first-order optimality condition implies that 
\begin{equation*}
(x_i - x_i^{t+1})^\top v_i(x^t) + \eta_i (x_i - x_i^{t+1})^\top A_i^{t+1}(x_i^{t+1} - x_i^t) \geq 0, \quad \textnormal{for all } x_i \in \XCal_i. 
\end{equation*}
Equivalently, we have
\begin{eqnarray*}
\lefteqn{ \tfrac{\eta_i}{2}((x_i^{t+1} - x_i)^\top A_i^{t+1}(x_i^{t+1} - x_i) - (x_i^t - x_i)^\top A_i^{t+1}(x_i^t - x_i))} \\ 
& \leq & (x_i - x_i^{t+1})^\top v_i(x^t) - \tfrac{\eta_i}{2}(x_i^{t+1} - x_i^t)^\top A_i^{t+1}(x_i^{t+1} - x_i^t) \\
& = & (x_i - x_i^t)^\top v_i(x^t) + (x_i^t - x_i^{t+1})^\top v_i(x^t) - \tfrac{\eta_i}{2}(x_i^{t+1} - x_i^t)^\top A_i^{t+1}(x_i^{t+1} - x_i^t) \\
& \leq & (x_i - x_i^t)^\top v_i(x^t) + \tfrac{1}{2\eta_i}v_i(x^t)^\top (A_i^{t+1})^{-1}v_i(x^t). 
\end{eqnarray*}
Rearranging this inequality, we have
\begin{eqnarray*}
\lefteqn{\tfrac{\eta_i}{2}(x_i^{t+1} - x_i)^\top A_i^{t+1}(x_i^{t+1} - x_i) - \tfrac{\eta_i}{2} (x_i^t - x_i)^\top A_i^t(x_i^t - x_i)} \\
& \leq & (x_i - x_i^t)^\top v_i(x^t) + (x_i^t - x_i)^\top\left(\tfrac{\eta_i}{2}A_i^{t+1} - \tfrac{\eta_i}{2}A_i^t\right)(x_i^t - x_i) + \tfrac{1}{2\eta_i}v_i(x^t)^\top (A_i^{t+1})^{-1}v_i(x^t). 
\end{eqnarray*}
Summing over $i = 1, 2, \ldots, N$, we have
\begin{eqnarray*}
\lefteqn{\sum_{i=1}^N \left(\tfrac{\eta_i}{2}(x_i^{t+1} - x_i)^\top A_i^{t+1}(x_i^{t+1} - x_i) - \tfrac{\eta_i}{2} (x_i^t - x_i)^\top A_i^t(x_i^t - x_i)\right)} \\ 
& \leq & (x - x^t)^\top v(x^t) + \tfrac{1}{2}\left(\sum_{i = 1}^N (x_i^t - x_i)^\top\left(\eta_i A_i^{t+1} - \eta_i A_i^t\right)(x_i^t - x_i)\right) + \tfrac{1}{2}\left(\sum_{i = 1}^N \tfrac{1}{\eta_i}v_i(x^t)^\top (A_i^{t+1})^{-1}v_i(x^t)\right). 
\end{eqnarray*}
Since $\GCal$ is $\alpha$-exp-concave and satisfies that $\|v(x)\| \leq G_i$ and $\|x - x'\| \leq D_i$ for all $x_i, x'_i \in \XCal_i$, we have
\begin{equation*}
\langle x - x^t, v(x) - v(x^t)\rangle \geq \tfrac{1}{4} \left(\sum_{i = 1}^N \min\{\tfrac{1}{4G_i D_i}, \alpha\}(x_i^t - x_i)^\top(v_i(x)v_i(x)^\top + v_i(x^t)v_i(x^t)^\top)(x_i^t - x_i)\right). 
\end{equation*}
Putting these pieces together yields that 
\begin{eqnarray*}
\lefteqn{\sum_{i=1}^N \left(\tfrac{\eta_i}{2}(x_i^{t+1} - x_i)^\top A_i^{t+1}(x_i^{t+1} - x_i) - \tfrac{\eta_i}{2} (x_i^t - x_i)^\top A_i^t(x_i^t - x_i)\right)} \\ 
& \leq & (x - x^t)^\top v(x) + \sum_{i=1}^N (x_i^t - x_i)^\top(\tfrac{\eta_i}{2}A_i^{t+1} - \tfrac{\eta_i}{2}A_i^t - \tfrac{1}{4}\min\{\tfrac{1}{4GD}, \alpha\} v_i(x^t)v_i(x^t)^\top)(x_i^t - x_i) \\ 
& & + \tfrac{1}{2}\left(\sum_{i = 1}^N \tfrac{1}{\eta_i}v_i(x^t)^\top (A_i^{t+1})^{-1}v_i(x^t)\right). 
\end{eqnarray*}
Summing over $t= 1, 2, \ldots, T$ yields the desired inequality. 

\section{Proof of Theorem~\ref{Thm:ONS-rate}}\label{app:ONS-rate}
We can see from the multi-agent ONS algorithm that $\eta_i = \frac{1}{2}\min\{\frac{1}{4GD}, \alpha\}$, $A_i^1 = I_{d_i}$ where $I_{d_i} \in \br^{d_i \times d_i}$ is an identity matrix, and $A_i^{t+1} = A_i^t + v_i(x^t)v_i(x^t)^\top$. Since $\XCal$ is convex and bounded with a diameter $D > 0$, we have
\begin{equation*}
\sum_{i = 1}^N \tfrac{\eta_i}{2}(x_i^1 - x_i)^\top A_i^1(x_i^1 - x_i) \leq \tfrac{\alpha D^2}{4}, \qquad \eta_i A_i^{t+1} - \eta_i A_i = \tfrac{1}{2}\min\{\tfrac{1}{4GD}, \alpha\} v_i(x^t)v_i(x^t)^\top. 
\end{equation*}
By Lemma~\ref{Lemma:AdaONS-MA}, we have
\begin{equation}\label{inequality:ONS-rate-first}
\sum_{t = 1}^T (x^t - x)^\top v(x) \leq \tfrac{\alpha D^2}{4} + \tfrac{1}{2}\left(\sum_{t = 1}^T \sum_{i = 1}^N \tfrac{1}{\eta_i}v_i(x^t)^\top (A_i^{t+1})^{-1}v_i(x^t)\right). 
\end{equation}
Since $\|v(x)\| \leq G$, $A_i^1 = I_{d_i}$ and $A_i^{t+1} = A_i^t + v_i(x^t)v_i(x^t)^\top$,~\citet[Lemma~11]{Hazan-2007-Logarithmic} guarantees that 
\begin{equation*}
\sum_{t=1}^T v_i(x^t)^\top (A_i^{t+1})^{-1}v_i(x^t) \leq d_i\log(TG^2+1), 
\end{equation*}
which implies that 
\begin{equation}\label{inequality:ONS-rate-second}
\sum_{t = 1}^T \sum_{i = 1}^N \tfrac{1}{\eta_i}v_i(x^t)^\top (A_i^{t+1})^{-1}v_i(x^t) \leq \max\left\{8GD, \tfrac{2}{\alpha}\right\}(d\log(TG^2+1)). 
\end{equation}
Plugging Eq.~\eqref{inequality:ONS-rate-second} into Eq.~\eqref{inequality:ONS-rate-first} yields that 
\begin{equation*}
\sum_{t = 1}^T (x^t - x)^\top v(x) \leq \tfrac{\alpha D^2}{4} + \max\left\{4GD, \tfrac{1}{\alpha}\right\}(d\log(TG^2+1)). 
\end{equation*}
By the definition of $\textsc{gap}(\cdot)$ and $\bar{x}^T$ (i.e., $\bar{x}^T = \frac{1}{T}\sum_{t=1}^T x^t$), we have
\begin{equation*}
\textsc{gap}(\bar{x}^T) \leq \tfrac{\alpha D^2}{4T} + \tfrac{d\log(TG^2+1)}{T}\max\left\{4GD, \tfrac{1}{\alpha}\right\}. 
\end{equation*}
This completes the proof. 

\section{Proof of Theorem~\ref{Thm:AdaONS-rate}}\label{app:AdaONS-rate}
We can see from Algorithm~\ref{alg:AdaONS-MA} that $\eta_i^{t+1} = \tfrac{1}{\sqrt{1 + \max\{M_i^1, \ldots, M_i^t\}}}$, $A_i^1 = I_{d_i}$ where $I_{d_i} \in \br^{d_i \times d_i}$ is an identity matrix, and $A_i^{t+1} = A_i^t + v_i(x^t)v_i(x^t)^\top$. Since $\XCal$ is convex and bounded with a diameter $D > 0$, we have
\begin{equation*}
\sum_{i = 1}^N \tfrac{\eta_i^1}{2}(x_i^1 - x_i)^\top A_i^1(x_i^1 - x_i) \leq \tfrac{D^2}{2}, \qquad \eta_i^{t+1}A_i^{t+1} - \eta_i^t A_i \preceq \tfrac{1}{\sqrt{1 + \max\{M_i^1, \ldots, M_i^t\}}}v_i(x^t)v_i(x^t)^\top. 
\end{equation*}
We can see from Eq.~\eqref{inequality:AdaONS-MA-key} that 
\begin{eqnarray}\label{inequality:AdaONS-rate-first}
\lefteqn{\sum_{t = 1}^T (x^t - x)^\top v(x) \leq \tfrac{D^2}{2} + \tfrac{1}{2}\left(\sum_{t = 1}^T \sum_{i = 1}^N \tfrac{1}{\eta_i^{t+1}}v_i(x^t)^\top (A_i^{t+1})^{-1}v_i(x^t)\right)} \\
& & + \sum_{t = 1}^T \sum_{i=1}^N \left(\tfrac{1}{2\sqrt{1 + \max\{M_i^1, \ldots, M_i^t\}}} - \tfrac{1}{4}\min\{\tfrac{1}{4GD}, \alpha\}\right) (x_i^t - x_i)^\top v_i(x^t)v_i(x^t)^\top(x_i^t - x_i). \nonumber
\end{eqnarray}
Since $\|v(x)\| \leq G$, $A_i^1 = I_{d_i}$ and $A_i^{t+1} = A_i^t + v_i(x^t)v_i(x^t)^\top$,~\citet[Lemma~11]{Hazan-2007-Logarithmic} guarantees that 
\begin{eqnarray*}
\sum_{t=1}^T \tfrac{1}{\eta_i^{t+1}} v_i(x^t)^\top (A_i^{t+1})^{-1}v_i(x^t) & \leq & \sqrt{1 + \max\{M_i^1, \ldots, M_i^T\}}\left(\sum_{t=1}^T v_i(x^t)^\top (A_i^{t+1})^{-1}v_i(x^t)\right) \\ 
& \leq & \sqrt{1 + \max\{M_i^1, \ldots, M_i^T\}}(d_i\log(TG^2+1)), 
\end{eqnarray*}
which implies that 
\begin{equation}\label{inequality:AdaONS-rate-second}
\sum_{t = 1}^T \sum_{i = 1}^N \tfrac{1}{\eta_i^{t+1}}v_i(x^t)^\top (A_i^{t+1})^{-1}v_i(x^t) \leq \sqrt{1 + \max_{1 \leq i \leq N, 1 \leq t \leq T} \{M_i^t\}}(d\log(TG^2+1)). 
\end{equation}
Plugging Eq.~\eqref{inequality:AdaONS-rate-second} into Eq.~\eqref{inequality:AdaONS-rate-first} yields that 
\begin{eqnarray*}
\lefteqn{\sum_{t = 1}^T (x^t - x)^\top v(x) \leq \tfrac{D^2}{2} + \tfrac{d\log(TG^2+1)}{2}\sqrt{1 + \max_{1 \leq i \leq N, 1 \leq t \leq T} \{M_i^t\}}} \\
& & + \sum_{t = 1}^T \sum_{i=1}^N \left(\tfrac{1}{2\sqrt{1 + \max\{M_i^1, \ldots, M_i^t\}}} - \tfrac{1}{4}\min\{\tfrac{1}{4GD}, \alpha\}\right) (x_i^t - x_i)^\top v_i(x^t)v_i(x^t)^\top(x_i^t - x_i). 
\end{eqnarray*}
Since $\|v(x)\| \leq G$ and $\XCal$ is convex and bounded with a diameter $D > 0$, we have
\begin{eqnarray*}
\lefteqn{\sum_{t=1}^T \sum_{i=1}^N \left(\tfrac{1}{2\sqrt{1 + \max\{M_i^1, \ldots, M_i^t\}}} - \tfrac{1}{4}\min\{\tfrac{1}{4GD}, \alpha\}\right)(x_i^t - x_i)^\top v_i(x^t)v_i(x^t)^\top (x_i^t - x_i)} \\ 
& \leq & \tfrac{G^2D^2}{2}\left(\sum_{t = 1}^T \max\left\{0, \max_{1 \leq i \leq N}\left\{\tfrac{1}{\sqrt{1 + \max\{M_i^1, \ldots, M_i^t\}}}\right\} - \tfrac{1}{2}\min\{\tfrac{1}{4GD}, \alpha\}\right\}\right). 
\end{eqnarray*}
Putting these pieces together yields that 
\begin{eqnarray*}
\lefteqn{\sum_{t = 1}^T (x^t - x)^\top v(x) \leq \tfrac{D^2}{2} + \tfrac{d\log(TG^2+1)}{2}\sqrt{1 + \max_{1 \leq i \leq N, 1 \leq t \leq T} \{M_i^t\}}} \\
& & + \tfrac{G^2D^2}{2}\left(\sum_{t = 1}^T \max\left\{0, \max_{1 \leq i \leq N}\left\{\tfrac{1}{\sqrt{1 + \max\{M_i^1, \ldots, M_i^t\}}}\right\} - \tfrac{1}{2}\min\{\tfrac{1}{4GD}, \alpha\}\right\}\right). \nonumber
\end{eqnarray*}
For simplicity, we let $\gamma = \frac{1}{2}\min\{\tfrac{1}{4GD}, \alpha\}$. By the definition of $\bar{x}^T$ (i.e., $\bar{x}^T = \frac{1}{T}\sum_{t=1}^T x^t$), we have
\begin{eqnarray*}
\textsc{gap}(\bar{x}^T) \leq \tfrac{D^2}{2T} + \tfrac{G^2D^2}{2T}\left(\sum_{t = 1}^T \max\left\{0, \max_{1 \leq i \leq N}\left\{\tfrac{1}{\sqrt{1 + \max\{M_i^1, \ldots, M_i^t\}}}\right\} - \gamma\right\}\right) + \tfrac{d\log(TG^2+1)}{2T}\sqrt{1 + \max_{1 \leq i \leq N, 1 \leq t \leq T} \{M_i^t\}}. 
\end{eqnarray*}
Taking the expectation of both sides, we have
\begin{eqnarray}\label{inequality:AdaONS-rate-third}
\EE[\textsc{gap}(\bar{x}^T)] & \leq & \tfrac{D^2}{2T} + \tfrac{G^2 D^2}{2T}\underbrace{\EE\left[\sum_{t = 1}^T \max\left\{0, \max_{1 \leq i \leq N}\left\{\tfrac{1}{\sqrt{1 + \max\{M_i^1, \ldots, M_i^t\}}}\right\} - \gamma\right\}\right]}_{\textbf{I}} \\ 
& & + \tfrac{d\log(TG^2+1)}{2T}\underbrace{\EE\left[\sqrt{1 + \max_{1 \leq i \leq N, 1 \leq t \leq T} \{M_i^t\}}\right]}_{\textbf{II}}.  \nonumber
\end{eqnarray}
It remains to bound the terms $\textbf{I}$ and $\textbf{II}$ using Proposition~\ref{Prop:GRV}.  Indeed, we have
\begin{equation*}
\textbf{I} = \sum_{t = 1}^T \EE\left[\max\left\{0, \max_{1 \leq i \leq N}\left\{\tfrac{1}{\sqrt{1 + \max\{M_i^1, \ldots, M_i^t\}}}\right\} - \gamma\right\}\right] \ \leq \ \sum_{t = 1}^T \PP\left(\max_{1 \leq i \leq N}\left\{\tfrac{1}{\sqrt{1 + \max\{M_i^1, \ldots, M_i^t\}}}\right\} - \gamma \geq 0\right). 
\end{equation*}
Considering $i^t = \argmax_{1 \leq i \leq N} \left\{\tfrac{1}{\sqrt{1 + \max\{M_i^1, \ldots, M_i^t\}}}\right\}$ that is a random variable and then recalling that $\{\max\{M_i^1, \ldots, M_i^t\}\}_{1 \leq i \leq N}$ are i.i.d., we have $i^t \in \{1, \ldots, N\}$ is uniformly distributed. This implies that 
\begin{equation*}
\textbf{I} \leq \tfrac{1}{N} \sum_{t = 1}^T \sum_{j=1}^N \PP\left(\max_{1 \leq i \leq N}\left\{\tfrac{1}{\sqrt{1 + \max\{M_i^1, \ldots, M_i^t\}}}\right\} - \gamma \geq 0 \mid i^t = j\right) = \tfrac{1}{N} \sum_{t = 1}^T \sum_{j=1}^N \PP\left(\tfrac{1}{\sqrt{1 + \max\{M_j^1, \ldots, M_j^t\}}} - \gamma \geq 0\right). 
\end{equation*}
Since $\{M_i^t\}_{1 \leq i \leq N, 1 \leq t \leq T}$ are i.i.d.\ geometric random variables with $p_0 = \frac{1}{\log(T+10)}$, Proposition~\ref{Prop:GRV} implies that
\begin{equation*}
\sum_{t=1}^T \PP\left(\tfrac{1}{\sqrt{1 + \max\{M_j^1, \ldots, M_j^t\}}} - \gamma \geq 0\right) \leq e^{\frac{p}{\gamma^2}} = e^{\frac{1}{\gamma^2\log(T+10)}}. 
\end{equation*}
Putting these pieces together yields that 
\begin{equation}\label{inequality:AdaONS-rate-I}
\textbf{I} \leq e^{\frac{1}{\gamma^2\log(T+10)}} = e^{\frac{(\max\{8GD, 2\alpha^{-1}\})^2}{\log(T+10)}}. 
\end{equation}
By using the similar argument with Proposition~\ref{Prop:GRV} and $p_0 = \frac{1}{\log(T+10)}$, we have
\begin{equation}\label{inequality:AdaONS-rate-II}
\textbf{II} \leq \sqrt{1 + \tfrac{1 + \log(NT)}{p_0}} = \sqrt{1 + \log(T+10) + \log(NT)\log(T+10)}. 
\end{equation}
Plugging Eq.~\eqref{inequality:AdaONS-rate-I} and Eq.~\eqref{inequality:AdaONS-rate-II} into Eq.~\eqref{inequality:AdaONS-rate-third} yields that 
\begin{equation*}
\EE[\textsc{gap}(\bar{x}^T)] \leq \tfrac{D^2}{2T}(1 + G^2 e^{\frac{(\max\{8GD, 2\alpha^{-1}\})^2}{\log(T+10)}}) + \tfrac{d\log(TG^2+1)}{2T}\sqrt{1 + \log(T+10) + \log(T)\log(T+10)}. 
\end{equation*}
This completes the proof.

\end{document}